\PassOptionsToPackage{unicode}{hyperref}
\PassOptionsToPackage{hyphens}{url}
\PassOptionsToPackage{table,dvipsnames,svgnames,x11names}{xcolor}
\documentclass[12pt]{article}
\usepackage{caption}

\usepackage[T1]{fontenc}
\usepackage[utf8]{inputenc}
\usepackage{textcomp}
\usepackage{lmodern}
\usepackage{amsmath,amssymb,amsfonts,amsthm,mathtools}
\IfFileExists{nccmath.sty}{\usepackage{nccmath}}{}
\IfFileExists{bbm.sty}{\usepackage{bbm}}{}

\usepackage{xcolor}
\usepackage{graphicx}
\usepackage{booktabs}
\IfFileExists{nicefrac.sty}{\usepackage{nicefrac}}{}

\usepackage{microtype}
\usepackage{url}
\usepackage{longtable,array}
\usepackage{multirow}
\usepackage{subcaption}
\usepackage{afterpage}
\usepackage[hang,flushmargin]{footmisc}
\usepackage{float}
\usepackage{enumitem}
\IfFileExists{ulem.sty}{\usepackage[normalem]{ulem}}{}

\usepackage{pdflscape}
\IfFileExists{placeins.sty}{\usepackage{placeins}}{}
\providecommand{\FloatBarrier}{}
\usepackage{algorithm}
\usepackage{algorithmic}
\usepackage[authoryear,round]{natbib}
\IfFileExists{xurl.sty}{\usepackage{xurl}}{}
\usepackage{hyperref}
\IfFileExists{cleveref.sty}{%
\usepackage[capitalize,noabbrev]{cleveref}
\crefname{paragraph}{Appendix}{Appendices}
\Crefname{paragraph}{Appendix}{Appendices}
\crefname{subparagraph}{Appendix}{Appendices}
\Crefname{subparagraph}{Appendix}{Appendices}
}{}
\providecommand{\cref}[1]{\autoref{#1}}
\providecommand{\Cref}[1]{\autoref{#1}}

\makeatletter
\def\maxwidth{\ifdim\Gin@nat@width>\linewidth\linewidth\else\Gin@nat@width\fi}
\def\maxheight{\ifdim\Gin@nat@height>\textheight\textheight\else\Gin@nat@height\fi}
\makeatother
\setkeys{Gin}{width=\maxwidth,height=\maxheight,keepaspectratio}

\theoremstyle{plain}
\newtheorem{theorem}{Theorem}
\newtheorem{proposition}[theorem]{Proposition}
\newtheorem{lemma}[theorem]{Lemma}

\theoremstyle{definition}
\newtheorem{definition}[theorem]{Definition}

\theoremstyle{remark}
\newtheorem{remark}[theorem]{Remark}

\newcommand{\N}{\mathbb{N}}

\DeclareMathOperator{\E}{\mathbb{E}}

\allowdisplaybreaks

\makeatletter
\newcommand{\printfnsymbol}[1]{\textsuperscript{\@fnsymbol{#1}}}
\makeatother

\renewcommand{\epsilon}{\varepsilon}

\newcommand{\src}{\mathrm{src}}

\newcommand{\Mult}{\mathrm{Mult}}
\newcommand{\RDC}{\mathrm{RDC}}

\hypersetup{
  pdftitle={Statistical attribute alignment for black-box generative AI via output post-processing},
  pdfauthor={Kevin Jiang; Morgane Austern; Edgar Dobriban; Jason M. Klusowski},
  colorlinks=true,
  linkcolor={black},
  filecolor={Maroon},
  citecolor={Black},
  urlcolor={Black}
}

\title{Statistical attribute alignment for black-box
 generative AI via output post-processing}
\author{Kevin Jiang \qquad Morgane Austern\\[0.5ex]
Edgar Dobriban \qquad Jason M. Klusowski}
\date{September 25, 2026}

\begin{document}
\maketitle

\begin{abstract}
Generative AI systems are increasingly used, but aligning their outputs with user requirements poses a continuing challenge. 
Here, we aim to ensure that the distribution of an attribute of an
AI-generated output aligns with a user-specified target. 
This is motivated by examples such as fairness, where we want to ensure that a protected attribute (e.g., gender, race, or age categories) follows a desired distribution, and synthetic data generation, where we want the generated data to be representative of a target distribution.
We study the practically important black-box access setting, where a user can repeatedly
query a generative AI model.
The goal is to return $m\ge 1$ 
outputs
whose joint attribute distribution is as close as possible to this target.
For both exact and 
approximate alignment, we develop algorithms 
that minimize the expected number of queries to the generator, and we further demonstrate their optimality as the number of requested outputs $m \rightarrow \infty$.
Experiments on text-to-image generation 
and geocoded persona generation tasks show that our post-processing algorithms 
improve statistical attribute alignment, complementing prompting-based interventions.
\end{abstract}

\section{Introduction}\label{sec:intro}

In recent years, 
generative artificial intelligence (genAI) 
models have seen increasing adoption 
across various domains such as text and image generation, code generation, and decision planning \citep{bommasani2021opportunities, brown2020language, ramesh2022hierarchical}. 
An important challenge is to align their outputs with the requirements of a particular application.
Alignment encompasses several objectives, including following instructions, reflecting human preferences, and satisfying safety or value constraints; see the recent surveys \cite{ji2025alignment,pan2025trainingfree}.
Influential approaches include reinforcement learning from human preferences \citep{christiano2017deep}, its application to instruction-following language models \citep{ouyang2022training}, constitutional AI based on principle-guided model feedback \citep{bai2022constitutional}, and direct preference optimization \citep{rafailov2023direct}.
 
However, despite great progress, there are a variety of challenges that remain in ensuring alignment.
 For instance, 
despite their impressive capabilities,
models are known to learn 
historical and representational 
biases present in their training data, 
leading to systematic disparities along protected attributes such as race, gender, and age 
\citep{bolukbasi2016man, buolamwini2018gender, bender2021dangers}. 
These biases 
can raise 
concerns, particularly 
where generative models are deployed in
public-facing settings \citep{barocas2023fairness, mehrabi2021survey}
or used to generate synthetic data for training downstream models and estimators \citep{wyllie2024fairness}.

In this work, we study one particular challenge within the broader alignment landscape, which we term 
\emph{statistical attribute alignment} (SAA): ensuring that the distribution of designated categorical 
attributes in generated outputs agrees with a user-specified target.
The attributes may be selected covariates or response labels, and the target may describe a desired population or sampling design.
Even a model aligned to instructions or preferences may generate attribute distributions that differ from those required by a particular user.

Fairness and synthetic data generation provide two motivating examples.
In a 
use case aiming to ensure fairness, 
the attributes are categorical covariates designated as protected, such as gender, race, or age categories, and the target specifies their desired representation \citep{barocas2023fairness}.
Matching this target addresses a distributional aspect of fairness.
As we will discuss later, the choice of protected attributes and the target is application-dependent.

Generative models also provide a potential source of data for statistical analysis, model training, and evaluation,
especially when collecting or annotating real observations is costly.
Synthetic data have a long history in statistics, including the release of synthetic microdata through multiple imputation \citep{rubin1993disclosure,raghunathan2003multiple}.
Modern generative models have expanded their use in data augmentation; see \citet{chen2024comprehensive,mumuni2024survey} for surveys and \citet{abdel2026harnessing} for a recent review.
AI-generated data may not match the desired distribution, in which case
 additional methods are required to correct such biases, see e.g.,  \cite{angelopoulos2023prediction,bashari2025synthetic,bashari2025general}, etc.
Our focus is on constructing the synthetic sample itself with a specified joint attribute distribution.
For example, a synthetic dataset for auditing a face recognition system should reflect the intended audited population.

A variety of methods for aligning AI models have been developed. 
These methods may intervene on the training data distribution, training and generation procedure, prompts, or outputs.
Their requirements vary with the available model access.
Methods with access to model internals can modify the architecture, training procedure, or generation dynamics
\citep[see for instance][]{xu2018fairgan, yu2020inclusive, choi2020fair, rajabi2022tabfairgan, teo2024fairtl,teo2024fairqueue, friedrich2025auditing},
while
\textit{black-box} and  training-free methods (see e.g.,  \citet{pan2025trainingfree}) 
only require that a user can sample from the model without requiring additional information such as model weights.
Black-box access
can be further divided based on the place of 
intervention, with
\textit{input-side} methods adjusting natural-language prompts
\citep{zhao2021ethical,bansal2022well,neumann2025position}
and 
\textit{output-side} or \textit{post-processing}
algorithms operating on the generated outputs, for example via filtering and runtime monitoring \citep{cano2021fairness, cheng2025runtime}. 

Here we focus on this black-box model access setting. 
This access mode is particularly important
because
many frontier models are distributed primarily through hosted web interfaces and APIs rather than as downloadable model weights.
They may be unavailable for local deployment, and serving models at this scale can require substantial computational resources.\footnote{White- and black-box methods can be combined and used in tandem \cite[see, e.g.,][]{schick2021self, cohen2025black, azam2025plug, gallegos2025self}.}
Within this setting, we select $m \ge 1$ generated outputs so that their joint attribute distribution aligns with a user-specified target.\footnote{We target the joint law of the returned attributes, rather than only their marginal proportions.
For a product target, this retains the sampling variation and dependence structure of independently sampled outputs.
This distinction is useful both for representation across repeated requests and for synthetic datasets used to compute statistics involving several observations, as illustrated by the pairwise audit task in \cref{sec:exp-pairwise-audit}.}
The attributes are categorical and may be covariates
(e.g., protected attributes such as gender or race)
or response labels
(e.g., whether a tumor is benign or malignant).
This approach also allows practitioners to adapt the attribute distribution of a model previously aligned to a different target, without retraining or fine-tuning the model.

Moreover, under black-box access, we do not assume knowledge of the \emph{source rates}, or rates at which the attributes are generated.
For example, when asking a text-to-image model to ``Generate an image of a doctor'',
one typically does not know the true rates at which an image of a male or female doctor is generated. 
These source rates may be different than what a practitioner or user wants (e.g., a uniform distribution over genders).

When the outputs of a black-box genAI method do not match well with the target distribution, we aim to post-process the outputs to bring them closer.
Since the source rates are unknown with black-box access, we will primarily be interested in developing \emph{source-rate-free}
post-processing methods which do not assume knowledge of these rates.
Our strategy will be to query the model multiple times and carefully sample a subset of outputs so that the resulting probability law of the returned attributes is close to a user-specified target distribution.
{\citet{block2023sample} consider a similar approach to obtaining
a single sample whose law approximates a target distribution---however, their rejection sampling 
construction requires access
to the likelihood ratio, which
depends on the unknown source rates.
Additional related work is discussed in \cref{app:arw}.
}

We begin by defining our strictest notion of statistical attribute alignment, called \emph{universal exactness}.
This 
requires the law of the attributes returned by a post-processing algorithm to exactly coincide with the target distribution.
We then show that a coupon collector type algorithm, called \emph{random demand coupon collector} (RDC), is source-rate-free, universally exact and minimizes the expected number of calls to the model among universally exact source-rate-free algorithms, up to first-order terms as the number of requested outputs $m \rightarrow \infty$. 

Universally exact algorithms often require many model calls, especially when there are attributes that occur rarely (see \cref{thm:universal-exact-m-cost-lower} for a formal statement).
Consequently, we relax universal exactness by allowing approximate statistical attribute alignment. 
In particular, we introduce
\emph{statistical alignment discrepancy} as a statistical distance between the attribute law of the
returned outputs and the target distribution.
We then show that a variant of the RDC algorithm, called \emph{thresholded anytime RDC (TA-RDC)}, is source-rate-free and can achieve any user-specified level of statistical alignment discrepancy. 
When the target distribution is a product law and discrepancy is measured by forward KL divergence, we additionally show 
TA-RDC is first-order optimal: 
among source-rate-free algorithms with the same control on the sampling discrepancy, it 
minimizes the number of expected model calls up to first-order terms as the number of requested outputs $m \rightarrow \infty$. 

Synthetic data augmentation also motivates the \emph{count-conditional} alignment criterion underlying our approximate algorithms.
A practitioner using synthetic data for augmentation often selects examples from a single generated dataset, whose attribute counts constrain the samples that can be returned.
The attribute law
underlying statistical alignment discrepancy averages over possible realizations of these counts, so good overall alignment can conceal poor alignment for individual datasets.
We therefore introduce a count-conditional alignment criterion that measures discrepancy given the observed attribute counts, and develop an algorithm that minimizes it for every realized count vector. 
This criterion both provides an interpretable objective for synthetic dataset curation and an upper bound on the statistical alignment discrepancy
(\cref{rem:count-cond-discrep-synth-data-aug,sec:count-cond-discrep-synth-data-aug}).

\Cref{fig:general-methodology-illustrative-figure} gives a high-level view of our black-box
post-processing
approach. 
The user specifies a desired distribution over attributes and
then queries a black-box generator to obtain a candidate batch of outputs. 
Since
a finite batch of samples
may not contain the target attributes in the desired
amounts, 
our procedures select from the available candidates so that the returned
outputs match the target distribution as closely as possible.

\begin{figure}[tbp]
    \centering
    \includegraphics[width=\linewidth]{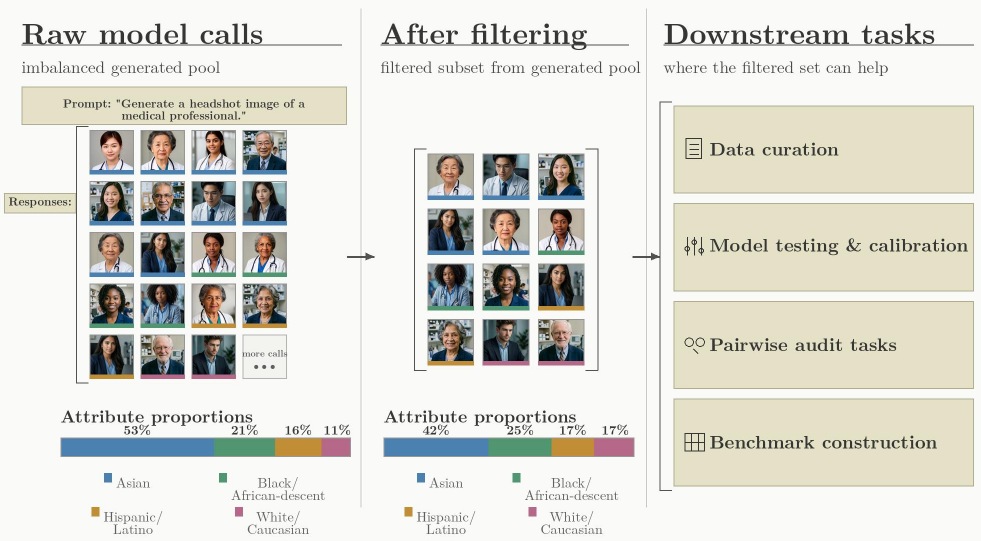}
    \caption{\footnotesize 
    High-level illustration of statistical attribute alignment via black-box output post-processing, using demographic attributes of generated headshots as an example. 
    A user 
    queries a black-box generator to form an annotated candidate batch.
    A black-box post-processing method then filters the
    candidate batch 
    by selecting 
    outputs that match a user-specified
    target distribution (e.g., uniform)
    as closely as possible, 
    subject to what was observed in
    the batch.}
    \label{fig:general-methodology-illustrative-figure}
\end{figure}

We illustrate our methods in a variety of generative AI tasks, ranging from text-to-image generation with multimodal AI models (\texttt{HiDream-O1-Image}, \texttt{Qwen-Image-2512}, \texttt{Flux-2-dev}) to text generation with a frontier language model  (GPT-5.5).
We argue that input-side interventions alone need not solve the problem.
For example, demographic hard prompts, which directly request a protected attribute, can fail, because the models do not follow instructions and the realized attribute may differ from the requested one. 
As post-processing algorithms, our methods, such as RDC and TA-RDC, 
can still substantially improve alignment with the target attribute distribution in this setting.
We further illustrate the use of synthetic datasets in a pairwise audit for face recognition systems. In this setting, matching a target joint law over attributes 
can yield more accurate estimates of downstream quantities, such as the false-match rate, than matching marginal proportions alone.

\section{Problem setup}\label{sec:problem-setup}

For 
a prompt $x\in\mathcal{X}$, 
a generative model $P$
induces a conditional distribution
$P(\cdot \mid x)$ over outputs $y \in \mathcal{Y}$.
 For instance, $x$ could be a textual prompt such as ``Generate an image of a doctor,''
 while $Y$ could be the generated image.
 We assume only black-box sampling access to the generative model (such as through the web or APIs), 
meaning that for any $x$, 
we can sample $Y\sim P(\cdot \mid x)$, but not directly access the internals of the model $P$.

An \textit{annotator} 
$\phi : \mathcal{Y} \to \mathcal{A}$ then maps outputs to a discrete 
set of attribute classes $\mathcal{A}$.
 These may be covariate or response categories in a synthetic data application.
 In the fairness setting, they correspond to protected attributes; for instance, in human images they may describe perceived race, gender, or age.
For notational convenience, identify $\mathcal{A}$ with $[k]=\{1,\dots,k\}$ after relabeling.

Our algorithms aim to align the joint distribution over attributes with a user-specified target for any given prompt $x$
by repeatedly querying the model. 
 When prompting the model and annotating the outputs,
 the attributes follow a certain distribution, which we refer to as the \textit{source rates}. 
 Before querying the model, these are unknown to the users. 
 They involve the randomness induced by sampling
 $Y\sim P(\cdot\mid x)$ from the model, and are defined as 
\begin{equation}\label{eq:src-rates}
    p_{i} \coloneqq p_{i}(x) \coloneqq \Pr_{Y\sim P(\cdot\mid x)}\bigl(\phi(Y)=i\bigr),
    \quad i\in\mathcal{A}.
\end{equation}
We collect them into the vector $p \coloneqq p(x) \coloneqq(p_{i})_{i\in\mathcal{A}}\in\Delta^{|\mathcal{A}|-1}$.
 These probabilities may differ from the target population or sampling design.
 In fairness applications, they might reflect undesirable historical biases learned by the model $P$.
 For instance, some races, genders, or ages might be unduly overrepresented in the generated images. 

Let the \textit{target joint law} $\nu \in \Delta(\mathcal{A}^m)$ be a probability law over sequences of attributes.
In the important special case where $\nu$ is a product law, let the
\textit{target rates} be
$q\coloneq(q_{i})_{i\in\mathcal{A}}\in\Delta^{|\mathcal{A}|-1}$ (e.g., the uniform distribution $q_{i} = 1/k, ~\forall i\in[k]$)
so that $\nu = q^{\otimes m}$.\footnote{The setting where
we are interested in aligning the distribution over a subset of attributes
$\mathcal{A}'\subseteq \mathcal{A}$ 
is  
subsumed by putting zero probability mass on the indices $\mathcal{A}\setminus \mathcal{A}'$ associated with $q$.
In particular, if the generative model has a nonzero chance of outputting an attribute that does not belong to $\mathcal{A}'$,
we can append a category named $\{\texttt{other}\}$ to $\mathcal{A}'$ so that $\mathcal{A}'\subset \mathcal{A}'\sqcup\{\texttt{other}\} =: \mathcal{A}$
and
put zero probability mass on its associated index in $q$.\label{ftnote:restricted-supp}} 
These probabilities are set by the user according to the application, for example to reflect a target population or a desired demographic composition. 
They could also depend on the input $x$, but we will typically only consider the simpler case when they do not.
A prototypical example might be to require that genders appear with equal probability when generating images of a particular occupation, such as doctors. 
In synthetic data generation, $q$ may instead encode the prevalence of response attributes or covariate strata in a specified population.
The specific way to set these probabilities is highly context-dependent \citep[see e.g.,][]{choi2020fair,de2023diffusionworldviewer}.

Our goal is to develop algorithms that apply regardless of the specific user choice of these probabilities, and we do not take a position about what specific choices are more or less desirable.
In addition, we are aiming
to preserve the distribution \emph{given the attributes}.
For example, when we output an image of a female physician, we want this image to follow the \emph{same} distribution as when it is generated from the original model (given the original output being a female physician).
This preserves within-attribute variation while changing the attribute distribution. 

The class of black-box post-processing algorithms we consider involves
two forms of sampling:
model calls, i.e., 
sampling from the conditional distribution $P(\cdot \mid x)$, 
and 
filtering among the generated outputs, involving additional randomness.
 Sampling from the model is usually more expensive, because it may require paying a fee to the model provider. 
 To clearly distinguish these two forms of sampling, we introduce the following terminology. 

\begin{definition}[M- and E-sampling]\label{def:M-E-sampling}
    For fixed $x$, 
    a model sample or \emph{M-sample} is a 
    query to the generative model which produces an output $Y\sim P(\cdot\mid x)$, 
    and hence an induced attribute $A=\phi(Y)$.
    An exogenous sample or \emph{E-sample} is any randomization that does not require querying $P$ and is independent of future M-samples, conditional on the current algorithmic history; e.g., drawing a uniform random variable $U\sim\mathrm{Unif}[0,1]$ or a random index $I\in\mathcal{A}$ as $I\sim Q$.
\end{definition}

A black-box post-processing algorithm may interweave M-sampling and E-sampling.
More formally, we 
consider the following class of algorithms. 

\begin{definition}[Black-box post-processing algorithms]\label{def:black-box-post-processing-algs}
Fix the number $m \ge 1$ of outputs  to produce.
\textit{Black-box post-processing algorithms} $\textrm{Alg}_m :\mathcal{P} \times \mathcal{X} \rightarrow \mathcal{Y}^m$---whose totality we denote by $\mathfrak{A}_m$---  take as input a generative model $P$ accessible 
through sampling.
For any input query $x$, 
such an algorithm sequentially queries the generative model to observe outcomes \(Y_1,\dots,Y_N \overset{\text{i.i.d.}}{\sim} P(\cdot \mid x)\) until a finite random number $N$
of M-samples has been obtained, where $N \geq m$ a.s. and $N$ is a stopping time with respect to the generated samples and any pre-output E-sampling. 
This produces annotated attributes 
\(A_t \coloneqq \phi(Y_t)\in\mathcal{A},~1\le t\le N \) with law 
\( A_1 \ldots A_N \overset{\text{i.i.d.}}{\sim} \textrm{Categorical}(p)\)\footnote{We discuss the setting with a noisy annotator in \cref{sec:discussion}.}.
 The algorithm may then perform arbitrary computations, including 
any number of E-samples, to  
 return a sample $(\widetilde Y_1,\dots,\widetilde Y_m)$ of size $m$ chosen without replacement from the observed M-samples, with associated annotated attributes $\widetilde A_j \coloneqq \phi(\widetilde Y_j), 1\le j\le m$.
 
 If the algorithm $\textrm{Alg}_m$ 
 does not use knowledge about the source rates $p$, then we further call it a \emph{source-rate-free} black-box post-processing algorithm, and denote the totality of such algorithms as $\mathfrak{A}_{m}^o \subset \mathfrak{A}_m$. 
\end{definition}

The post-processing algorithm $\textrm{Alg}_m \in \mathfrak{A}_m$ 
induces a distribution over the annotations of its final outputs
    \begin{align}\label{eq:output-label-law}
         P_{\textrm{Alg}_m}(\widetilde A = \widetilde a \mid x)
         \coloneqq
         \Pr(\widetilde A_1=a_1,\ldots,\widetilde A_m=a_m\mid x)
\end{align}
for $(a_1,\ldots,a_m)\in\mathcal A^m$. 
We call this distribution the \emph{output attribute law} of $\textrm{Alg}_m$.
The conditional probability integrates over every algorithmic random variable not determined by the input prompt $x$,
including any M-sampling to obtain outcomes $Y_1,\ldots,Y_N$
as well as any E-sampling performed before and after the outcomes are realized.

In practice,
M-sampling acts as the computational bottleneck.  
We thus measure the efficiency of these post-processing algorithms through the expected number of calls to the black-box model, which we call the \emph{M-cost}. 
Although the source rates are assumed to be unknown, the M-cost itself may depend on the source and target rates.
\begin{definition}[M-cost]\label{def:m-cost}
    The (source-rate pointwise) \emph{M-cost} of $\textrm{Alg}_m \in \mathfrak{A}_m$ is the expected number of calls to the model, defined as
    \(
        \mathcal{C}_m^{\textrm{Alg}_m}(p) \equiv \mathcal{C}_m^{\textrm{Alg}_m} (p,q) \coloneqq \mathbb{E}[N]
    \) computed over all problem randomness.
\end{definition}

{\bf Goal and target joint law.}
Since we want to handle black-box generative models which can only be accessed through sampling,
we do \emph{not} assume the source rates $p$ are known. 
Thus,
our primary goal is to find source-rate-free black-box post-processing algorithms $\textrm{Alg}_m\in\mathfrak{A}_m^o$
whose output attribute law (cf.~\cref{eq:output-label-law}) either approximates or is exactly
a user-specified target joint law $\nu$ over the space of attribute sequences $\mathcal{A}^m = [k]^m$.
For example, this may be
the law of an i.i.d.~sample of size $m$ from $q$, so that
$\nu = q^{\otimes m}$.

Among these algorithms that attain the target law exactly or 
for a specified level of approximation 
(we later specify how to quantify approximation error in \cref{sec:approx-sampling-fairness}), 
we further seek to characterize those that minimize the M-cost.

\section{Exact statistical attribute alignment}
\label{sec:exact-sampling}

We begin by defining the 
class of black-box post-processing algorithms that attain the target law exactly for arbitrary source rates $p \in \Delta^{k-1}$. 
Throughout this section, assume the support condition $\operatorname{supp}(\nu)\subseteq(\operatorname{supp}(p))^m$ unless otherwise explicitly stated.

\begin{definition}[Universal exactness]
\label{def:universal-exactness}
A post-processing algorithm $\textrm{Alg}_m \in \mathfrak{A}_m$
is 
\emph{universally exact} for the target joint law if its output attribute law 
$P_{\textrm{Alg}_m}(\widetilde{A} = \widetilde a \mid x)
=\nu(\widetilde a)$ (cf.~\cref{eq:output-label-law})
for every source rate $p$ satisfying $\operatorname{supp}(\nu)\subseteq(\operatorname{supp}(p))^m$.
Let $\mathfrak A_m(\nu) \subset \mathfrak{A}_m$ 
denote the class of universally exact algorithms with finite M-cost, and 
$\mathfrak A_m^o(\nu) \subset \mathfrak A_m(\nu)$ the subclass of those algorithms that are source-rate-free.
\end{definition}

For $t\ge0$, define the \emph{source count process} $\vec C(t)$
and its \emph{feasible support} $\Omega_t$ by
\begin{equation}\label{eq:feasible-support}
    C_i(t)=\sum_{s=1}^t1\{A_s=i\},\qquad
    \Omega_t=\Bigl\{\vec a\in[k]^m:
    \#\{j\in[m]:a_j=i\}\le C_i(t),\forall i\in[k]\Bigr\}.
\end{equation}
A black-box post-processing algorithm selects $m$ outputs from the first
$t\ge m$ queries and must thus return an attribute sequence in
$\Omega_t\equiv\Omega_m(\vec C(t))$.

The random demand coupon collector (RDC) algorithm
is a source-rate-free algorithm that 
first samples a target attribute sequence from the target joint law and then M-samples until the target sequence enters the feasible support. In other words, it repeatedly M-samples until the stopping time
\( T_{\nu}^\star
    \coloneq \inf\{t\ge m: \vec A^\star \in \Omega_t\} \); see \cref{alg:rdc}. 

\begin{algorithm}[tbp]
\caption{\footnotesize \textsc{Random Demand Coupon Collector} (RDC)}
\label{alg:rdc}
\footnotesize
\begin{algorithmic}[1]
\REQUIRE Output size $m$, target joint law $\nu$ on $[k]^m$,
annotator $\phi$.
\ENSURE Returned outputs $(\widetilde Y_1,\ldots,\widetilde Y_m)$.

\STATE E-sample $\vec A^\star\sim \nu$.
\STATE Query M-samples and construct the feasible support $\Omega_t$ as in \cref{eq:feasible-support}
until the stopping time
\[
    T_{\nu}^\star
    \coloneq \inf\{t\ge m: \vec A^\star \in \Omega_t\}.
\]
\STATE \textbf{return} outputs
$(\widetilde Y_1,\ldots,\widetilde Y_m)$
with attribute sequence $\vec A^\star$, selected uniformly without
replacement within each attribute from $(Y_1,\ldots,Y_{T_{\nu}^\star})$.
\end{algorithmic}
\end{algorithm}

\begin{theorem}[Exact source-rate-free unconditional sampling]
\label{thm:rdc-exact-main}
RDC is a source-rate-free black-box post-processing algorithm that terminates almost surely and is universally exact.
\end{theorem}
The proof is given in \cref{prf:rdc-exact-main}. RDC is source-rate-free as it does not use knowledge of the source rates $p$, and it is 
universally exact by construction of its random stopping time $T_\nu^\star$.

Denote 
the pointwise M-cost of the RDC algorithm
$\mathcal C_m^{\RDC}(p,\nu)=\mathbb E_p[T_{\nu}^\star]$
and the pointwise optimal universal cost for source-rate-free algorithms as
$   \mathcal C_m^{\mathrm{univ}}(p,\nu)
    =\inf_{\mathrm{Alg}_m\in\mathfrak A_m^o(\nu)}\mathbb E_p[N].
$
By definition, 
$ 
    C_m^{\mathrm{univ}}(p,\nu)
    \le
\mathcal C_m^{\RDC}(p,\nu)
$.
In the setting with a single returned output $m=1$, it turns out that RDC is optimal. 
\begin{theorem}[Exact optimality for one returned output]
\label{thm:rdc-m1-optimal}
For $m=1$, 
\(\mathcal C_1^{\mathrm{univ}}(p,\nu)
    =\mathcal C_1^{\RDC}(p,\nu)
    =\sum_{i\in \textrm{supp}(\nu)}{\nu_i}/{p_i}\).
Thus, for $m=1$, RDC is pointwise optimal 
at every source rate satisfying $\operatorname{supp}(\nu)\subseteq(\operatorname{supp}(p))^m$.
\end{theorem}

For general $m\ge1$, RDC is, in general, not optimal, but satisfies the following bound instead.
For $j\in[m]$ and $i\in[k]$, let \(\nu_j(i)\coloneqq \nu\{\vec a\in[k]^m:a_j=i\}\)
be the probability that the $j$th coordinate of an attribute sequence $\vec A \sim \nu$
belongs to attribute $i$.
\begin{theorem}[Universal source-rate-free lower bound and RDC upper bound]
\label{thm:rdc-witness-and-upper}
We have
\begin{equation}\label{eq:rdc-main-sandwich}
    \mathcal C_m^{\RDC}(p,\nu)
    \le
    \mathcal{C}_m^{\mathrm{univ}}(p,\nu)
    +\left(
        \sum_{i\in \operatorname{supp}(p)} \left( \sum_{j=1}^m \nu_j(i) \right) \cdot  \frac{1-p_i}{p_i^2}
    \right)^{1/2}.
\end{equation}
\end{theorem}
The proofs of \cref{thm:rdc-m1-optimal,thm:rdc-witness-and-upper} center around a version of Wald's identity 
that conditions on the returned attribute sequence of a universally exact algorithm; see \cref{lem:conditional-wald-main} in \cref{prf:rdc-m1-optimal}.
Since \(\sum_{i\in \operatorname{supp}(p)} \left( \sum_{j=1}^m \nu_j(i) \right) = m\), the additive term in
\cref{eq:rdc-main-sandwich} is $O(\sqrt{m})$. 
Every post-processing algorithm requires at least $m$ model queries,
so dividing both sides of the 
inequality by $m$ implies \({\mathcal C_m^{\mathrm{univ}}(p,\nu)}/
    {\mathcal C_m^{\RDC}(p,\nu)} \to 1\) as $m\to\infty$.
Thus, RDC is first-order optimal in the sense of minimizing the M-cost among all source-rate-free universally exact algorithms as the number of returned outputs $m \rightarrow \infty$.

In \cref{app:unconditional-count-tree-lp}, 
we show 
that minimizing the M-cost 
over all universally exact source-rate-free
algorithms
is equivalent to an infinite-dimensional linear program.
This program is difficult to solve for a number of reasons.
For universal exactness, the stopping time of a post-processing algorithm can not have a deterministic upper bound $N \le n$ for some $n < \infty$ 
(cf.~\cref{thm:no-capped-budget-uni-exactness}). 
Consequently, the stopping time has unbounded support, so the primal program has infinitely many variables, with their number growing combinatorially with the number of M-samples.
Furthermore, standard constrained Markov decision process methods can reduce programs with a finite number of constraints to finite occupation-measure programs  \citep[see, e.g.,][]{altman2021constrained}. However, 
universal exactness
imposes a continuum of globally coupled constraints, and thus these methods do not directly apply.

The RDC algorithm, with its first-order optimality guarantees, 
serves as an intuitive yet appealing alternative to solving the universal exactness problem exactly.
Nonetheless, the stopping time of the RDC algorithm, and any universally exact algorithm may be large when one or more source rates are rare.

\begin{theorem}[Universal expected M-cost lower bound]
\label{thm:universal-exact-m-cost-lower}
Fix $m\ge1$ 
and assume the support condition $\textrm{supp}(\nu) \subseteq(\textrm{supp}(p))^m$. 
Then,
every
universally exact 
algorithm
$\textrm{Alg}_m\in\mathfrak A_m(\nu)$ satisfies
\[
    \mathcal C_m^{\textrm{Alg}_m}(p,q)
    \ge
    \max_{i \in \operatorname{supp}(p)}\frac{\sum_{j=1}^m \nu_j(i)}{p_i}.
\]
\end{theorem}
For the proof, see \cref{prf:universal-exact-m-cost-lower}.
In particular,
no universally exact source-rate-free procedure has M-cost
uniformly bounded over all interior source rates $p \in \textrm{int}(\Delta^{k-1})$.
In order to obtain
controlled query complexity, we therefore also consider approximate target
sampling.

Finally, when the support condition \(\operatorname{supp}(\nu)\subseteq(\operatorname{supp}(p))^m\) does not hold, 
universal exactness is impossible, as
we can never M-sample an attribute sequence in 
$\operatorname{supp}(\nu) \backslash (\operatorname{supp}(p))^m$. 
These observations motivate us to consider approximating rather than exactly attaining the target joint law.

\section{Approximate statistical attribute alignment}\label{sec:approx-sampling-fairness}

Following the discussion at the end of \cref{sec:exact-sampling}, we now consider approximate sampling that allows deviations between the output attribute law and the target joint law.
More formally,
we will measure the quality of approximation via 
$f$-divergences.
\begin{definition}[$f$-divergences]\label{def:f-div}
    Let $f:[0,\infty)\to\mathbb{R}$ be a convex function with $f(1)=0$.
    For two probability distributions $P$ and $Q$ on a finite or countable space
    $\mathcal X$, the forward $f$-divergence from $P$ to $Q$ is
    \(
        D_f(P\|Q)=\sum_{x:Q(x)>0} Q(x) f\left(\frac{P(x)}{Q(x)}\right)
        + f_\infty P\{x:Q(x)=0\},
    \)
    where $f_\infty\coloneqq\lim_{t\to\infty}f(t)/t$ and the expression is
    interpreted in the usual extended-real sense.\footnote{In the absolutely continuous
    case $P\ll Q$, this reduces to
    \(
        D_f(P\|Q)
        =\mathbb E_Q\left[f\left(\frac{P}{Q}\right)\right]
        =\sum_{x\in\mathcal X}Q(x)f\left(\frac{P(x)}{Q(x)}\right).
    \)}
    If $f$ is strictly convex on $[0,\infty)$, then the $f$-divergence is called
    a \textit{strictly convex} $f$-divergence. If $D_f(P\|Q)=0$ holds if and only
    if $P=Q$, then the $f$-divergence is said to be \textit{separable}.
\end{definition}
 
We use the conventions
\(D_{\mathrm{KL}}(P\|Q)=\sum_x P(x)\log(P(x)/Q(x))\)
and
\(\mathrm{TV}(P,Q)=\tfrac12\sum_x|P(x)-Q(x)|\).
Forward KL divergence
is strictly convex and separable, whereas TV is convex and separable but not strictly convex. 

\begin{definition}[Unconditional statistical alignment discrepancy]
\label{def:uncond-sampling-fairness-discrepancy}
For fixed source rates $p\in\Delta^{k-1}$, 
the \emph{unconditional $f$-statistical alignment discrepancy} of $\textrm{Alg}_m$ relative to $\nu$ is
\begin{equation}\label{eq:unconditional-joint law-main}
    \Delta_f^{\mathrm{un}}(p) \equiv \Delta_f^{\mathrm{un}}(\textrm{Alg}_m;\nu;p)
    \coloneqq
    D_f\bigl(P_{\textrm{Alg}_m}(\widetilde{A}\mid x)\|\nu\bigr)
\end{equation}
where $P_{\textrm{Alg}_m}(\widetilde{A}\mid x)$ is the output attribute law in \cref{eq:output-label-law} that may depend on the source rates $p$ and target joint law $\nu$.
\end{definition}

The algorithm has $\varepsilon$-approximate unconditional statistical alignment discrepancy at $p$ 
if $\Delta_f^{\mathrm{un}}(p)\le\varepsilon$, 
and perfect unconditional discrepancy if $\Delta_f^{\mathrm{un}}(p)=0$. 
For a separable $f$-divergence, perfect unconditional discrepancy is equivalent to $\mathcal \Pr(\widetilde{ A}\mid x)=\nu$;
if this holds for every source
rate satisfying $\operatorname{supp}(\nu)\subseteq(\operatorname{supp}(p))^m$,
the algorithm is universally exact (cf.~\cref{def:universal-exactness}).

\begin{definition}[Universal $\varepsilon$-approximate sampling]
\label{def:universal-approximate-sampling}
Let $D_f$ be a forward $f$-divergence and let $\varepsilon\ge0$.
A post-processing algorithm $\textrm{Alg}_m\in\mathfrak{A}_m$ is
\emph{universally $(f,\varepsilon)$-approximate} for the target law $\nu$
if its output attribute law satisfies
\begin{equation}\label{eq:universal-approximate-sampling}
    D_f\!\left(
      P_{\textrm{Alg}_m}(\widetilde A\mid x)\middle\|\nu
    \right)
    =\Delta_f^{\mathrm{un}}(\textrm{Alg}_m;\nu;p)
    \le\varepsilon
\end{equation}
for every $p\in\operatorname{int}(\Delta^{k-1})$.  If additionally
$\textrm{Alg}_m\in\mathfrak{A}_m^o$, it is a
\emph{universally $(f,\varepsilon)$-approximate source-rate-free algorithm}.
\end{definition}

For a separable $f$-divergence, requiring that an algorithm is universally $(f,\varepsilon)$-approximate for
$\varepsilon=0$ is equivalent to 
requiring its universal exactness on the interior source domain.
Within the class of algorithms satisfying universal approximate alignment, a natural objective is to determine the minimum achievable M-cost
and to characterize its minimizers.

\begin{definition}[Universal source-rate-free M-cost]\label{def:uni-source-rate-free-M-cost}
    For fixed source rates $p \in\operatorname{int}(\Delta^{k-1})$,
    define the \emph{universal source-rate-free M-cost} as
    \begin{align}\label{eq:uni-source-rate-free-M-cost}
        \mathcal{C}_{m,\varepsilon}^{\mathrm{univ},f}(p;\nu)
    \coloneqq
    \inf_{\substack{\textrm{Alg}_m\in\mathfrak{A}_m^o:\\
    \Delta_f^{\mathrm{un}}(\textrm{Alg}_m;\nu;p')\le\varepsilon
    \quad\forall p'\in\operatorname{int}(\Delta^{k-1})}}
    \mathcal{C}_m^{\textrm{Alg}_m}(p).
    \end{align}
\end{definition}

Analogous to universal exactness,
the minimization in \cref{eq:uni-source-rate-free-M-cost} can be formulated as an infinite dimensional program.
This optimization is difficult to solve for reasons similar to those in the exact alignment setting (cf.~\cref{app:unconditional-count-tree-lp}), with the additional complication that
the approximate sampling constraints need not be linear for general \(f\)-divergences.

Furthermore, natural attempts to bypass this optimization using tractable
heuristic procedures --such as approximate rejection sampling based on
estimated source rates 
--
appear
to yield subpar performance; 
see
\cref{app:alg-approx-sampling-uncond-discrep}.
Therefore, we consider the following count-conditional proxy as an upper bound for the unconditional discrepancy.

\subsection{Count-conditional sampling}
\label{sec:count-conditional-target-sampling}
For a realized  
batch $\mathcal{B}_N \coloneq (Y_1 ,\ldots, Y_N)$ of $N$ observed outcomes \(Y_1,\dots,Y_N \overset{\text{i.i.d.}}{\sim} P(\cdot \mid x)\), 
the post-processing algorithm $\textrm{Alg}_m$ induces a distribution over the annotations of its chosen outputs
    \begin{align*}%
         P_{\textrm{Alg}_m}^{\mathcal{B}_N}(\widetilde A = \widetilde a)
         \coloneqq
        \Pr(\widetilde A_1=a_1,\ldots,\widetilde A_m=a_m\mid x,\mathcal B_N)
    \end{align*}
for $(a_1,\ldots,a_m)\in\mathcal A^m$.
In addition to conditioning on the input prompt $x$, 
this probability further conditions on the realized batch $\mathcal{B}_N$ and
integrates over any E-sampling performed before and after the batch is realized.

Now, for the random batch $\mathcal{B}_N$, define the empirical attribute count vector 
\(\vec C=(C_1,\ldots,C_k)\)
with coordinates \(C_i\coloneqq\sum_{t=1}^N1\{A_t=i\}, i\in[k].\)
Then, for every realized count vector $\vec c$ with positive probability under the source rates, 
let
\begin{equation}\label{eq:count-cond-output-label-law}
    P^c_{\textrm{Alg}_m}(\widetilde A = \widetilde a)
    \coloneqq
    \Pr(\widetilde A_1=a_1,\ldots,\widetilde A_m=a_m \mid x,\vec C = \vec c)
\end{equation}
be the algorithm's output attribute law after conditioning on the realized attribute counts.
Analogous to the output attribute law of \cref{eq:output-label-law}, we call this distribution the \emph{count-conditional output attribute law} of $\textrm{Alg}_m$.

For a fixed number of model queries $N=n$, 
the count vector $\vec C$ is minimal sufficient for $p$ under the i.i.d. sampling assumption $A_1,\ldots,A_n \overset{\text{i.i.d.}}{\sim} \textrm{Categorical}(p)$ (cf.~\cref{def:black-box-post-processing-algs}). 
Hence, when it comes to inference about the source rates, there is no loss of information when conditioning on the counts; for completeness, we include a formal theorem and proof of this statement in \cref{prf:count-minimal-sufficient}.
Therefore, it is reasonable to condition on the count vector $\vec C$ and aim to construct the best possible approximation of the target distribution given this vector of counts. 

\begin{definition}[Count-conditional statistical alignment discrepancy]
\label{def:count-cond-sampling-fairness-discrepancy}
For fixed source rates $p\in\Delta^{k-1}$, 
the \emph{count-conditional $f$-statistical alignment discrepancy} of $\textrm{Alg}_m$ relative to $\nu$ is
\begin{equation}\label{eq:count-conditional-sampling-distortion}
    \Delta_f^{\vec C}(p)\equiv 
    \Delta_f^{\vec C}(\textrm{Alg}_m;\nu;p)
    \coloneqq
    \mathbb E_{\vec C}\left[
        D_f\left(P_{\textrm{Alg}_m}^{\vec C}\|\nu\right)
    \right].
\end{equation}
\end{definition}
The algorithm has $\varepsilon$-approximate count-conditional statistical alignment discrepancy at $p$ if $\Delta_f^C(p)\le\varepsilon$, and perfect count-conditional target sampling if $\Delta_f^C(p)=0$. 

\begin{proposition}[Relation between the unconditional and count-conditional criteria]
\label{prop:criterion-decomposition}
Fix source rates $p\in\Delta^{k-1}$.
Then, for every forward $f$-divergence,
\begin{equation}
\label{eq:criterion-order}
    \Delta_f^{\mathrm{un}}(\textrm{Alg}_m;\nu;p)
    \leq
    \Delta_f^{\vec C}(\textrm{Alg}_m;\nu;p).
\end{equation}
\end{proposition}
This follows directly from Jensen's inequality and the convexity of $f$-divergences.
Our strategy will be to 
find algorithms 
that directly
minimize the count-conditional objective.
By the above result, these algorithms will yield universally approximate methods under
the unconditional discrepancy
and thus yield upper bounds on the universal
source-rate-free M-cost (cf.~\cref{def:uni-source-rate-free-M-cost}).

\begin{remark}[Count-conditional discrepancy and synthetic data augmentation]\label{rem:count-cond-discrep-synth-data-aug}
As discussed in the introduction, 
practitioners often select synthetic examples from a single generated dataset, whose observed attribute counts constrain the samples that can be returned.
The statistical alignment discrepancy in \cref{def:uncond-sampling-fairness-discrepancy} averages the output attribute law over the randomness of the count vectors, and then measures its distance from the target joint law.
The count-conditional discrepancy instead measures this distance separately for each count vector before averaging.

Controlling the count-conditional objective also yields provable guarantees for analyses involving the augmented data.
Under a shared attribute-conditional model (see, e.g., the label shift assumption of \citep{lipton2018detecting}), 
the count-conditional discrepancy under TV controls the difference in bounded
downstream criteria
between augmentation with the selected synthetic examples and augmentation with ideal target samples.
This is formalized in \cref{sec:count-cond-discrep-synth-data-aug}.
\end{remark}

\subsection{Anytime RDC}\label{sec:anytime-rdc}

To obtain universally approximate algorithms, 
we consider two different stopping time rules that truncate the RDC stopping time, either by stopping when a specified level of precision
or a capped budget 
is reached.
For intermediate times up to and at these truncated RDC stopping times, we use the following 
exponential clock race construction to define 
post-processing algorithms that minimize the count-conditional discrepancy \citep[see, e.g.,][and \cref{app:arw} for further discussion]{gillespie1976general}.
Since model queries are performed sequentially, it is useful to  define the newly feasible support \(\mathcal D_t \coloneqq\Omega_t\setminus\Omega_{t-1}\) at time $t \ge m$ 
with the convention \( \Omega_{m-1}\coloneqq\varnothing\).

\begin{definition}[Anytime random demand coupon collector (A-RDC)]
\label{def:anytime-rdc}
Fix a target joint law \(\nu\) on \([k]^m\) and
E-sample an attribute sequence $\vec A^\star \sim \nu$.
Let the RDC completion time be
\(T_\nu^\star
\coloneqq
\inf\{t\geq m:\vec A^\star\in\Omega_t\}\) (cf.~\cref{alg:rdc}).
Now, for each \(m\leq t<T_\nu^\star\) such that 
\(\nu(\mathcal D_t)>0\), 
independently sample an exponential clock time
\( G_t\sim\operatorname{Exp}\bigl(\nu(\mathcal D_t)\bigr)\)
and attribute sequence
\(
\vec U_t\sim\nu(\cdot\mid\mathcal D_t)\).
If \(t<T_\nu^\star\) and \(\nu(\Omega_t)>0\), define the current clock winner as
\begin{equation}
\label{eq:anytime-rdc-snapshot}
\widehat{A}_t
\coloneqq
\vec U_{\tau_t},
\qquad
\tau_t \coloneqq \underset{\substack{m \leq s \leq t :~ \nu(\mathcal{D}_s) > 0}}{\arg\min} G_s;
\end{equation}
otherwise, 
if \(\nu(\Omega_t)=0\), choose any sequence in the nonempty feasible set \(\Omega_t\). 
At time \(T_\nu^\star\), stop querying and return \(\vec A^\star\).
\end{definition}

\cref{alg:anytime-rdc} gives pseudo-code for the algorithm and 
\cref{fig:a-rdc-illustration} illustrates its evolution.
A-RDC generalizes RDC by returning feasible attribute
sequences before the completion time $T_\nu^\star$.
Under the support condition \(\operatorname{supp}(\nu)\subseteq(\operatorname{supp}(p))^m\) 
of \cref{thm:rdc-exact-main}, this
time is almost surely finite, and both algorithms achieve zero unconditional
discrepancy.

We now show that A-RDC traces an optimal frontier 
under the count-conditional objective at all times $t\ge m$. 
Let $\mathfrak{A}_{m,t} \subset \mathfrak{A}_m$ denote the class of post-processing algorithms whose stopping time is deterministically upper bounded by $N \le t$ for some $t \ge m$. 
Define, for any convex generator $f$ with $f(1)=0$ and the usual finite-space extended conventions,
\begin{equation}\label{eq:psi-f}
    \Psi_f(\alpha)
    \coloneqq
    \begin{cases}
        \alpha f(1/\alpha)+(1-\alpha)f(0), & \alpha\in(0,1],\\
        \displaystyle\lim_{\beta\downarrow 0}
        \bigl\{\beta f(1/\beta)+(1-\beta)f(0)\bigr\}, & \alpha=0.
    \end{cases}
\end{equation}
For standard choices of $f$-divergences (cf.~\cref{def:f-div}),
$\Psi_{\mathrm{KL}}(\alpha)=-\log\alpha$,
$\Psi_{\mathrm{TV}}(\alpha)=1-\alpha$,
with the usual conventions $-\log 0=+\infty$ and $0^{-1}=+\infty$.

\begin{theorem}[Optimal count-conditional frontier]
\label{thm:opt-count-cond-f-curve}
Fix a prompt $x\in\mathcal X$, source rates $p\in \Delta^{k-1}$, and a target joint law $\nu$ on $\mathcal A^m$. 
For every $t\ge m$, 
let $\vec C(t)$ 
be the source count process
at time $t$ (cf.~\cref{eq:feasible-support})
and write
\(
    \alpha_t
    \coloneq \nu(\Omega_m(\vec C(t))).
\)
On the event $\alpha_t>0$, the count-conditional output attribute law (cf.~\cref{eq:count-cond-output-label-law})
of A-RDC satisfies
\begin{equation}\label{eq:a-rdc-count-cond-output-label-law}
    P_{\textrm{A-RDC}}^{\vec C}\left(\widehat{A}_t \mid x \right)
    =\nu\left(\cdot\mid\Omega_m(\vec C(t))\right).
\end{equation}

Furthermore,
for every forward $f$-divergence,
\begin{equation}\label{eq:count-conditional-optimal-lb}
    \inf_{\mathrm{Alg}_m\in\mathfrak A_{m,t}}
    \Delta_f^C(\mathrm{Alg}_m;\nu;p) =
    \mathbb E\left[
        \Psi_f\left(\nu(\Omega_m(\vec C(t)))\right)
    \right],
\end{equation}
and
A-RDC achieves this minimum for each $t \ge m$ simultaneously.

\end{theorem}

\begin{algorithm}[tbp]
\caption{\textsc{Anytime RDC algorithms}}
\label{alg:anytime-rdc}
\footnotesize
\begin{algorithmic}[1]
\REQUIRE Output size $m$, target joint law $\nu$ on $[k]^m$,
annotator $\phi$, and variant: A-RDC, 
thresholded A-RDC (TA-RDC) 
with $f$-divergence
and tolerance $\varepsilon\ge0$, 
or capped A-RDC (CA-RDC) with cap $n\ge m$.
\ENSURE Returned outputs 
$(\widetilde Y_1,\ldots,\widetilde Y_m)$.

\STATE \textbf{Stopping module:} set
\[
    \sigma\coloneqq
    \begin{cases}
        +\infty, & \text{A-RDC},\\
        \tau_{f,\varepsilon}, & \text{TA-RDC},\\
        n, & \text{CA-RDC},
    \end{cases}
\]
where $\tau_{f,\varepsilon}$ is defined in
\cref{def:thresholded-anytime-rdc}.

\STATE E-sample $\vec A^\star\sim\nu$ and initialize
$\Omega_{m-1}\leftarrow\varnothing$,
$G_{\min}\leftarrow+\infty$, and
$\vec A^{\mathrm{cur}}\leftarrow\varnothing$.

\FOR{$t=1,2,\ldots$}
    \STATE Query one M-sample $Y_t$ and set
    $B_t\leftarrow\phi(Y_t)$.
    \IF{$t\ge m$}
        \STATE Set
        \begin{equation}
        \label{eq:alg-anytime-rdc-newly-feasible}
        \begin{aligned}
            \Omega_t
            &\leftarrow
            \bigl\{(B_{s_1},\ldots,B_{s_m}):
            s_1,\ldots,s_m\in[t]\text{ distinct}\bigr\},\\
            \mathcal D_t
            &\leftarrow\Omega_t\setminus\Omega_{t-1}.
        \end{aligned}
        \end{equation}

        \STATE \textbf{if} $\vec A^\star\in\Omega_t$ \textbf{then}
        set $\vec A^{\mathrm{cur}}\leftarrow\vec A^\star$,
        $N\leftarrow t$, and \textbf{break}.

        \STATE Set
        \begin{equation}
        \label{eq:alg-anytime-rdc-newly-feasible-mass}
            \delta_t\leftarrow\nu(\mathcal D_t).
        \end{equation}

        \IF{$\delta_t>0$}
            \STATE Independently E-sample
            \begin{equation}
            \label{eq:alg-anytime-rdc-incremental-draw}
                G_t\sim\operatorname{Exp}(\delta_t),
                \qquad
                \vec U_t\sim\nu(\,\cdot\mid\mathcal D_t).
            \end{equation}
            \STATE \textbf{if} $G_t<G_{\min}$ \textbf{then}
            set $G_{\min}\leftarrow G_t$ and
            $\vec A^{\mathrm{cur}}\leftarrow\vec U_t$.
        \ENDIF

        \STATE \textbf{if} $\vec A^{\mathrm{cur}}=\varnothing$ \textbf{then}
        choose any $\vec A^{\mathrm{cur}}\in\Omega_t$.

        \STATE \textbf{if} $t=\sigma$ \textbf{then}
        set $N\leftarrow t$ and \textbf{break}.
    \ENDIF
\ENDFOR

\STATE \textbf{return} outputs
$(\widetilde Y_1,\ldots,\widetilde Y_m)$
with attribute sequence $\vec A^{\mathrm{cur}}$, selected uniformly without
replacement within each attribute from $(Y_1,\ldots,Y_N)$.

\end{algorithmic}
\end{algorithm}
\begin{figure}
    \centering
    \includegraphics[width=0.90\linewidth]{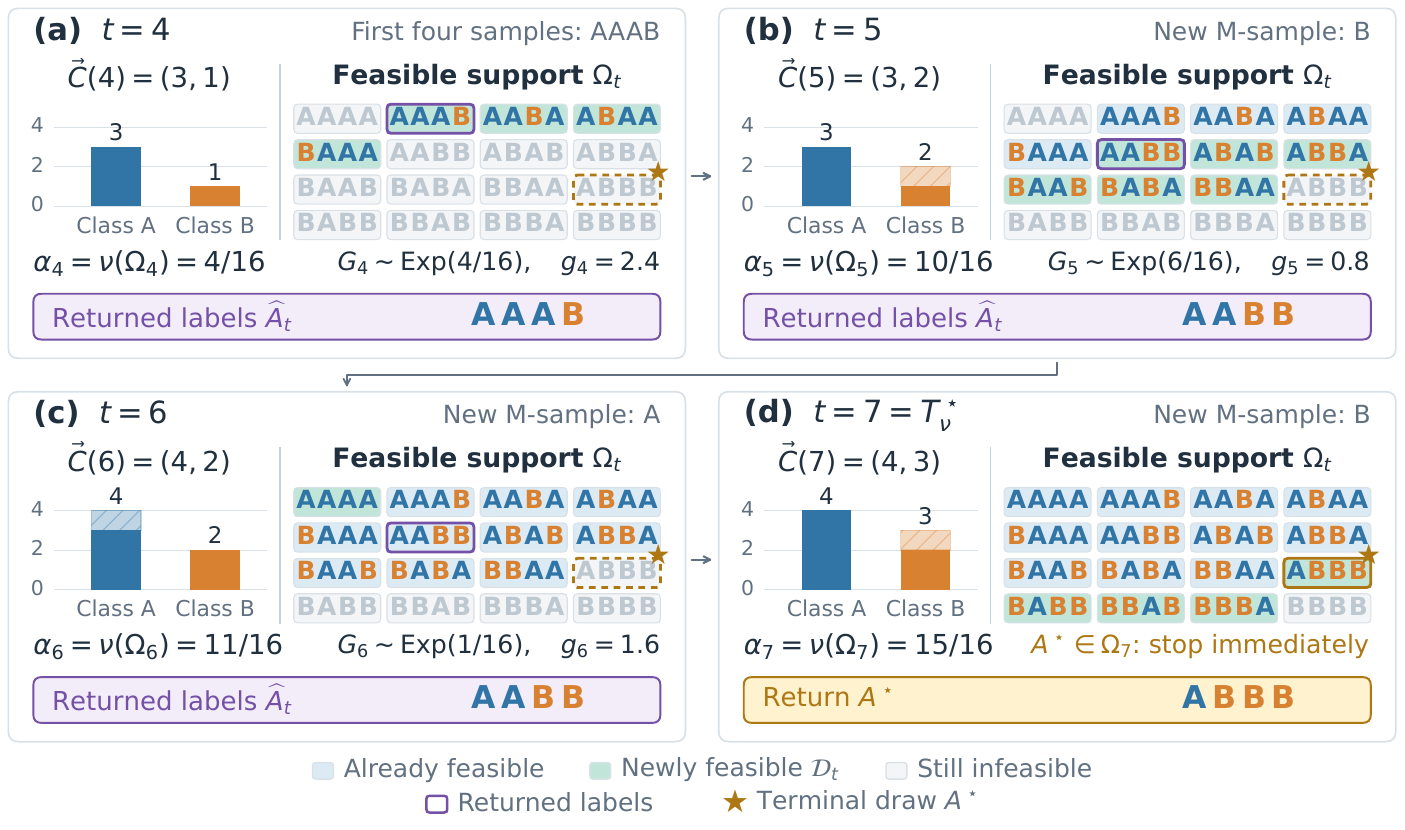}
    \caption{Illustration of the A-RDC algorithm of \cref{def:anytime-rdc}.
    There are $k=2$ attributes, $m=4$ outputs, and the target joint law is $\nu = q^{\otimes m}$ for uniform $q = (1/2,1/2)$.
    $\vec C(t)$, the source-count process of \cref{eq:feasible-support}, records the histogram of the M-sampled attributes.
    At time $t \ge m$,
    if the realized clock time $g_t < g_{\min} \coloneq \min_{1\le s <t}g_s$, we uniformly sample $\widehat A_t$ 
    from the newly feasible attribute sequences (in green). 
    Otherwise, we keep the returned attributes from time $t-1$, with the convention that $g_{\min} = +\infty$ for $t=1$.
    Once the terminal draw $A^\star \sim \nu$
    enters the feasible support, the A-RDC algorithm stops and returns $\vec A^\star$.
    }
    \label{fig:a-rdc-illustration}
\end{figure}

{\bf Thresholded anytime RDC.} 
We now consider a universally approximate 
algorithm that truncates the RDC stopping time with a count-conditional certificate.
This certificate controls the count-conditional discrepancy as a proxy for the unconditional objective.

\begin{definition}[Thresholded anytime RDC (TA-RDC)]
\label{def:thresholded-anytime-rdc}
For a forward $f$-divergence and tolerance
\(\varepsilon\ge0\), 
define the certification stopping time \(\tau_{f,\varepsilon}
    \coloneqq
    \inf\bigl\{t\ge m:\Psi_f(\alpha_t)\le\varepsilon\bigr\}\).
Then, \emph{TA-RDC} runs the A-RDC procedure \cref{def:anytime-rdc} 
but now
stops at \(N_{f,\varepsilon}^{\textrm{TA-RDC}}
    \coloneqq
    T_\nu^\star\wedge\tau_{f,\varepsilon}\)
and returns output attributes 
\(\widehat{\vec A}_{N_{f,\varepsilon}^{\textrm{TA-RDC}}}\) (cf.~\cref{eq:anytime-rdc-snapshot}); see \cref{alg:anytime-rdc}.
\end{definition}

\begin{theorem}[Finite-sample exact-or-certified optimality]
\label{thm:thresholded-anytime-rdc}
For every forward $f$-divergence and every \(\varepsilon\ge0\), the same
TA-RDC rule is source-rate-free and universally
$(f,\varepsilon)$-approximate (cf.~\cref{def:universal-approximate-sampling}). 
More precisely, it terminates
almost surely for every source rate $p \in \operatorname{int}(\Delta^{k-1})$ and satisfies
\begin{equation}\label{eq:thresholded-anytime-rdc-unconditional-bound}
    \Delta_f^{\mathrm{un}}
    \left(
        \mathrm{TA\text{-}RDC};\nu;p
    \right)
    \le
    \mathbb E_p\!\left[\Psi_f(\alpha_{\tau_{f,\varepsilon}})\right]
    \le\varepsilon.
\end{equation}
\end{theorem}

For the proof, see \cref{prf:thresholded-anytime-rdc}.
If the $f$-divergence is separable and \(\varepsilon=0\), then
\(T_\nu^\star \le \tau_{f,\varepsilon}\) almost surely, since $\Psi_f(\alpha_t) = 0$ if and only if $\alpha_t =1$. 
This occurs only when $\operatorname{supp}(\nu) \subseteq \Omega_t$, so that the feasible support contains the initially 
sampled attribute sequence $\vec A^\star \sim \nu$.
Thus,
\(N_{f,0}=T_\nu^\star\), and TA-RDC reduces to the A-RDC algorithm of \cref{def:anytime-rdc}.

For the product target \(\nu=q^{\otimes m}\) and forward KL divergence,
abbreviate the M-cost of TA-RDC as
\(\mathcal C_{m,\varepsilon}^{\mathrm{TA\text{-}RDC}}(p,q)\), and write \(\mathcal C_{m,\varepsilon}^{\mathrm{univ}}(p,q)
\coloneqq
\mathcal C_{m,\varepsilon}^{\mathrm{univ},\mathrm{KL}}
\bigl(p;q^{\otimes m}\bigr) \)
for the corresponding universal source-rate-free M-cost.
Then,
for any specified tolerance $\varepsilon >0$, 
TA-RDC is first-order optimal among all source-rate free universally $(\textrm{KL},\varepsilon)$-approximate algorithms.

\begin{theorem}[Finite-sample and first-order universal approximate optimality of TA-RDC]\label{thm:first-order-uni-opt-ta-rdc}
    Fix $p,q\in\operatorname{int}(\Delta^{k-1})$.
    For $\delta\ge0$, define the KL frontier
\[
\rho_{p,q}^\star(\delta)
\coloneqq
\min_{\substack{r\in\Delta^{k-1}\\
D_{\mathrm{KL}}(r\|q)\le \delta}}
\max_{i\in[k]}\frac{r_i}{p_i}.
\]
For any sequence $\varepsilon_m\ge0$ such that
$\varepsilon_m/m\to \delta\ge0$, the M-cost of
TA-RDC and universal source-rate-free M-cost (cf.~\cref{def:uni-source-rate-free-M-cost}) 
satisfy
\begin{equation}\label{eq:first-order-opt-ta-rdc}
    \frac{\mathcal C_{m,\varepsilon_m}^{\mathrm{TA\text{-}RDC}}(p,q)}{m}
\longrightarrow\rho_{p,q}^\star(\delta),
\qquad
\frac{\mathcal{C}_{m,\varepsilon_m}^{\mathrm{univ}}(p,q)}{m}
\longrightarrow\rho_{p,q}^\star(\delta),
\end{equation}
respectively.
Moreover, for every $m\ge1$ and every $\varepsilon\ge0$,
\begin{equation}\label{eq:ta-rdc-finite-m-additive-bound}
    \mathcal C_{m,\varepsilon}^{\mathrm{TA\text{-}RDC}}(p,q)
    \le
    \mathcal C_{m,\varepsilon}^{\mathrm{univ}}(p,q)
    +\frac{3\sqrt m}{\min_{i\in[k]}p_i}.
\end{equation}
In particular, the relative optimality is uniform over the tolerance:
\begin{equation}\label{eq:ta-rdc-uniform-first-order-optimality}
\sup_{\varepsilon\ge0}
\frac{\mathcal C_{m,\varepsilon}^{\mathrm{TA\text{-}RDC}}(p,q)}
{\mathcal C_{m,\varepsilon}^{\mathrm{univ}}(p,q)}
\longrightarrow1.
\end{equation}
\end{theorem}
For the proof, see \cref{prf:first-order-uni-opt-ta-rdc}.
At $\varepsilon=0$, the stopping time rule of TA-RDC coincides with that of RDC and \cref{thm:first-order-uni-opt-ta-rdc}
recovers
the first-order optimality of RDC. 
We study the 
finite sample performance of TA-RDC in \cref{sec:synth-exp} and find that the uniform 
convergence occurs quickly for moderate $m$ under various
settings. 
We also give first-order optimal frontiers under other divergences 
and procedures that attain them in
\cref{sec:algos-anytime-rdc}.

{\bf Capped anytime RDC.}
If a practitioner specifies a query cap \(n\ge m\), they may run A-RDC
until either the exact completion time \(T_{\nu}^\star\) or the cap \(n\)
is reached.  This gives the \emph{capped-budget anytime RDC ($\textrm{CA-RDC}$)}
procedure, which stops at time
\(N_{n}^{\mathrm{CA\text{-}RDC}}
    \coloneqq
    T_{\nu}^\star\wedge n\)
and returns output attributes
\(\widehat{A}_{T_{\nu}^\star\wedge n}\) (cf.~\cref{eq:anytime-rdc-snapshot}); see \cref{alg:anytime-rdc}. 

Since CA-RDC and A-RDC agree pathwise
(i.e., \( \widehat{A}_{T_{\nu}^\star\wedge n}
=
\widehat{A}_n\) a.s.), 
the discrepancy of CA-RDC under any forward $f$-divergence
satisfies 
\[
    \Delta_f^{\mathrm{un}}
    \left(
    \mathrm{CA\text{-}RDC};
    \nu;
    p
    \right)
    \le 
    \mathbb E_p\left[
    D_f\left(
     P_{\textrm{A-RDC}}^{\vec C}\left(\widehat{A}_n\right)
    \middle|
    \nu
    \right)
    \right]
    = 
    \mathbb E_p\left[
    \Psi_f(\alpha_n)
    \right]
\]
Thus, the unconditional discrepancy of CA-RDC is bounded above by the 
minimum count-conditional discrepancy attainable by any post-processing algorithm with a capped budget $N \le n$. 
We numerically compare the thresholded and capped anytime RDC methods in \cref{sec:synth-exp}.

\section{Experiments}\label{sec:exp}
We provide simulations and experiments with generative AI models to evaluate statistical attribute alignment, with applications to demographic representation and synthetic dataset construction.

\subsection{Simulations}\label{sec:synth-exp}

In order to provide support for the validity of our theoretical results, we first provide experiments in simulated settings. 

{\bf Setting and methodology.}
We consider a synthetic setting with $k=8$ attributes, source and target rates \( p=(0.03,0.05,0.08,0.12,0.16,0.18,0.18,0.20)\) and \(q=(0.14,0.12,0.13,0.15,0.14,0.12,0.11,0.09)\), respectively, and vary the output size $m \in \{100,500,1000\}$. 
We also set the target joint law to
\(\nu=q^{\otimes m}\).
There are no additional features in any of the three settings:
the generated outcomes are identified with their attributes.
Additional experiments varying
the source rates, target rates, number of attributes, and output size are
reported in \cref{app:exp:results-synth-exp}.

We consider the following post-processing algorithms.
\begin{enumerate}
    \item[(i)] \textit{DS}, direct-sampling, queries the model $m$ times and returns the queried outputs.
    It lies in $\mathfrak{A}_{m}^o$, and its unconditional output attribute law is $p^{\otimes m}$.
    Thus, unless the outputs are already aligned, in the sense that the source and target rates are equal $p=q$, DS has a positive statistical alignment discrepancy under any separable forward $f$-divergence (cf.~\cref{def:f-div}).

    \item[(ii)] \textit{RDC}, random demand coupon collector (cf.~\cref{alg:rdc} in \cref{sec:exact-sampling}),
    lies in $\mathfrak{A}_{m}^o(q^{\otimes m})$ (cf.~\cref{def:universal-exactness}).
    It
    first E-samples an ordered target sequence with count vector $\vec L\sim\Mult(m,q)$ and then queries the source until every demand is feasible. 
    Its M-cost is the expected value of the stopping time
    \(
        T_{\nu}^\star=\inf\{t:\vec C(t)\ge \vec L\}.
    \)
    RDC returns the initially sampled target sequence, so its unconditional discrepancy is exactly zero.
    
    \item[(iii)] \textit{Oracle RS}, oracle rejection sampling (cf.~\cref{def:oracle-rs} in \cref{app:ors}), lies in $\mathfrak{A}_{m}(q^{\otimes m})$.
    It sequentially M-samples an output $Y_s \sim P(\cdot \mid x)$ 
    and accepts it with probability 
    $\left( \max_{j\in \textrm{supp}(q)}{q_j}/{p_j} \right)^{-1} \cdot q_i/p_i$ where $i =\phi(Y_s) \in [k]$ is the annotated attribute.
    It continues to M-sample until it accepts $m$ proposals, so that its
    M-cost is \( m \cdot  \max_{j\in \textrm{supp}(q)}{q_j}/{p_j}\).
    The distribution of the attributes conditional on acceptance is exactly $q$, so the unconditional discrepancy of Oracle RS is exactly zero.
    This method minimizes the 
    M-cost among all universally exact post-processing algorithms with known source rates (cf.~\cref{prop:opt-perfect-sampling-fairness-oracle})
    However, since the acceptance probabilities
    require knowledge of the source rates, this is not a practical black-box algorithm.

    \item[(iv)] \(\textit{TA-RDC}\), 
    thresholded anytime random demand coupon collector (cf.~\cref{alg:anytime-rdc} in \cref{sec:algos-anytime-rdc}), lies in $\mathfrak{A}_{m}^o$.
    This algorithm runs anytime RDC in \cref{def:anytime-rdc} until either the stopping time $T_{\nu}^\star$ or statewise $f$-certificate at a specified level of $\varepsilon\ge 0$
    is reached.
    
    \item[(v)] \(\textit{CA-RDC}\), 
    capped anytime random demand coupon collector (cf.~\cref{alg:anytime-rdc} in \cref{sec:algos-anytime-rdc}), lies in $\mathfrak{A}_{m}^o$.
    This algorithm runs anytime RDC in \cref{def:anytime-rdc} until either the stopping time $T_{\nu}^\star$ or cap $n$ is reached.

\end{enumerate}

{\bf Reported metrics.}
For \(d\in[0,D_{\mathrm{KL}}(p\|q)]\), let
\[
    \bar\Delta_m(d)
    \coloneqq
    \frac{1}{m}\Delta_{\mathrm{KL}}^{\mathrm{un}}
    \bigl(\mathrm{TA\text{-}RDC}(md)\bigr),
    \qquad
    \bar C_m(d)
    \coloneqq
    \frac{1}{m}\E_p\!\left[
        N_{\mathrm{KL},md}^{\mathrm{TA\text{-}RDC}}
    \right],
\]
be the unconditional discrepancy and M-cost of TA-RDC normalized by the number of outputs $m$.
We vary the normalized tolerance \(d\) to trace three
cost--discrepancy profiles: the realized TA-RDC profile
\(d\mapsto(\bar\Delta_m(d),\bar C_m(d))\), its nominal certificate
profile \(d\mapsto(d,\bar C_m(d))\), and the first-order optimal KL
frontier \(d\mapsto(d,\rho_{p,q}^\star(d))\)
(cf.~\cref{thm:first-order-uni-opt-ta-rdc}).

For CA-RDC, we vary the cap $n \in [m,6m]$
to trace out a cost-discrepancy profile
\((\Delta_{\mathrm{KL}}^{\mathrm{un}}
\left(
    \mathrm{CA\text{-}RDC}
\right)/m, \E_p[N_{n}^{\textrm{CA-RDC}}]/m)\).
Population-level expressions and Monte Carlo estimators for the M-cost and unconditional discrepancies 
are given in 
\cref{app:metric-evaluation}.

{\bf Results.}
\Cref{fig:synth-cost-discrepancy-profiles} compares TA-RDC and CA-RDC
at three output sizes $m \in [100,500,1000]$.
RDC and Oracle RS both
have zero discrepancy, but RDC is source-rate-free whereas Oracle RS
uses the source rates and is thus not a practical method.
At the opposite endpoint, direct-sampling returns the first \(m\) outputs and hence has normalized M-cost one but a large sampling discrepancy.

Between these endpoints, the anytime RDC methods of TA-RDC and CA-RDC 
achieve normalized
KL discrepancies 
below \(0.05\) 
at normalized M-costs ranging from \(1.85\) to \(2.18\). 
Thus, in this setting, querying the model only about twice as many
times as the number of returned outputs
is sufficient to
substantially reduce the discrepancy from direct-sampling.

The finite-sample profiles also closely track the first-order optimal
frontier, even with a moderate output size of \(m=100\).
Moreover, the nominal statewise certificate for TA-RDC remains  close to its unconditional KL discrepancy. 
This illustrates that the count-conditional objective underlying the certificate can be a
useful proxy for the unconditional sampling discrepancy.

\begin{figure}[tbp]
    \centering
    \includegraphics[width=\linewidth]{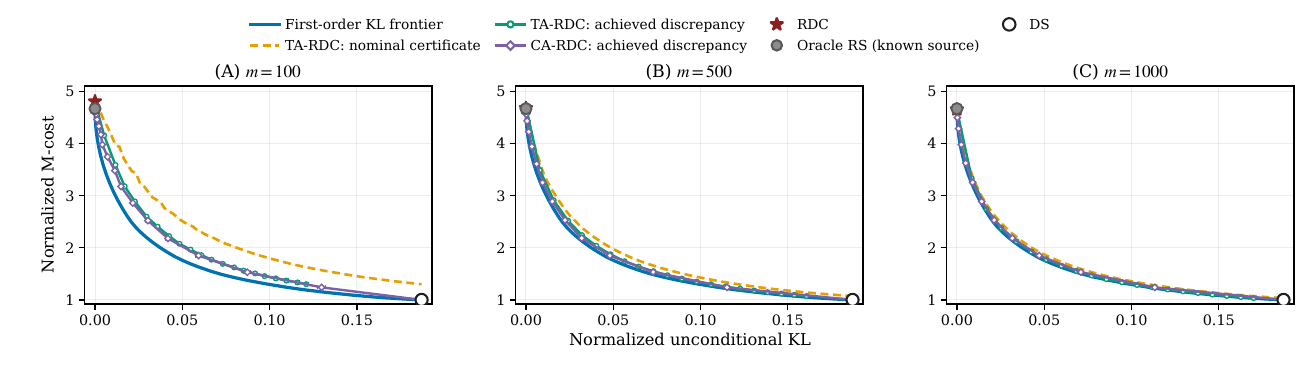}
    \caption{\footnotesize
    Cost--discrepancy profiles under KL divergence for the \(k=8\)
    non-uniform source and target distributions with output sizes
    \(m\in\{100,500,1000\}\).
    The horizontal axis reports normalized unconditional KL discrepancy,
    \(D_{\mathrm{KL}}(P_{\mathrm{Alg}_m}\|q^{\otimes m})/m\), and the
    vertical axis reports normalized M-cost, \(\mathbb E_p[N]/m\).
    The nominal TA-RDC certificate and the first-order optimal frontier
    are included for comparison.
    RDC and Oracle RS attain zero discrepancy, although Oracle RS
    requires knowledge of the source rates.
    DS denotes direct-sampling, which returns the first \(m\) candidate
    samples and has normalized M-cost one.
    }
    \label{fig:synth-cost-discrepancy-profiles}
\end{figure}

\subsection{Empirical experiments}\label{sec:emp-exp}

We now provide empirical experiments with various generative AI models. 
\subsubsection{Protected attribute alignment in text-to-image generation}\label{sec:t2i-exp-general}

We consider
text-to-image (T2I) generation settings:
generating (i) headshots of individuals (\cref{sec:exp-headshots-individuals}), 
and (ii) group photos of people (\cref{sec:exp-group-photos})
in a professional 
medicine setting.
We seek to align the distribution of demographic attributes, including
gender, race, and/or age, with a specified target.
This is a fairness-motivated instance of statistical attribute alignment.
In our experiments, 
we consider various methods that generate batches of images (or, M-samples)
and black-box post-processing methods
that are  
applied to these batches.
Since our focus is on black-box access, we compare with approaches that are applicable in that setting.

{\bf M-sampling methods.} 
We consider the following M-sampling methods that modify the text of the
input prompt to the generative AI model. These are roughly ordered by
increasing demographic specificity of the input prompt, and help showcase
different ways in which generated attribute distributions may remain misaligned
even after a prompting intervention is applied.

\begin{enumerate}
    \item[(i)] \textit{HPS} (Hard Prompt Searching,
    \citealp{ding2021cogview})
    fixes only the professional
    context of the prompt 
    and does not mention
    protected attributes such as gender, age, or race. 
    This tests the
    model's default distribution over protected attributes.
    For example, in the individual-headshot setting, the
    prompt is
    \textit{``Show a photorealistic professional headshot of a person working
    in medicine or healthcare, in professional attire.''}

    \item[(ii)] \textit{TSI} (Targeted Sampling Instruction) uses a single
    natural language prompt to ask the generator to sample from the desired
    protected attribute law. This baseline is inspired by work on ethical
    natural language interventions \citep{zhao2021ethical,bansal2022well} and
    prompt-based bias control in text-to-image generation \citep{clemmer2024precisedebias,shin2024can}. 
    In particular, TSI leaves
    the randomization implicit in the model by instructing it to first choose
    a target category or count-vector attribute uniformly at random and then
    generate an image matching that choice.

    \item[(iii)] \textit{Demographic HPS} (Hard Prompt Searching,
    \citealp{ding2021cogview}) externally randomizes over attribute-specific hard
    prompts. 
    In detail, HPS forms a set of hard prompts
    \(\{x_1',\ldots,x_k'\}\), where \(x_i'\in\mathcal X\) 
    requests 
    attribute \(i\in[k]\), and samples from the mixture
    \(
        P_{\rm HPS}(\cdot\mid x)
        =
        \sum_{i=1}^k
        q_i \cdot P_\theta(\cdot\mid x_i').
    \)
    For example, an attribute-specific prompt \(x_i'\) in the
    professional-headshot setting is
    \textit{``Show a photorealistic professional headshot of a young Asian
    woman working in medicine or healthcare, in professional attire.''}
        Thus, HPS differs from TSI as the randomization 
        is performed externally, before querying the generative model.
\end{enumerate}

{\bf Post-processing methods.}
We consider 
the RDC and anytime RDC algorithms TA-RDC and CA-RDC from
\cref{sec:synth-exp}, as well as direct-sampling (DS) as a baseline.
As post-processing methods, we apply them
on top of any of the M-sampling methods 
(e.g., TSI or demographic HPS).

In settings where hard prompting
reliably returns outputs with the prompted attribute, post-processing methods 
cannot meaningfully reduce the unconditional discrepancy (cf.~\cref{def:uncond-sampling-fairness-discrepancy}). Indeed, if hard prompting is perfectly reliable, 
one could randomly sample output attributes so that the source rates $p$ are equal to those of the target $q$.
Then, the discrepancy of a method that returns the first $m$ outputs is exactly zero.

However, as we will demonstrate empirically, HPS and TSI can fail,
because the generative AI model may not necessarily follow instructions or satisfy constraints reliably.
This is 
especially true 
for more niche or derived attributes 
(e.g., the social vulnerability index in persona generation tasks; see,~\cref{sec:exp-persona-generation}).

{\bf GenAI models used.}
Publicly available and open-weight diffusion models 
have been widely adopted as
synthetic image generators in applied data augmentation studies,
see, e.g., \cite{stockl2023evaluating, tian2023stablerep, bluethgen2025vision}
and \cref{app:exp:details-emp-exp-general} for further discussion.
Consequently,
we generate candidate images using open-weight text-to-image endpoints hosted by \texttt{fal},\footnote{\url{fal.ai}} 
a platform that provides querying access to various models.

Specifically, we consider the \texttt{fal}-hosted T2I models
\texttt{Flux-2-dev} \citep{blackforestlabs2026flux2dev},
~\texttt{Qwen-Image-2512} \citep{qwen2025qwenimage2512}, 
and
\texttt{HiDream-O1-Image} \citep{hidream2026o1image}, which we abbreviate to \texttt{Flux}, \texttt{Qwen}, and \texttt{HiDream}, respectively.
For certain combinations of M-sampled methods (e.g., HPS and TSI) and models (e.g., \texttt{Flux} and \texttt{Qwen}), we defer T2I 
empirical results to 
\cref{app:exp:results-individual-headshot} 
and \cref{app:exp:results-group-photo} for conciseness.

For annotation, 
we use OpenAI API calls to \texttt{GPT-5.5} as a
vision-language annotator \citep{openai2026gpt55docs}.
Prior work has shown this approach can have high agreement with human annotation.\footnote{See \cref{app:exp:details-emp-exp-general} for further discussion.}
General implementation details for T2I experiments are given in \cref{app:exp:details-emp-exp-general} with experiment specific prompts, 
model configuration
settings,
and annotation procedures provided in \cref{app:exp:details-t2i-individual-headshot} and \cref{app:exp:details-t2i-group-photo}.

{\bf Experimental setting and reported metrics.}
For each T2I experiment, model, and M-sampling method, we first
generate and annotate a fixed pool of images
\(Y_s\in\mathcal{Y}\), \(s=1,\ldots,1000\),
which we re-use across
evaluations.
The target joint law is $\nu = q^{\otimes m}$ where the experiment specific target rates $q$ 
vary across tasks with different attribute definitions and numbers of attributes $k$.

We report two types of quantities.
First, 
we record three 
pool-level diagnostics that depend only on the frozen pool of 1000 images and the target rates: 
(i)
the \textit{on-support rate} is
the fraction
of images whose annotations lie in
\(\operatorname{supp}(q)\);
(ii) the \textit{target-attribute coverage} is the number of 
unique attributes within the support of the target rates that are observed in the pool of images
and, (iii), for hard prompt M-sampling methods that directly request a protected
attribute group (e.g., demographic HPS), 
we additionally record the \textit{compliance rate}, defined as the
fraction of images whose requested and realized annotations
match.

Second, for each post-processing method $\textrm{Alg}_m$,
we approximate KL statistical alignment discrepancies 
by 
sampling outputs with replacement from the frozen pool of images (cf.~\cref{app:metric-evaluation}).
Since the T2I model may produce images whose annotated attribute lies outside the support of the target rates, $\textrm{supp}(q)$,
there is a strictly positive probability that an M-sampled batch of $n$ outputs $\mathcal{B}_n = (Y_1,\ldots, Y_n)$ 
admits no feasible attribute sequence under the target joint law.
In this case, the KL discrepancy of CA-RDC is necessarily $+\infty$, so we instead
report the feasible-run discrepancy \(D_{\mathrm{KL}}(P_{\textrm{CA-RDC}}(\widehat A_n\mid\alpha_n>0)\|q^{\otimes m})\) and the infeasibility probability 
$\beta_n \coloneq 1 - \Pr(\alpha_n > 0)$;
see \cref{app:metric-evaluation-ca-rdc} for more details.

In the main text, we report results for \(m=50\) returned images and, as in \cref{sec:synth-exp}, trace out cost-discrepancy profiles by varying either the normalized tolerance level $d \coloneq \varepsilon/m \in [0,3]$ 
or the
capped-budget \(n \in [m,20m]\) for TA-RDC and CA-RDC, respectively.
We defer results for discrepancies under other divergences and
settings with other output sizes 
to \cref{app:exp:results-individual-headshot,app:exp:results-group-photo,app:exp:results-svi}.

\subsubsection{Professional headshots of individuals}\label{sec:exp-headshots-individuals}

{\bf Experimental setup.}
The target distribution $q$
is uniform over tuples of  
(1)  a binary perceived-gender attribute with bins \{\texttt{male-presenting}, \texttt{female-presenting}\};
(2)
age over two categories with bins \{\texttt{young}, \texttt{old}\}; and
(3)
race/ethnicity presentation over four categories with bins
\{\texttt{White/Caucasian}, \texttt{Black/African descent}, \texttt{Asian}, \texttt{Hispanic/Latino}\}, 
yielding a total of $k=16$ attributes.
We abbreviate the gender, age, and race bins as M/F, Y/O, and W/B/A/H, respectively.

For this experiment, the annotator is required to select one of the 16 attributes, so that the annotated attribute of all sampled images lies in the support of the target rates by construction;
see \cref{sec:exp-group-photos,sec:exp-persona-generation} below for 
experiments that allow outputs whose annotated attribute lies outside $\textrm{supp}(q)$.

{\bf Main results.}
\textit{Hard prompt compliance and observed attributes.}
For \texttt{HiDream}, 
HPS and TSI only produce three and
five target attributes out of 16 in our 1000 image pool, respectively (cf.~\cref{tab:individual-headshot-hidream-support}).
Thus,  
these two prompts 
do not produce images from most 
target attributes. 
In contrast, 
demographic HPS, which explicitly prompts for a target attribute,
covers all 16 attributes and has a high
compliance rate of \(0.974\).
We study why hard prompting fails
in \cref{app:exp:results-individual-headshot}.

\begin{table}
\centering
\small
\setlength{\tabcolsep}{3pt}
\caption{\footnotesize Pool-level diagnostics 
for the individual-headshot prompt sources using the \texttt{HiDream-O1-Image} model. 
Attribute abbreviations are gender/age/race: F/M, Y/O, and W/B/A/H.}
\label{tab:individual-headshot-hidream-support}
\footnotesize
\begin{tabular}{@{}p{0.20\linewidth}ccp{0.36\linewidth}@{}}
\toprule
M-sampling method & Compliance rate & Target-attribute coverage & Attributes observed\\
\midrule
HPS & -- & 3/16 & F/Y/W, F/Y/B, F/Y/H \\
\cmidrule(lr){1-4}
TSI & -- & 5/16 & F/Y/B, F/Y/H, F/O/B, F/O/A, F/O/H \\
\cmidrule(lr){1-4}
Demographic HPS & 0.974 & 16/16 & All \\
\bottomrule
\end{tabular}
\end{table}

\textit{Reductions in sampling discrepancies.}
Across all three T2I models with demographic HPS M-sampling,
the anytime RDC methods consistently decrease the unconditional KL discrepancies towards zero (cf.~\cref{fig:emp-t2i-individual-headshot-cost-discrep-profiles}).
For example, at capped-budget \(n=500\), the KL discrepancy of CA-RDC is below
\(1.5\times10^{-4}\) across all models 
with normalized M-costs of \(5.14\) for \texttt{Flux}
and $\approx 3.10$ for both \texttt{Qwen} and \texttt{HiDream}.
Similarly, at closely matched normalized M-costs
of \(5.17,3.06\) and \(3.11\) for
\texttt{Flux}, \texttt{Qwen} and
\texttt{HiDream}, respectively,
TA-RDC achieves discrepancies
below \(1.3\times10^{-4}\) for all models.

\begin{figure}
    \centering
    \includegraphics[width=1.0\linewidth]{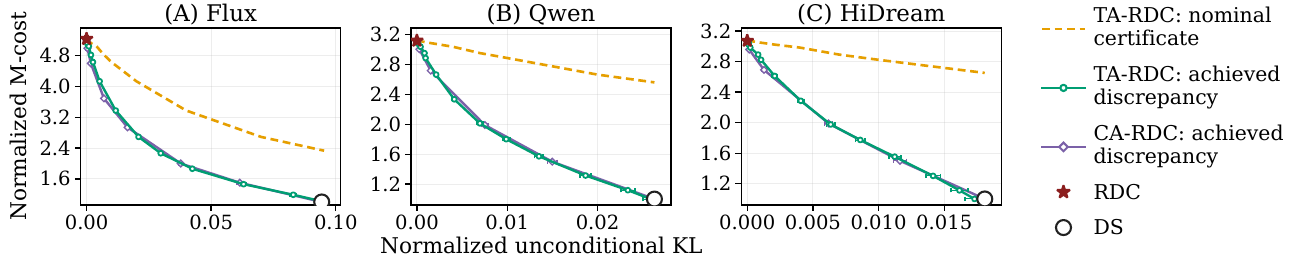}
    \caption{\footnotesize
    Statistical alignment discrepancies under KL divergence normalized by $m=50$, the number of returned images.
    The images in the pool are generated using 
    demographic HPS and are annotated into $k=16$ target attributes.
    The panels correspond to the three T2I models.
    }
    \label{fig:emp-t2i-individual-headshot-cost-discrep-profiles}
\end{figure}

\subsubsection{Synthetic datasets for pairwise face recognition audits}
\label{sec:exp-pairwise-audit}
We next consider a synthetic dataset curation task.
In facial
identity audits, 
a verifier seeks to identify duplicate pairs of individuals in
a returned set of face images. 
A central quantity in such systems is the false
match rate over image pairs, which may vary with demographic attributes
\citep[see, e.g.,][]{klare2012face, grother2019face, grother2022face}.
This quantity depends on 
the distribution of pairs of images, rather than only on marginal distributions.
Hence, proportion matching (PM), 
which selects feasible attribute counts as close as possible to target proportions, can be misleading here.

{\bf Experimental setup.}
Using the demographic HPS headshot pools from \cref{sec:exp-headshots-individuals}, we follow a standard face-verification workflow of embedding two face
images and applying a threshold to their embedding similarity as outlined in
\cite{schroff2015facenet,deng2019arcface}.
In particular, 
each image
is embedded using a Siamese face-recognition network
(\texttt{InsightFace} \texttt{buffalo\_l}\footnote{\url{https://github.com/deepinsight/insightface}}),
and a pair is classified as a same-person match when the cosine similarity
of the two embeddings exceeds a threshold \(t\). 
Because no identity ground truth is available for the generated images, we adopt the operational convention that each generated image represents a distinct individual; under this convention, any positive same-person decision is counted as a false match.
We keep the same \(k=16\) demographic classes, output size
\(m=50\), and target joint law  $\nu = q^{\otimes m}$ for uniform $q$.
We compare our CA-RDC method
with proportion matching at a capped-budget of $n=200$.
As a benchmark,
we also record downstream metrics of an oracle whose output label law is exactly the joint target law.

{\bf Results.}
At cosine threshold $t=0.45$, CA-RDC yields false-match rates closer to the oracle benchmark across all three image models; see \cref{tab:pairwise-identity-audit-cross-model}. 
For \texttt{HiDream}, the oracle rate is $6.404\%$, compared with $6.099\%$ for CA-RDC and $4.670\%$ for PM.
Thus, targeting a marginal objective underestimates the false-match rate.
This illustrates that a
synthetic data curation task can require specifying and approximating a joint attribute law rather than a marginal law.
Implementation details and a threshold sensitivity analysis are deferred to \cref{app:exp:details-results-pairwise-audit-ca-rdc-vs-pm}.

\begin{table}[t]
\centering
\small
\setlength{\tabcolsep}{3pt}
\caption{\footnotesize
Synthetic pairwise audit task at cosine threshold $t=0.45$. Same- and different-attribute match rates are estimated from all valid image pairs in the frozen pool of 1000 demographic HPS images. 
Oracle is an infeasible benchmark with output attribute law equal to the target joint law $\nu = q^{\otimes m}$.
Parentheses report standard errors across ten sampled batches.}
\label{tab:pairwise-identity-audit-cross-model}
\begin{tabular}{lccccc}
\toprule
Model &
$\widehat\alpha_{\rm same}^{(2)}$ &
$\widehat\alpha_{\rm diff}^{(2)}$ &
Oracle & CA-RDC & PM \\
\midrule
\texttt{Flux} &
$3.726\%$ &
$0.0395\%$ &
$0.270\%$ &
$0.257\%\;(0.005\%)$ &
$0.205\%\;(0.001\%)$ \\
\texttt{Qwen} &
$83.019\%$ &
$0.449\%$ &
$5.610\%$ &
$5.295\%\;(0.112\%)$ &
$4.089\%\;(0.000\%)$ \\
\texttt{HiDream} &
$94.662\%$ &
$0.520\%$ &
$6.404\%$ &
$6.099\%\;(0.156\%)$ &
$4.670\%\;(0.000\%)$ \\
\bottomrule
\end{tabular}
\end{table}

\subsubsection{Compositional group photos}
\label{sec:exp-group-photos}

{\bf Experimental setup.}
We next consider a compositional group photo setting, where each generated image
is requested to 
contain a foreground trio of adults in a medical/healthcare scene. 
The attributes we aim to balance
are the demographics of the people. 
In particular, 
we keep track of the counts of the group members'
age, gender, and race/ethnicity.
The target distribution is uniform 
over all possible combinations of (one man \& two women OR two men \& one woman)
$\times$ one White/Caucasian, one Black/African descent, \&
(one Hispanic/Latino OR one Asian adult),
$\times$
(two young adults \& one older adult OR three young adults).
This yields \(k = 2 \times 2 \times 2 = 8\) attributes.

TSI and demographic HPS prompts may request, for example, 
an image whose three foreground adults comprise of two women \& one man; two young adults \& one older adult; and one White/Caucasian, one Asian, \& one person of Black/African descent.
Consequently, the prompt does not specify the exact demographics of each individual.
Full details are deferred to \cref{app:exp:details-t2i-group-photo}.

\begin{figure}[!t]
    \centering
    \includegraphics[width=0.8\linewidth]{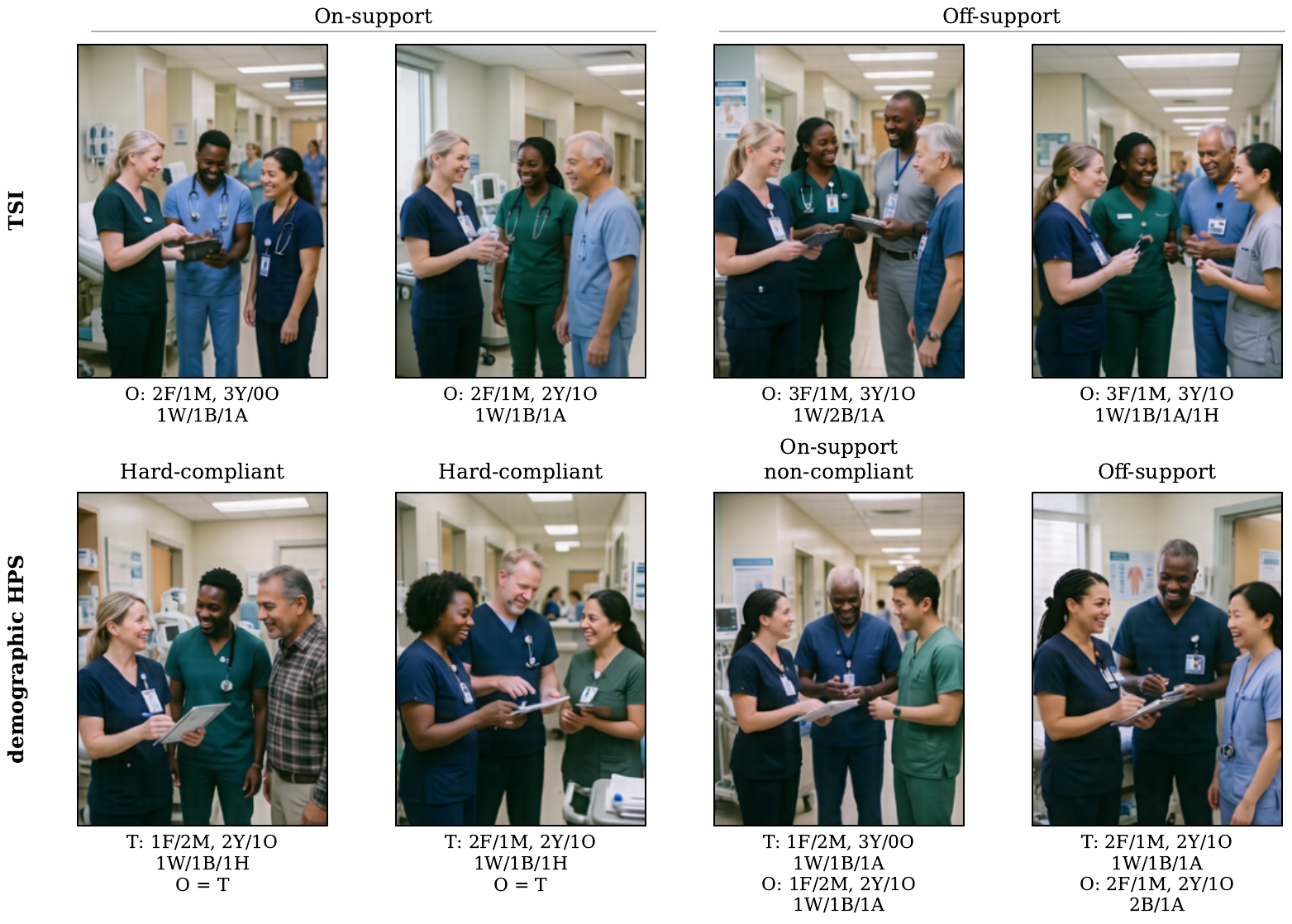}
    \vspace{-0.25em}
    \caption{\footnotesize Examples of \texttt{HiDream} group-photo generations. The top row shows on- and off-support TSI generations; captions report the annotated observed count vector (O). The bottom row shows hard-compliant, on-support non-compliant, and off-support demographic HPS generations; captions report the prompted target count vector (T) and annotated observed count vector (O). Abbreviations are F/M for female/male, Y/O for young/old, and W/B/A/H for White/Black/Asian/Hispanic.
    }
    \label{fig:group-photo-hidream-examples}
\end{figure}

{\bf Main results.}
\textit{Qualitative results.}
\Cref{fig:group-photo-hidream-examples} shows that TSI 
often does not follow instructions.
The same figure separates
demographic HPS outputs into hard-compliant, 
on-support non-compliant, and
off-support examples. 
Importantly, 
images that are 
non-compliant yet on-support can be used to improve the count-conditional sampling discrepancy via post-processing methods like the anytime RDC algorithms.

\textit{Hard prompt compliance and on-support rate.}
\Cref{tab:group-photo-hidream-support} shows that 
models find it 
substantially harder 
to follow instructions about group composition than in the
individual headshot experiments. 
For the
\texttt{HiDream} model, TSI has on-support rate \(0.028\), while demographic HPS
raises the on-support rate to \(0.624\) and covers all eight target attributes.
Even under demographic HPS, however, the 
compliance rate is only
\(0.219\), 
indicating that following group compositional constraints
is difficult.

\begin{table}
\centering
\footnotesize
\caption{\footnotesize Pool-level diagnostics for the group-photo prompt sources using the \texttt{HiDream-O1-Image} model.}
\label{tab:group-photo-hidream-support}
\begin{tabular}{lrrc}
\toprule
M-sampling method & On-support rate & Compliance rate & Attributes observed \\
\midrule
TSI & 0.028 & -- & 6/8 \\
\cmidrule(lr){1-4}
Demographic HPS & 0.624 & 0.219 & 8/8 \\
\bottomrule
\end{tabular}
\end{table}

\textit{Reductions in sampling discrepancy.}
Even when explicitly specifying the demographics in the input prompt, \texttt{HiDream}
produces images whose attribute lies outside of the target support.
Consequently, 
we report the feasible-run discrepancy 
as well as the infeasibility probability $\widehat \beta_n$ (cf.~\cref{sec:t2i-exp-general}). 

As seen in \cref{fig:emp-t2i-group-photo-cost-discrep-profiles},
TA-RDC and CA-RDC meaningfully reduce KL discrepancies: 
CA-RDC with a capped-budget of $n=1000$ achieves a feasible-run discrepancy of \(3.6\times10^{-4}\),
with a normalized M-cost of
\(11.85\).
Similarly, TA-RDC reduces the discrepancy to \(8.2\times10^{-4}\) 
at a matched M-cost of \(11.97\).

Furthermore,
the infeasibility probability decreases rapidly as the capped-budget increases.
At \(n=100\), $\widehat\beta_n$ is \(0.0043\) and by \(n=200\) it falls to less than \(10^{-27}\).
Thus, even in the presence of substantial off-target samples,
CA-RDC requires only a moderate query cap to produce a feasible
output sequence 
with probability effectively equal to one.

\begin{figure}
    \centering
    \includegraphics[width=1.0\linewidth]{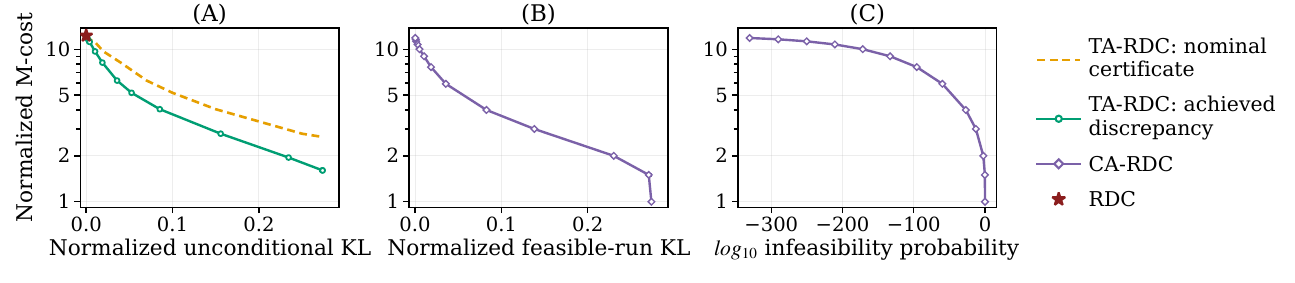}
    \caption{\footnotesize
    Various metrics for the group-photo setting with 
    demographic HPS and the \texttt{HiDream} model.
    There are $k=8$ target attributes and each post-processing method returns $m=50$ images. 
    Because hard prompting 
    generates images whose annotated attribute lies outside $\textrm{supp}(q)$, 
    we report the feasible-run discrepancies \(D_{\mathrm{KL}}(P_{\textrm{CA-RDC}}(\widehat A_n\mid\alpha_n>0)\|q^{\otimes m})\) and the infeasibility probability 
    $\beta_n \coloneq 1 - \Pr(\alpha_n > 0)$ (cf.~\cref{sec:t2i-exp-general}) in panels (B) and (C), respectively.}
    \label{fig:emp-t2i-group-photo-cost-discrep-profiles}
\end{figure}

\subsubsection{Geocoded persona generation}\label{sec:exp-persona-generation}

{\bf Experimental setup.}
We next consider a synthetic patient-persona generation task, motivated by constructing evaluation datasets for patient-facing digital health assistants and related conversational health systems
\citep{reichenpfader-denecke-2024-simulating,kyung2025patientsim,liao2024automatic,rashidian2025aiagents,cook2025virtual}.
We generate patient personas, with fields including 
an age band,
care setting, city, state, ZIP code, request type, and patient message.
The audited attribute is the CDC 2022 Social Vulnerability Index (SVI) quartile of
the generated ZIP code.\footnote{CDC/ATSDR SVI data downloads:
\url{https://svi.cdc.gov/dataDownloads/data-download.html}.}
We discretize the national overall SVI percentile into four quartiles, assign
invalid or unmatched ZIP codes to an auxiliary attribute with target mass zero, and
set the target distribution to be uniform over the four valid SVI quartiles.
We use 
\texttt{GPT-5.5} to generate personas.

{\bf M-sampling and post-processing methods.}
As in the text-to-image experiments, we distinguish between the method used to
generate outputs
and the post-processing rules applied thereafter to them. 
We consider two M-sampling methods.

\begin{enumerate}
    \item[(i)] \textit{Hard SVI prompting.}
    This direct hard-prompt baseline externally randomizes over the four
    requested SVI quartiles and queries the model with an attribute-specific prompt.
    For example, for a requested quartile, the prompt asks for a ZIP whose SVI percentile lies in the corresponding quartile range. 

    \item[(ii)] \textit{County-seeded prompting.}
    Here, we sample a
    county/state seed and ask the model to generate a persona located in or
    near that county. The prompt does not mention SVI, social vulnerability, or
    the target quartile. This is intended to broaden geographic support
    without asking the model to reason about the target SVI quartile.
\end{enumerate}

For each M-sampling method, we run RDC and the anytime RDC procedures of TA-RDC and CA-RDC.
Detailed prompts, the 
annotation procedure, and
example personas
are provided in \cref{app:exp:details-svi,app:exp:results-svi}.

{\bf Reported metrics.}
Following the protocol of \cref{sec:t2i-exp-general}, 
we first generate and annotate a fixed pool of \(1000\) personas for each M-sampling method.
We then report the three pool-level diagnostics of the on-support rate, target-attribute coverage, and compliance rate.
For each post-processing method, we also track
household and geographic diversity metrics.
In particular, for each set of $m$ returned personas,
we record the proportion of outputs whose ZIP code lies above the $75$th national percentile for a socioeconomic component (i.e., housing-cost burden, transportation access, and uninsurance rate), as well as the top-five-city share rate
(share of responses in the five most frequent cities).

Identical to the protocol of \cref{sec:exp-headshots-individuals},
we approximate the KL statistical alignment discrepancies by sampling with replacement from the frozen pool of generated personas.
We record the feasible-run discrepancy and
infeasibility probability 
for settings where there are personas whose annotated attribute lies outside the support of the target rates, $\mathrm{supp}(q)$.

In the main text, we consider output sizes of \(m=50\) personas and, as in \cref{sec:synth-exp}, trace out cost-discrepancy profiles by varying either the normalized tolerance level $d \coloneq \varepsilon/m \in [0,3]$ 
or the
capped-budget \(n \in [m,20m]\) for TA-RDC and CA-RDC, respectively.
We defer results for discrepancies under other divergences and
settings with other output sizes to
\cref{app:exp:results-svi}.

\begin{table}[t]
    \centering
    \scriptsize
    \setlength{\tabcolsep}{2.6pt}
    \renewcommand{\arraystretch}{1.15}
    \caption{\footnotesize
    SVI4 persona-generation results at \(m=50\).
    Pool-level diagnostics are reported once per M-sampling method, while the
    remaining columns report means and standard deviations over sets of returned
    persona outputs.
    The three socioeconomic rates give the proportions of returned personas,
    among those linked to the 2022 national ZCTA-level SVI data, whose ZCTAs
    lie above the 75th national percentile for the corresponding indicator.
    The top-five city--state share is the proportion belonging to the five
    most frequently represented city--state pairs within each returned set.
    Anytime RDC methods are selected by choosing the nearest hyperparameter
    setting (normalized tolerance \(d=\varepsilon/m\) or cap \(n\) for
    TA-RDC and CA-RDC, respectively) whose M-cost is halfway between direct
    sampling and RDC.
    }
    \label{tab:svi-main}
    \resizebox{\textwidth}{!}{%
    \begin{tabular}{@{}lccc@{\hspace{6pt}}llccccc@{}}
        \toprule
        \multicolumn{4}{c}{Pool-level diagnostics}
        &
        \multicolumn{7}{c}{Returned personas after post-processing}
        \\
        \cmidrule(lr){1-4}
        \cmidrule(lr){5-11}
        M-sampling method
        & \shortstack{On-support\\rate}
        & \shortstack{Target-attribute\\coverage}
        & \shortstack{Compliance\\rate}
        & \shortstack{Post-processing\\method}
        & Setting
        & \(\text{M-cost}/m\)
        & \shortstack{Housing-\\cost burden}
        & \shortstack{No-vehicle\\prevalence}
        & \shortstack{Uninsured\\prevalence}
        & \shortstack{Top-5-city--state\\share}
        \\
        \midrule

                \multirow{4}{*}{Hard SVI prompting}
        & \multirow{4}{*}{\(1.00\)}
        & \multirow{4}{*}{\(4/4\)}
        & \multirow{4}{*}{\(0.35\)}
        & RDC
        & --
        & \(21.43\;(7.86)\)
        & \(0.35\;(0.07)\)
        & \(0.39\;(0.07)\)
        & \(0.11\;(0.04)\)
        & \(0.42\;(0.06)\)
        \\

        & & & &
        TA-RDC
        & \(d=0.07\)
        & \(11.97\;(4.57)\)
        & \(0.41\;(0.07)\)
        & \(0.46\;(0.07)\)
        & \(0.12\;(0.05)\)
        & \(0.37\;(0.05)\)
        \\

        & & & &
        CA-RDC
        & \(n=600\)
        & \(11.74\;(1.02)\)
        & \(0.41\;(0.07)\)
        & \(0.46\;(0.07)\)
        & \(0.12\;(0.05)\)
        & \(0.38\;(0.05)\)
        \\

        & & & &
        DS
        & --
        & \(1.00\;(0.00)\)
        & \(0.60\;(0.07)\)
        & \(0.62\;(0.07)\)
        & \(0.20\;(0.06)\)
        & \(0.39\;(0.06)\)
        \\

        \cmidrule(lr){1-11}

        \multirow{4}{*}{County-seeded prompting}
        & \multirow{4}{*}{\(0.98\)}
        & \multirow{4}{*}{\(4/4\)}
        & \multirow{4}{*}{\(\text{--}\)}
        & RDC
        & --
        & \(1.74\;(0.41)\)
        & \(0.27\;(0.06)\)
        & \(0.31\;(0.06)\)
        & \(0.17\;(0.05)\)
        & \(0.15\;(0.03)\)
        \\

        & & & &
        TA-RDC
        & \(d=0.04\)
        & \(1.40\;(0.20)\)
        & \(0.30\;(0.06)\)
        & \(0.34\;(0.06)\)
        & \(0.18\;(0.05)\)
        & \(0.16\;(0.03)\)
        \\

        & & & &
        CA-RDC
        & \(n=70\)
        & \(1.37\;(0.07)\)
        & \(0.30\;(0.06)\)
        & \(0.34\;(0.07)\)
        & \(0.18\;(0.05)\)
        & \(0.16\;(0.03)\)
        \\

        & & & &
        DS
        & --
        & \(1.00\;(0.00)\)
        & \(0.32\;(0.07)\)
        & \(0.36\;(0.07)\)
        & \(0.19\;(0.06)\)
        & \(0.16\;(0.03)\)
        \\
        
        \bottomrule
    \end{tabular}%
    }
\end{table}

{\bf Main results.}
\textit{Hard compliance and geographic diversity.}
\cref{tab:svi-main} demonstrates that hard SVI prompting has a compliance rate of \(0.356\) and is thus not reliable.
Using these hard SVI prompted samples 
with direct-sampling 
also results in highly concentrated
socioeconomic factors.
For example, the proportion of returned personas associated with high
housing-cost burden is \(0.60\) under direct-sampling, compared with
\(0.35\)--\(0.41\) under RDC and the anytime RDC methods.
This demonstrates that aligning the target attributes also changes the distribution of correlated auxiliary features.
We provide 
diagnostics for hard prompt failures
in
\cref{app:exp:results-svi}.

\textit{Reductions in sampling discrepancy.}
County-seeded prompting produces personas whose annotated attribute lies outside of $\mathrm{supp}(q)$ -- 
of the 1000 generated personas, 18 were unmatched to an SVI quartile. 
Consequently, 
we report the feasible-run discrepancy and infeasibility probability $\widehat \beta_n$ 
(cf.~\cref{sec:t2i-exp-general}).

As can be seen in \cref{fig:emp-capped-budget-curve-svi},
both TA-RDC and CA-RDC reduce the discrepancies under either
persona-generation method.
However, 
the degree of reduction appears to heavily depend on the source rates.
For example, under hard SVI prompting, the observed quartile distribution is highly imbalanced despite having no observed off-support outputs (see, also, \cref{app:exp:results-svi}).
In this case, TA-RDC and CA-RDC both require a normalized M-cost of \(\approx 12 \) to achieve a KL discrepancy of \(<0.05\).

In contrast,
county-seeded prompting produces a more balanced distribution over valid SVI quartiles.
Consequently, TA-RDC can achieve a similar discrepancy 
as the anytime RDC methods under hard SVI prompting at a much smaller normalized M-cost of \(1.02\).
These results illustrate that post-processing model outputs \textit{without} hard prompting may actually improve the statistical alignment discrepancies over M-sampling methods that explicitly request certain features.

Furthermore, like with the group photo T2I experiments in \cref{sec:exp-group-photos}, the CA-RDC
infeasibility probability under county-seeded prompting decreases
rapidly with the capped-budget \(n\). 
At the moderately sized cap of $n=70$, CA-RDC can achieve a feasible-run normalized
KL discrepancy of \(0.013\), with a negligible infeasibility probability of \(\widehat{\beta}_{70}=3.8\times10^{-20}\).

\begin{figure}
    \centering
    \includegraphics[width=0.85\linewidth]{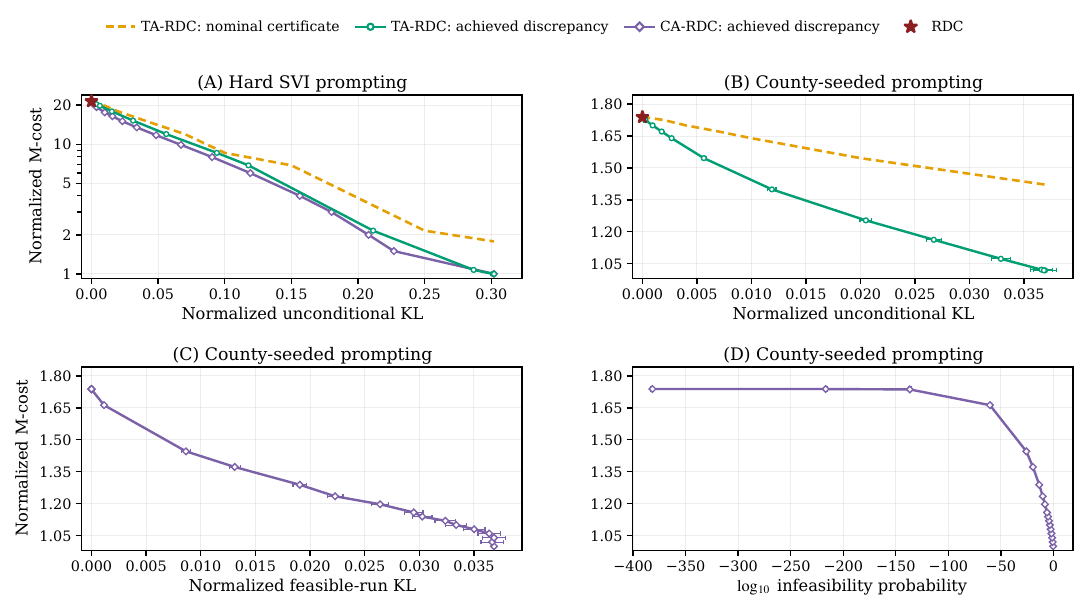}
    \caption{\footnotesize
    Various metrics for the SVI persona generation task with \(k=4\) target quartiles and \(m=50\) returned personas. 
    Panel (A) reports the KL discrepancy with hard SVI prompting for both TA-RDC and CA-RDC,
    since the annotated attributes of all 1000 generated personas lie inside \(\textrm{supp}(q)\).
    In contrast, 
    county-seeded prompting produces $18$ personas with attributes outside of \(\textrm{supp}(q)\) (cf.~\cref{tab:svi-main}), so panel (B) plots the cost-discrepancy profile of only 
    TA-RDC.
    For CA-RDC,
    we report the feasible-run discrepancies \(D_{\mathrm{KL}}(P_{\textrm{CA-RDC}}(\widehat A_n\mid\alpha_n>0)\|q^{\otimes m})\) and the infeasibility probability 
    $\beta_n \coloneq 1 - \Pr(\alpha_n > 0)$ (cf.~\cref{sec:t2i-exp-general}) in panels (C) and (D), respectively.
    }
    \label{fig:emp-capped-budget-curve-svi}
\end{figure}

\section{Discussion}\label{sec:discussion}

Statistical attribute alignment (SAA) provides a common formulation for alignment goals including fairness and synthetic data augmentation tasks.
Our results characterize the model-query cost of exactly or approximately achieving a target joint law on output attributes under post-processing black-box access.
The target distribution is supplied by the practitioner, and our post-processing methods
can adapt the attribute distribution of generators that have already undergone other forms of alignment.
In this way, SAA complements broader alignment methods in generative AI.
We discuss several directions that build on the SAA framework.

\emph{Dominance and admissibility.}
For a set of source rates \(\mathcal P\subseteq\Delta^{k-1}\)
and two post-processing algorithms
\(\mathrm{Alg}_m,\mathrm{Alg}'_m\in\mathfrak A_m\), we say that
\(\mathrm{Alg}'_m\) \emph{dominates} \(\mathrm{Alg}_m\) over
\(\mathcal P\) if
\begin{align}
    \mathcal C_m^{\mathrm{Alg}'_m}(p)
    &\leq
    \mathcal C_m^{\mathrm{Alg}_m}(p); \qquad 
    \Delta_f^{\mathrm{un}}(\mathrm{Alg}'_m;\nu;p)
    \leq
    \Delta_f^{\mathrm{un}}(\mathrm{Alg}_m;\nu;p)
    \label{eq:discrepancy-dominance}
\end{align}
for every \(p\in\mathcal P\), with at least one of these inequalities strict for at least one \(p\in\mathcal P\). 
A post-processing algorithm is then called 
\emph{(cost-discrepancy) admissible} over \(\mathcal P\) if it is not dominated by another algorithm (see, for example, classical decision-theoretic notions of dominance and admissibility in 
\citet{wald1950statistical} and \citet{lehmann1998theory}).

\cref{thm:rdc-m1-optimal} shows that the RDC algorithm is admissible over $\mathcal{P}=\textrm{int}(\Delta^{k-1})$ when $m=1$.
Furthermore, when the target joint law is $\nu=q^{\otimes m}$,
direct-sampling is also admissible, since no post-processing algorithm can dominate it at $p=q$; see \cref{fig:discussion-admissibility}.
Though the TA-RDC procedure of \cref{def:thresholded-anytime-rdc} 
is first-order optimal in the limit as $m\rightarrow \infty$,
it is not necessarily admissible 
over $\mathcal{P} = \Delta^{k-1}$, in general.
In particular, there are black-box post-processing algorithms that may strictly improve the unconditional discrepancy while achieving the same M-cost;
we give one such procedure in \cref{sec:algos-anytime-rdc}.

\begin{figure}[tbp]
    \centering
    \includegraphics[width=\linewidth]{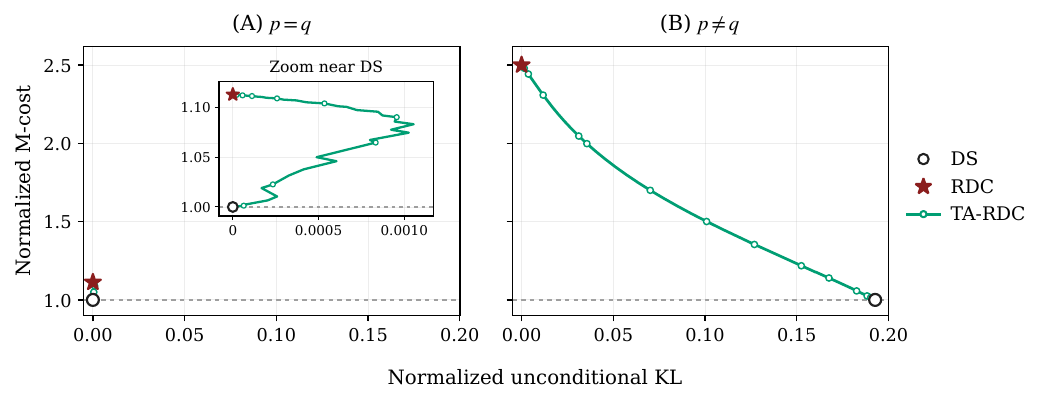}
    \caption{Source-dependent cost--discrepancy profiles for direct-sampling
    (DS), RDC, and TA-RDC, with $m=100$ and target joint law $\nu = q^{\otimes m}$ with uniform target rates $q=(0.5,0.5)$:
    (A) $p=q$; (B) $p=(0.8,0.2)$.
    In this small-scale setting, M-cost and unconditional KL discrepancy are evaluated explicitly through enumeration of count state spaces (cf.~\cref{app:metric-evaluation}).
    In (A), DS attains both zero discrepancy and minimum cost.
    For tolerance values $\varepsilon > m \log 2$, TA-RDC simplifies to DS, while for $\varepsilon = 0$, it reduces to RDC. In either of these two cases, it achieves zero unconditional discrepancy.
    Small irregularities for intermediate tolerances $0 < \varepsilon < m\log 2$ 
    reflect the discrete nature of 
    the count vectors used for the stopping time (cf.~\cref{def:thresholded-anytime-rdc}).
    In (B), TA-RDC reduces discrepancy relative to DS with additional M-cost.}
    \label{fig:discussion-admissibility}
\end{figure}

\emph{Relevance as genAI capabilities improve.}
As model query costs decline \citep[see, e.g.,][]{maslej2025artificial}, 
so does the cost of generating a larger batch of samples. 
This makes universally exact and approximate algorithms less costly, thus increasing the practical applicability of our approach.
Moreover, no single target distribution or input prompt 
is appropriate for every use case, so better models do not eliminate the need for a downstream method that can steer outputs toward an application specific target without retraining. 
In particular, a model previously aligned to one demographic composition or training population may require a different attribute distribution for a new application.
Since our method uses only black-box samples, it does not need to be redesigned when providers update their model or target specifications.

\emph{Noisy annotator.}
Suppose there is a latent true attribute $G$ associated with $Y \sim P(\cdot \mid x)$ with rates $\pi_i=\Pr(G=i)$.
If the annotator $\phi$ 
is noisy, $\phi(Y)$ is not necessarily always equal to the true attribute $G$.
One can model its error through a confusion rate matrix 
\[
C_{ij}=\Pr(G=i\mid A=j)
=
\frac{\pi_iK_{ij}}{\sum_{\ell}\pi_{\ell}K_{\ell j}}; \qquad K_{ij}=\Pr(A=j\mid G=i).
\]
\citep[see, e.g.][]{patrini2017making,lipton2018detecting}.

Given $C$, choose proxy target rates \(u^\star
    \in
    \arg\min_{u\in\Delta^{k-1}}
    D_f(Cu\|q),\)
and set $\bar q=Cu^\star$. 
If $q\in C\Delta^{k-1}$, we may choose $u^\star$ so that $\bar q=q$ 
and thus exact sampling for the true attribute is possible.
Otherwise, exact sampling
with a single output is impossible, and $\bar q$ is the
closest output attribute law under the chosen divergence. 

For general $m\ge 1$ and a target joint law $\nu = q^{\otimes m}$,
RDC run with
target rates 
$u^\star$ has true attribute law $\bar q^{\otimes m}$ unconditionally, so that it is universally exact if $q\in C\Delta^{k-1}$.
For approximate sampling, 
an annotation error aware 
anytime RDC method changes the rate of the exponential clock in
\cref{def:anytime-rdc} from $\nu(\mathcal{D}_t)$
to 
$(u^\star)^{\otimes m}(\mathcal D_t)$.
Its count-conditional output attribute law (cf.~\cref{eq:count-cond-output-label-law}) is then \(Q_{\vec c}^{u^\star}
    \coloneq 
    (u^\star)^{\otimes m}
    \bigl(\cdot\mid\Omega_m(\vec c)\bigr)\)
for intermediate times $t \ge m$
where $\Omega_m(\vec c)$ and $\vec c = \vec c(t)$
are the realized feasible support and source counts, respectively (cf.~\cref{eq:feasible-support}).
Since $C^{\otimes m}(u^\star)^{\otimes m}=\bar q^{\otimes m}$, the data
processing inequality gives
\[
    D_f\left(
    C^{\otimes m}Q_{\vec c}^{u^\star}
    \middle\|
    \bar q^{\otimes m}
    \right)
    =
    D_f\left(
    C^{\otimes m}Q_{\vec c}^{u^\star}
    \middle\|
    C^{\otimes m}(u^\star)^{\otimes m}
    \right)
    \leq
    D_f\left(
    Q_{\vec c}^{u^\star}
    \middle\|
    (u^\star)^{\otimes m}
    \right).
\]
Thus, minimizing the count-conditional discrepancy under the proxy target rates
controls
the distortion from the closest reachable product law.

{\small
\bibliography{refs}
}

\appendix
\onecolumn

\section{Appendix}

\subsection{Additional related work}
\label{app:arw}

{\bf Distributional and training-free alignment.}
The alignment literature spans objectives concerning instructions, preferences, safety, and values \citep{ji2025alignment}, as well as methods that act before, during, or after decoding without additional training \citep{pan2025trainingfree}.
Distributional constraints also arise in controlled generation: \citet{khalifa2021distributional} formulate pointwise and distributional constraints on text, characterize a target law that minimizes KL divergence from a pretrained language model, and train a new model to approximate that law.
Our setting instead concerns a fixed generator available only through sampling, with unknown source attribute probabilities.
We seek guarantees for the joint attribute law of a selected sample and characterize the number of model queries required to attain them.
Input-side controls, including conditional generation \citep{keskar2019ctrl} and learned soft prompts, may change the source distribution to which our procedures are applied; the access required by such controls depends on the method.

{\bf Synthetic data for training and inference.}
The statistical literature on synthetic microdata emphasizes the connection between the synthesis mechanism and valid downstream analysis \citep{rubin1993disclosure,raghunathan2003multiple}.
Recent surveys discuss generative augmentation for learning \citep{chen2024comprehensive,mumuni2024survey} and the assumptions needed for synthetic data to support statistical inference \citep{abdel2026harnessing}; applications in healthcare raise related questions of utility and distributional fidelity \citep{giuffre2023harnessing,lu2023machine}.
Repeated training on generated data can also distort a learned distribution \citep{shumailov2024collapse}.
SAA addresses a specified attribute mismatch in the resulting generator, including one previously trained toward a different target, rather than all sources of synthesis error.
Prediction-powered inference \citep{angelopoulos2023prediction} and synthetic-powered predictive and general inference \citep{bashari2025synthetic,bashari2025general} use real observations to support inference despite imperfect auxiliary predictions or synthetic data.
Our contribution is complementary to these inferential corrections: it concerns sample construction before downstream use.
Under a shared attribute-conditional model, \cref{sec:count-cond-discrep-synth-data-aug} makes this connection precise through a bound on bounded downstream criteria.

{\bf Distributional fairness as an application.}
Recent works \citep[see, e.g.,][]{choi2020fair, zhang2023iti, teo2023measuring} have adopted sampling-based or distributional fairness objectives, but they do not study the same black-box post-processing setting or prove similar optimality results established in our work.
We retain this application-specific perspective in the following review of approaches to fairness in generative models.

{\bf Fair synthetic data augmentation.}
A substantial early line of fairness-aware generative modeling predates modern diffusion models and treats the generator primarily as a mechanism for producing synthetic data that improves downstream fairness. 
\citet{xu2018fairgan} propose FairGAN, a GAN-based framework for learning a generator whose samples remain close to the original data distribution while reducing the dependence between generated records and protected attributes. The aim is not merely to generate plausible data, but to generate data on which downstream classifiers exhibit improved fairness when trained on the synthetic sample. 

\citet{xu2019fairgan+} extend this idea in FairGAN+ by coupling the generator with a classifier and additional adversarial components, so that the model jointly pursues fair data generation and fair downstream classification. \citet{sattigeri2019fairness} develop a related auxiliary-classifier GAN formulation that targets group fairness criteria such as demographic parity and equality of opportunity, and demonstrate the approach on both tabular and image-like data. \citet{amini2019uncovering} take a different generative route: rather than training a GAN to output fair synthetic data, they use a variational autoencoder to learn latent structure in the training distribution and then adaptively upweight underrepresented regions during model training, with applications to racial and gender bias in face detection. These works are important because they show that generative models can be used to repair biased datasets or improve fairness of downstream predictors. 
However, they are different from the black-box post-processing setting we pursue, as they are training-time or data-generation interventions and assume access to the training procedure, to a dataset to be modified, or to internal model components. 

{\bf Adversarial generative modeling for fairness.}
Another line of adversarial generative work explicitly uses information-theoretic objectives to reduce the dependence between sensitive attributes and generated or predicted quantities. \citet{chen2024information} propose an information-minimizing GAN for fair generation and classification, in which the generator and classifier are trained jointly so as to reduce mutual-information-type dependence involving sensitive attributes while preserving classification utility. Their method also introduces an ANOVA-based latent factor to mitigate accuracy loss, and the resulting framework is designed to improve downstream fair classification from generated data. This is related to our work in its use of generative modeling to address fairness, but the object of control is different. Information-minimizing GANs aim to remove sensitive-attribute information from learned representations or downstream predictions, whereas our target distribution may assign positive mass to protected attributes and asks for the returned samples to follow that target law. For example, in a text-to-image application, our goal may be to generate a batch whose gender, race, or age annotations follow a specified distribution, not to make those attributes uninformative or unrecoverable from the images.

{\bf Few-shot and conditional generative modeling for fairness.} Few-shot and conditional generative modeling provide another route to more balanced generation. \citet{sinha2021d2c} introduce D2C, a diffusion-decoding framework that combines an unconditional VAE, a learned diffusion prior over latent representations, and contrastive representations to enable few-shot conditional image generation. Because D2C can adapt an unconditional generator to new labels or manipulation constraints using a small labeled set, it can in principle be used to increase the number of samples from rare or desired attributes. 
Furthermore, 
\cite{zhang2023iti} develop an input-side procedure that learns and appends attribute-specific soft tokens to hard prompts.
This is also an approach that requires white-box access. 
Moreover, \cite{teo2024fairqueue} introduces a gray-box method that interweaves hard prompts and fairness-aware prompts during the decoding process.
\citet{shrestha2024fairrag} use retrieval augmentation and a learned module that projects reference images into the text-embedding space; this likewise requires more than black-box output sampling.

{\bf Fairness for T2I diffusion models.} 
For text-to-image diffusion models, several recent works intervene directly in the diffusion model or in the prompting mechanism. \citet{friedrich2023fair} propose Fair Diffusion, a deployment-time approach for steering text-to-image diffusion models by adding fairness-aware guidance so that generated images can be shifted toward desired proportions of identity groups. This work is important because it demonstrates that fairness interventions can be applied after pretraining and need not require full retraining of the text-to-image model. Nevertheless, Fair Diffusion still relies on access to the diffusion sampling process and modifies the generation dynamics themselves. 

\citet{de2023diffusionworldviewer} introduce DiffusionWorldViewer, an interactive system for exposing the ``worldview'' reflected by a text-to-image model and allowing users to edit generated outputs toward alternative demographic compositions. Their work is especially close in spirit to our user-specified target distribution: rather than assuming a single universal definition of fairness, DiffusionWorldViewer lets users inspect the demographic distribution produced by a model and then adjust prompts or editing controls to represent a desired perspective. This is consistent with our view that $q$ should be supplied by the user or application context rather than imposed universally. 

\citet{shen2023finetuning} also formulate fairness for text-to-image diffusion models as distributional alignment. They introduce a distributional alignment loss that steers generated attributes toward a user-defined target distribution, together with an adjusted direct fine-tuning method for optimizing losses defined on generated images. Empirically, their approach reduces gender, race, and intersectional biases for occupational prompts, and it can represent target distributions beyond exact parity, such as simultaneously debiasing some attributes while setting a nonuniform target for another. This is one of the closest source-side analogues of our target-distribution framework. However, it requires fine-tuning the diffusion model or its soft tokens and therefore assumes substantially more access than black-box sampling. 

\citet{choi2024fair} study fair sampling in diffusion models through an attribute-switching mechanism. Their method changes the diffusion sampling path by switching sensitive attributes during generation, and they provide theoretical and empirical evidence that this can produce fair synthetic data while preserving utility, without requiring additional model training. The fairness notion they use is connected to the $\epsilon$-fairness framework of \citet{feldman2015certifying}, where a representation is considered fair when the protected attribute cannot be accurately inferred from the remaining information. In the binary case discussed in this line of work, a predictor $f:\mathcal X\to\mathcal A$ must have sufficiently large balanced error, expressed through conditions such as $P(f(X)=0\mid A=1)+P(f(X)=1\mid A=0)>2\epsilon$. This criterion is substantively different from ours. Attribute switching and $\epsilon$-fairness aim to obscure or remove information about protected attributes, 
whereas our target-law sampling objective matches the output attribute distribution with a specified target law.
Moreover, the switching mechanism is tailored to diffusion sampling, while our algorithms are model-agnostic and apply to any generator from which one can obtain black-box M-samples.

{\bf SAA as a complementary intervention.} Altogether, these papers show that alignment for generative modeling in tasks like fairness and synthetic data augmentation
has been approached through 
data repair, adversarial training, weak supervision, latent reweighting, diffusion guidance, prompt or interface design, fine-tuning, and specialized diffusion samplers.
Our formulation of SAA isolates the downstream statistical problem faced by a user who can only query a black-box generator and post-process its outputs toward an application-specific attribute law.

This perspective allows us to prove optimality and lower bound
results that are independent of the architecture of the underlying generative model.
In particular, 
the source-rate-free and universal optimality results of our RDC and TA-RDC methods (cf.~\cref{thm:rdc-exact-main,thm:first-order-uni-opt-ta-rdc}) make our 
algorithms complementary to the methods
above:
the source rates $p$ may differ between various fairness intervention methods, but our algorithms, which do not assume knowledge of these rates, can still
provide exact or approximate statistical attribute alignment with first-order optimal model-query costs.

{\bf Fairness constraints on attribute proportions and sequential decision making.}
In the context of algorithmic fairness, 
there are some provable guarantees
for 
ensuring 
proportions of protected attributes (e.g., old/young)
obey some bounds \citep{cano2021fairness, hu2022achieving, alamdari2023remembering}
or that protected attributes appear relatively frequently 
(e.g., a particular attribute appears at least once for every $\beta$ images produced, where $\beta >0$ is a user-specified hyperparameter; see \cite{cheng2025runtime}). 
However, constraints on attribute proportions or frequencies need not ensure that
the returned attributes follow a specified joint distribution, as we consider in our work.
This distinction matters for downstream tasks such as pairwise audits,
whose results depend on the distribution of pairs of outputs \citep[see, e.g.,][]{klare2012face, grother2019face, grother2022face}; see, also, \cref{sec:exp-pairwise-audit}.

{\bf Exponential races and random-priority couplings.}
The exponential clock construction used in
\cref{def:anytime-rdc} is a classical random-priority device.
Let $\{E_a:a\in\mathcal S\}$ be independent with
$E_a\sim\operatorname{Exp}(w_a)$. For any fixed subset
$B\subseteq\mathcal S$ with positive total weight,
$\Pr(\arg\min_{a\in B}E_a=b)=w_b/\sum_{a\in B}w_a$.
Hence one clock field simultaneously gives a draw from the weights
normalized on every fixed subset; the same statement holds conditionally
for a random subset independent of the clocks. This competing-exponentials
rule underlies race-based event simulation, including the stochastic
simulation algorithm of \citet{gillespie1976general}. In discrete-choice
language, the winner follows Luce's proportional-choice rule
\citep{luce1959individual}, while ordering all clocks gives the
Plackett--Luce ranking law \citep{plackett1975analysis}. It is also
equivalent to the Gumbel-max trick: writing $E_a=Z_a/w_a$ with
$Z_a\overset{\mathrm{i.i.d.}}{\sim}\operatorname{Exp}(1)$ gives
$\arg\min_aE_a=\arg\max_a\{\log w_a+G_a\}$, where
$G_a=-\log Z_a$ are i.i.d. standard Gumbel variables
\citep{yellott1977relationship}.

In machine learning, this representation is part of the broader
perturb-and-MAP framework, which uses random perturbations to reduce
sampling or partition-function calculations to optimization
\citep{papandreou2011perturb,hazan2012partition}. For our purposes, a
particularly close connection is the Gumbel-process framework of
\citet{maddison2014astar,maddison2016poisson}, which extends Gumbel-max to
general spaces and couples maximizers over many measurable subsets.
Equivalent exponential or Gumbel keys appear in weighted sampling without
replacement and reservoir sampling \citep{efraimidis2006weighted}; taking
the top $k$ perturbed scores gives Gumbel-top-$k$ sampling and the
associated Plackett--Luce order \citep{kool2019stochastic}. Related
persistent-randomness couplings arise in consistent weighted sampling and
weighted MinHash \citep{ioffe2010improved}, and in counterfactual models
that share Gumbel variables across interventions
\citep{oberst2019counterfactual}.

Our use differs in the object being coupled. The target law $\nu$ is fixed,
while the source stream generates a nested random family of feasible
supports $\{\mathcal{D}_t\}_{t\ge m} \coloneq \{\Omega_t \backslash\Omega_{t-1}\}_{t\ge m}$ (cf.~\cref{eq:feasible-support}), 
independently of the clock field. 
Restricting one
target-weighted priority field to $\Omega_t$ yields the count-conditional output attribute law 
$\nu(\cdot\mid\Omega_t)$ at every deterministic time $t \ge m$. 
Once the global minimum-clock target sequence becomes feasible,
it remains selected and has law $\nu$, giving the RDC stopping time; see
\cref{thm:opt-count-cond-f-curve}. 
In our setting,
the clocks therefore combine a classical priority
construction with the feasibility constraints imposed by a coupon-collector type process, 
yielding an optimal count-conditional distortion frontier.

\subsection{Additional algorithms for exact sampling}
\label{app:alg-exact-sampling-uncond-discrep}

\subsubsection{Universally exact algorithms with known source rates}\label{app:ors}

In this section, 
we state
the traditional oracle sampler introduced in \cite{von195113} that assumes knowledge of the source rates $p$ and connect it with our problem formalization. 
For the target joint law $\nu = q^{\otimes m}$, 
we demonstrate that oracle rejection sampling is universally exact and achieves the minimum possible M-cost among all universally exact algorithms.

\begin{definition}[Oracle rejection sampler]\label{def:oracle-rs}
Fix an input \(x\) and a generative model $P$ from which we can sample outputs via
$P(\cdot\mid x)$. 
Assume the support condition $\textrm{supp}(q) \subseteq \textrm{supp}(p)$.
Define
\[
    c(x) \coloneqq \min_{i\in \textrm{supp}(q)}\frac{p_i}{q_i}
    =\left(\max_{i\in \textrm{supp}(q)}\frac{q_i}{p_i}\right)^{-1}.
\]

The oracle sampler proceeds as follows.
Repeat until acceptance:
\begin{enumerate}
    \item Query an M-sample \(Y\sim P(\cdot\mid x)\) and compute its attribute \(I=\phi(Y)\).
    \item Draw an independent E-sample \(U\sim\mathrm{Unif}[0,1]\).
    \item Accept \(Y\) if
    \[
        U \le r_I,
        \qquad\text{where}\qquad
        r_i \coloneqq
        \begin{cases}
            c(x)\cdot \dfrac{q_i}{p_i}, & i\in \textrm{supp}(q),\\[6pt]
            0, & i\notin \textrm{supp}(q).
        \end{cases}
    \]
    Otherwise, reject and repeat.
\end{enumerate}
To produce \(m\) outputs, the procedure is run independently until \(m\) samples
have been accepted.
\end{definition}

Each iteration uses one model query and succeeds with probability \(c(x)\).
Hence, the number of model queries required per accepted sample is geometric with mean
\(
    \frac{1}{c(x)}.
\)
and, consequently, the M-cost to produce $m$ samples is
\begin{equation}\label{eq:m-cost-ors}
    \mathcal{C}_m^{\mathrm{OracleRS}}(p,q)
    = m\cdot \frac{1}{c(x)}
    = m\cdot \max_{i\in \textrm{supp}(q)}\frac{q_i}{p_i}.
\end{equation}

Furthermore, the oracle rejection sampler is universally exact since, for an arbitrary source rate $p \in\Delta^{k-1}$, a single proposal is accepted with probability
\[
    \Pr(\text{accept}\mid x)
    = \sum_{i\in\textrm{supp}(q)} p_i \cdot r_i
    = \sum_{i\in\textrm{supp}(q)} p_i\cdot \frac{c(x)q_i}{p_i}
    = c(x).
\]
Moreover, conditioning on acceptance, the attribute distribution of the output satisfies
\[
    \Pr(\phi(Y)=i\mid x)
    = \frac{p_{i}\cdot r_i}{\sum_{j\in\textrm{supp}(q)} p_j \cdot r_j}
    = q_{i},
    \qquad \forall i\in\textrm{supp}(q).
\]
Thus, the oracle sampler sequentially produces i.i.d. outputs whose attributes follow the marginal law $q$ of the target joint law $\nu = q^{\otimes m}$. 

The oracle rejection sampler adaptively 
stops after a random number of proposals, so that it lies in $\mathfrak{A}_{m}$.
The next result 
records the optimality of the oracle sampler. 

\begin{proposition}[Optimality for universal sampling with known source rates]\label{prop:opt-perfect-sampling-fairness-oracle}
    Assume the support condition $\textrm{supp}(q) \subseteq \textrm{supp}(p)$ and denote the pointwise M-cost of a black-box post-processing algorithm $\textrm{Alg}_m\in\mathfrak A_m$ as $\mathcal{C}_m^{\textrm{Alg}_m}(p,q) = \E_p[N]$.
    Then,
    the pointwise optimal universal cost satisfies
    \[ 
    \inf_{\textrm{Alg}_m \in \mathfrak{U}_m(q^{\otimes m})} \mathcal{C}_m^{\textrm{Alg}_m}(p,q)
    =
    m\cdot \max_{i\in\operatorname{supp}(p)}\frac{q_i}{p_i}.
    \]
    The infimum is achieved by the
    oracle rejection sampler in \cref{def:oracle-rs}.
\end{proposition}

The result immediately follows from \cref{thm:universal-exact-m-cost-lower} (note that \(\sum_{j=1}^m \nu_j(i) = m\cdot q_i\) for $\nu = q^{\otimes m}$), 
and the M-cost of the oracle rejection sampler recorded in \cref{eq:m-cost-ors}.

\subsubsection{Universally exact source-rate-free algorithms}\label{app:unconditional-count-tree-lp}

Every post-processing algorithm 
returns attribute sequences from the feasible support $\Omega_m(\vec C(t))$ 
of \cref{eq:feasible-support} after $t\ge m$ M-samples.
For the purpose of minimizing M-cost, it suffices to characterize the algorithm's behavior through its count-conditional output count law (cf.~\cref{eq:count-cond-output-count-law}), which 
is the probability law on count vectors $\vec\ell \in \mathcal L_m
    \coloneqq
    \left\{
    \vec\ell\in\mathbb N_{\geq 0}^k:
    \sum_{i=1}^k\ell_i=m
    \right\}$ 
    induced by the 
    count-conditional output \emph{attribute} law of \cref{eq:count-cond-output-label-law}.
Indeed, given 
returned counts $\vec \ell \in \mathcal{L}_m$,
sampling the target joint law $\nu$ conditional on these returned counts 
$\widetilde A \sim \nu(\cdot \mid \{\vec L = \vec \ell\})$
preserves feasibility and M-cost, 
where  $L_i=\sum_{j=1}^m1\{\widetilde A_j=i\}$ are the attribute counts of the output attribute sequence.

Consequently, we consider variables $x_{\vec c,\vec \ell}$ which are
jointly indexed by the source and
returned count vectors 
\[
\left\{
    \vec c(t)\in\mathbb N_{\geq0}^k:
    \sum_{i=1}^kc_i=t
\right\}
\quad\text{and}\quad
\left\{
    \vec\ell\in \mathcal{L}_m:
    \ell_i\leq c_i(t),\ \forall i\in[k]
\right\}.
\]
At each source count vector $\vec c$, a post-processing algorithm can also elect to continue M-sampling and not return any outputs.
The 
variables $x_{\vec c, \vec \ell}$ and continuation variables $y_{\vec c}$ for $t \ge m$
then constitute the primary variables of an optimization problem.
Let
\begin{align*}
    \rho_\nu(\vec\ell)
\coloneqq
\sum_{\substack{\vec a\in\mathcal A^m:
\sum_{r=1}^m1\{a_r=i\}=\ell_i,\ \forall i\in[k]}}
\nu(\vec a), \qquad \vec \ell \in \mathcal{L}_m
\end{align*}
denote the law over the count vectors induced by the target joint law $\nu$. Additionally, abbreviate 
$p^{\vec c} = \prod_{i=1}^k p_i^{c_i}$.
Then, at a designated source rate \(p_0\), 
the minimum universally exact
source-rate-free M-cost is the value of the following linear program
\begin{align}
    \inf_{\{x_{\vec c,\vec\ell},y_{\vec c}\geq0\}}
    \quad&
    \sum_{\vec\ell\in\mathcal L_m}
    \sum_{\vec c\geq\vec\ell}
    |\vec c|x_{\vec c,\vec\ell}p_0^{\vec c}
    \notag \\
    \text{s.t.}\quad&
    y_{\vec 0}=1,
    \notag \\
    &
    y_{\vec c}
    +
    \sum_{\substack{\vec\ell\in\mathcal L_m: \vec \ell \leq \vec c}}
    x_{\vec c,\vec\ell}
    =
    \sum_{i:c_i>0} y_{\vec c-\vec e_i},
    \qquad \vec c\neq\vec 0,
    \notag \\
    &
    \sum_{\vec c\geq\vec\ell} 
    x_{\vec c,\vec\ell} \cdot p^{\vec c} 
    =
    \rho_\nu(\vec\ell),
    \qquad
    \vec\ell\in\mathcal L_m,\quad
    p\in\Delta^{k-1}: \operatorname{supp}(\nu)\subseteq(\operatorname{supp}(p))^m.
    \label{eq:lp-universal-exactness}
\end{align}

\cref{eq:lp-universal-exactness} imposes universal exactness through
a global constraint on a continuum of source rates.
Furthermore, for universal exactness, the stopping time cannot have a deterministic
finite upper bound.

\begin{theorem}[No capped-budget universal exactness]\label{thm:no-capped-budget-uni-exactness}
    Fix $m\geq1$, $k\geq2$, and a target joint law $\nu$. Suppose a source-rate-free algorithm $\textrm{Alg}_m \in \mathfrak{A}_m^o$
    satisfies $N\leq n$ almost surely for some deterministic $n<\infty$. Then the set of interior source rates $p$ for which
    \[
    P_{\textrm{Alg}_m}(\widetilde{A} \mid x)=\nu
    \]
    has Lebesgue measure zero. 
    In particular, an algorithm with a capped-budget $n \in \N$ fails to be exact for almost every interior source rate $p$.
\end{theorem}

For the proof, see \cref{prf:no-capped-budget-uni-exactness}.
Consequently, the primal program has infinitely many variables, with their number growing combinatorially with the number of M-samples.
Together with the continuum of globally coupled constraints in
\eqref{eq:lp-universal-exactness}, this makes the program difficult to
solve, and standard finite-state constrained Markov decision process
methods do not directly apply
\citep[see, e.g.,][]{altman2021constrained}.

\subsection{Additional algorithms for approximate sampling}\label{app:alg-approx-sampling-uncond-discrep}

We consider the following heuristic solution 
to 
finding universally approximate source-rate-free
procedures that minimize the M-cost (cf.~\cref{def:uni-source-rate-free-M-cost}) under a target joint law $\nu = q^{\otimes m}$.
First, estimate the source rates on an adaptively sized pilot sample, then
freeze the estimates, and finally run ordinary rejection sampling on freshly drawn
M-samples using the estimated source rates.
We call this procedure
\emph{approximate rejection sampling}
(ARS).

\emph{Anytime confidence set and stopping time rule.}
To estimate the source rates, 
we construct anytime-valid multinomial confidence sequences as developed in
\citet{lindon2022anytime}
and use 
a relative width stopping time rule 
\citep[see, e.g.,][]{chow1965asymptotic,flegal2015relative}.

More formally, let \(C_i(t)=\sum_{s=1}^t 1\{\phi(Y_s)=i\}\) be the source count process (cf.~\cref{eq:feasible-support}) 
and
\(\widehat p_i(t)=C_i(t)/t\) be the empirical plug-in estimates for the source-rates.  
Following \citet{lindon2022anytime},
for \(u\in\Delta^{k-1}\), define the
confidence set 
\begin{equation}
    \mathcal P_t(\delta_{\rm CS})
    =
    \left\{u\in\Delta^{k-1}:
    \mathcal E_t(u)<\delta_{\rm CS}^{-1}\right\},
    \qquad
    \mathcal E_t(u)
    =
    \frac{B(\vec 1+\vec C(t))}{B(\vec 1)}
    \prod_{i=1}^k u_i^{-C_i(t)},
    \label{eq:ars-lindon-malek-cs}
\end{equation}
where \(B\) is the multivariate beta function. In particular, the anytime guarantee holds
\(\Pr_p\{p\in\mathcal P_t(\delta_{\rm CS})\ \forall t\}\ge
1-\delta_{\rm CS}\). 
Let
\begin{equation*}
    L_i(t)=\inf_{u\in\mathcal P_t(\delta_{\rm CS})}u_i,
    \qquad
    U_i(t)=\sup_{u\in\mathcal P_t(\delta_{\rm CS})}u_i
\end{equation*}
be its coordinate projections.
Then, 
given a deterministic set of look times
\(\mathcal G\), the pilot stopping time is
\begin{equation*}
    \tau_\gamma
    =
    \inf\left\{t\in\mathcal G:
    \max_{i\in[k]}
    \frac{U_i(t)-L_i(t)}{2\widehat p_i(t)}
    \leq\gamma\right\},
\end{equation*}
with the convention that the maximum is \(+\infty\) until every empirical
rate is positive. 
Thus, 
\(\gamma\) is a
tuning parameter, with smaller values corresponding to more precise source-rate
estimates and pilot samples.

\emph{Approximate rejection sampling on fresh M-samples.} At \(\tau_\gamma\), freeze the empirical plug-in estimates for the source rates
\(\widehat p=\widehat p(\tau_\gamma)\) and set the empirical rejection sampling probabilities
\begin{equation}
    \widehat c
    =\min_{i\in\operatorname{supp}(q)}\frac{\widehat p_i}{q_i},
    \qquad
    a_i
    =
    \begin{cases}
        \widehat c\,q_i/\widehat p_i,&q_i>0,\\
        0,&q_i=0.
    \end{cases}
    \label{eq:ars-frozen-acceptance}
\end{equation}
The algorithm discards the pilot samples and performs rejection sampling on freshly drawn M-samples.
Like the oracle rejection sampler (cf.~\cref{app:ors}),
ARS independently accepts a proposal of attribute \(i\) with probability
\(a_i\) and stops after \(m\) acceptances.  

Conditional on the stopped pilot
\(\mathcal F_{\tau_\gamma}\), the proposal acceptance probability and the
output attribute law of the accepted samples 
are thus
\begin{equation}
    \lambda
    =\sum_{i=1}^k p_i a_i,
    \qquad
    s_i
    =\Pr_p\{\widetilde A_j=i\mid\mathcal F_{\tau_\gamma}\}
    =\frac{p_i a_i}{\lambda}.
    \label{eq:ars-conditional-law}
\end{equation}
In other words, conditional on \(\mathcal F_{\tau_\gamma}\), 
the attributes of the accepted outputs are i.i.d. draws from \(s\), so that
their count vector is \(\Mult(m,s)\). 
If
\(\widehat p=p\), then \(s=q\) and the procedure reduces to oracle rejection
sampling.

Details on the evaluation of the M-cost and statistical alignment discrepancy of ARS 
are deferred to \cref{app:metric-evaluation-rs}. 
Numerical results for a synthetic experiment can be found in 
\cref{app:exp:results-synth-exp}.

\subsection{Count-conditional discrepancy and synthetic data augmentation}\label{sec:count-cond-discrep-synth-data-aug}

The count-conditional discrepancy is also relevant in synthetic data augmentation, where a downstream analysis is performed using a single augmented dataset.
Suppose the model and target distributions
satisfy a label shift condition on the attributes:
\[
    A\sim\operatorname{Categorical}(p),
    \qquad
    A^{\mathrm{tar}}\sim\operatorname{Categorical}(q),
    \qquad
    Y\mid A=a \sim Y^{\mathrm{tar}}\mid A^{\mathrm{tar}}=a \sim K_a,
\]
where \(K_a\) is the common attribute-conditional law for attribute \(a\in\mathcal A\); see, for example, \citet{lipton2018detecting}.
Here, the target distribution describes the population that the synthetic data is intended to represent.
Let
\(\vec Y^{\mathrm{tar}}=(Y_1^{\mathrm{tar}},\ldots,Y_{n_{\mathrm{tar}}}^{\mathrm{tar}})\)
be an ideal dataset obtained by drawing
\[
\vec A^{\textrm{tar}}\sim q^{\otimes n_{\mathrm{tar}}},
\qquad
Y_j^{\textrm{tar}}\mid A_j^{\textrm{tar}}=a\sim K_a,
\quad j\in[n_{\mathrm{tar}}],
\]
and let
\(\vec Y^\star=(Y_1^\star,\ldots,Y_m^\star)\)
be an ideal augmentation sample obtained by drawing an attribute sequence $\vec A^\star \sim q^{\otimes m}$ with the same attribute-conditional laws $K_a$.

Likewise, conditional on \(\vec C\), let
\(\widetilde{Y}=(\widetilde Y_1,\ldots,\widetilde Y_m)\)
be the returned synthetic sample obtained from
\(\widetilde{A}\sim P_{\mathrm{Alg}_m}^{\vec C}\)
(cf.~\cref{eq:count-cond-output-label-law}) and the same attribute-conditional laws \(K_a\).
We consider post-processing algorithms whose stopping time and output selection rule
depend only on annotated attributes and independent E-sampling, with returned outputs selected uniformly within each attribute, as in RDC and its anytime variants (cf.~\cref{alg:rdc,alg:anytime-rdc}).
These procedures preserve the product of the attribute-conditional laws given the selected attributes and source counts $\vec C$.

Then, for every bounded downstream criterion
\(G:\mathcal Y^{n_{\mathrm{tar}}+m}\to[0,1]\),
the
data processing inequality 
through the common attribute-conditional laws gives, for each count vector $\vec c$ of positive probability,
\begin{equation*}%
\left|
\mathbb E\left[G(\vec Y^{\mathrm{tar}},\widetilde Y)\mid\vec C=\vec c\right]
-\mathbb E\left[G(\vec Y^{\mathrm{tar}},\vec Y^\star)\right]
\right|
\leq
\mathrm{TV}\left(P_{\mathrm{Alg}_m}^{\vec c},q^{\otimes m}\right).
\end{equation*}
Averaging the absolute conditional difference over $\vec C$ therefore bounds it by the count-conditional discrepancy. In particular,
\begin{align}
\left|
\mathbb E\left[
G(\vec Y^{\mathrm{tar}},\widetilde{ Y})
\right]
-
\mathbb E\left[
G(\vec Y^{\mathrm{tar}},\vec Y^\star)
\right]
\right|
&\le
\mathbb E_{\vec C}\left[
\mathrm{TV}\left(
P_{\mathrm{Alg}_m}^{\vec C},
q^{\otimes m}
\right)
\right]
\equiv 
\Delta_{\mathrm{TV}}^{\vec C}
(\mathrm{Alg}_m;q^{\otimes m};p),
\label{eq:count-conditional-augmentation-control}
\end{align}

More specifically, write \(\widetilde n=n_{\mathrm{tar}}+m\), and suppose that replacing one observation in the augmented dataset changes \(G\) by at most \(\gamma_{\widetilde n}\).
A maximal-coupling argument then strengthens \eqref{eq:count-conditional-augmentation-control} to
\[
    \left|
    \mathbb E\left[
    G(\vec Y^{\mathrm{tar}},\widetilde{ Y})
    \right]
    -
    \mathbb E\left[
    G(\vec Y^{\mathrm{tar}},\vec Y^\star)
    \right]
    \right|
    \le
    \bigl(m\gamma_{\widetilde n}\wedge1\bigr)\cdot
    \Delta_{\mathrm{TV}}^{\vec C}
    (\mathrm{Alg}_m;q^{\otimes m};p).
\]
For example, an average of fixed, 
observation wise 
losses taking values in $[0,1]$ has
\(\gamma_{\widetilde n}=1/\widetilde n\), yielding the factor
\(m/(n_{\mathrm{tar}}+m)\).
Analogous distributional mismatch terms arise in stability bounds for generative data augmentation \citep[see, e.g., Theorem~3.1 of][]{zheng2023toward}.

\subsection{Additional algorithms for anytime RDC}\label{sec:algos-anytime-rdc}

In this section, 
we discuss how to efficiently compute the target mass of the feasible support $\alpha_t \coloneq \nu(\Omega_m(\vec C(t)))$ (cf.~\cref{def:thresholded-anytime-rdc}) under the target joint law $\nu = q^{\otimes m}$.

We then discuss a variant of the TA-RDC algorithm (cf.~\cref{sec:anytime-rdc}), 
called 
\emph{RM (randomized-mass RDC)}, 
that modifies 
the output attribute law away from one that minimizes the count-conditional discrepancy.
On a high level, RM constructs a stopping time rule whose survival probability function decreases linearly with the target mass of the feasible support $\alpha_t$.
This results in stopping times that are stochastically no greater than that of TA-RDC.
Under mild conditions,
it provably reduces the M-cost while preserving TA-RDC’s universal unconditional discrepancy bound.

Finally, we provide first-order optimal frontiers under TV and Hellinger distance, and discuss a simple randomized RDC/direct-sampling procedure that attains them.

{\bf Computing the target mass of the feasible support.}
When $\nu = q^{\otimes m}$, 
the target mass of the feasible support set $\nu (\Omega)$ (cf.~\cref{thm:opt-count-cond-f-curve} in \cref{sec:approx-sampling-fairness})
is exactly that of a multinomial c.d.f.,
\[
    \alpha(\vec c)\equiv \nu(\Omega) = q^{\otimes m}(\Omega)
    =
    \Pr_{\vec L\sim\mathrm{Mult}(m,q)}
    \left(L_i\le c_i,~\forall i\in[k]\right).
\]

Computing multinomial c.d.f.s has been extensively studied in computational statistics \citep[see, e.g.,][and references therein]{levin1981representation, lebrun2013efficient, hayter2014recursive}.
Inspired by these works, we use recursion identities of multinomial 
c.d.f.s as in \cite{hayter2014recursive} to compute these probabilities,
and 
provide the full algorithm in \cref{alg:multinomial-c.d.f.} for completeness.

\begin{algorithm}[tbp]
\caption{\footnotesize \textsc{MultinomialCDF}$(h,q,\vec r)$: (upper-rectangular) Multinomial c.d.f. probability via \cite{hayter2014recursive} recursion}
\label{alg:multinomial-c.d.f.}
\scriptsize
\begin{algorithmic}[1]
\REQUIRE Integer $h\ge 0$, probability vector $q\in\Delta^{k-1}$, capacity vector $\vec r\in\mathbb{Z}_{\ge 0}^k$.
\ENSURE $P=\Pr[L_i\le r_i~\forall i\in[k]]$ for $L\sim\mathrm{Mult}(h,q)$.

\vspace{0.5em}
\IF{$\exists i:r_i<0$}
    \STATE \textbf{return} $0$.
\ENDIF
\IF{$h=0$}
    \STATE \textbf{return} $1$.
\ENDIF

\STATE Let $S\leftarrow\{i\in[k]:q_i>0\}$.
\STATE Restrict $q\leftarrow(q_i)_{i\in S}$ and $\vec r\leftarrow(r_i)_{i\in S}$, and reset $k\leftarrow |S|$.
\STATE \COMMENT{Coordinates with $q_i=0$ have $L_i=0$ almost surely and do not affect the multinomial c.d.f.}

\IF{$\sum_{i=1}^k r_i<h$}
    \STATE \textbf{return} $0$.
\ENDIF
\IF{$k=1$}
    \STATE \textbf{return} $\mathbf{1}\{r_1\ge h\}$.
\ENDIF
\STATE Replace each $r_i$ by $\min\{r_i,h\}$.

\vspace{0.5em}
\IF{$k=2$}
    \STATE Set $a\leftarrow \max\{0,h-r_2\}$ and $b\leftarrow \min\{h,r_1\}$.
    \STATE \textbf{return}
    \( \sum_{x=a}^{b}
        \binom{h}{x}q_1^xq_2^{h-x}\).
\ENDIF

\vspace{0.5em}
\STATE Define prefix and suffix capacity sums
\( B_i\coloneqq\sum_{j=1}^i r_j\) 
and \(R_i\coloneqq\sum_{j=i}^k r_j\).
\STATE For $i=1,\dots,k-1$, define the feasible domain
\[
    \mathcal D_i
    \coloneqq
    \left\{
    w\in\mathbb Z:
    \max\{0,h-R_{i+1}\}\le w\le \min\{h,B_i\}
    \right\}.
\]

\vspace{0.5em}
\STATE Initialize $g_{k-1}(w)\leftarrow 1/(h-w)!$ for each $w\in\mathcal D_{k-1}$.
\FOR{$i=k-2$ down to $1$}
    \STATE For each $w\in\mathcal D_i$, set
    \[
    g_i(w)\leftarrow
    \sum_{j=\max\{h-R_{i+2},w\}}^{\min\{h,w+r_{i+1}\}}
    \left(\frac{q_{i+1}}{q_{i+2}}\right)^j
    \frac{1}{(j-w)!}
    g_{i+1}(j).
    \]
\ENDFOR

\vspace{0.5em}
\STATE Compute
\[
    P
    \leftarrow
    h!q_k^h
    \sum_{j\in\mathcal D_1}
    \left(\frac{q_1}{q_2}\right)^j
    \frac{1}{j!}
    g_1(j).
\]
\STATE \textbf{return} $P$.
\end{algorithmic}
\end{algorithm}

{\bf M-cost improvements over TA-RDC at a fixed discrepancy tolerance.}
We provide a universally $(f,\varepsilon)$-approximate algorithm that has the same unconditional discrepancy bound as TA-RDC (cf.~\cref{def:thresholded-anytime-rdc}) 
while having a potentially smaller M-cost.

\begin{definition}[Randomized-mass RDC (RM)]\label{def:randomized-mass-rdc}
Fix a target joint law \(\nu\) on \([k]^m\). After \(t\geq m\) M-samples, let \(\vec C(t)\) be the accumulated source-count vector (cf.~\cref{eq:feasible-support}) 
and set
\[
    \Omega_t\coloneqq\Omega_m\bigl(\vec C(t)\bigr),
    \qquad
    \Omega_{m-1}\coloneqq\varnothing,
    \qquad
    \mathcal D_t\coloneqq\Omega_t\setminus\Omega_{t-1},
\]
where \(\Omega_m\bigl(\vec C(t)\bigr)\) is the feasible support
(cf.~\cref{eq:feasible-support}).
Furthermore, fix a tolerance $\varepsilon\ge0$ and let
\[
\theta_{f,\varepsilon}
\coloneqq\inf\{a\in[0,1]:\Psi_f(a)\le\varepsilon\}.
\]

If $\theta_{f,\varepsilon}>0$, independently E-sample
$U\sim\operatorname{Unif}(0,1)$ and stop at
\begin{equation}\label{eq:rm-stopping-time}
N_{f,\varepsilon}^{\mathrm{RM}}
\coloneqq\inf\{t\ge m:\alpha_t\ge\theta_{f,\varepsilon}\cdot U\},
\end{equation}
where \(\alpha_t = \nu\bigl(\Omega_m(\vec C(t))\bigr)\) is the target mass of the feasible support.
At this time, E-sample output attributes from
$\vec U\sim\nu(\cdot\mid\mathcal D_{N_{f,\varepsilon}^{\mathrm{RM}}})$.
If $\theta_{f,\varepsilon}=0$, return the first $m$ queried outputs instead.
\end{definition}

In general, RM may use any stopping time rule with conditional survival probability
\(\Pr\left(
N_{f,\varepsilon}^{\mathrm{RM}}>t
\middle|\mathcal G_t
\right)
=
{(\theta_{f,\varepsilon}-\alpha_t)_+}/
{\theta_{f,\varepsilon}}\).
For $\theta_{f,\varepsilon}>0$, the likelihood ratio of RM's output attribute law 
conditional on the entire source-count path is at most $1/\theta_{f,\varepsilon}$.
Consequently, convexity of the $f$-divergence and averaging over source paths gives,
for every interior source rate $p$,
\[
    \Delta_f^{\mathrm{un}}(\mathrm{RM};\nu;p)
    \le\Psi_f(\theta_{f,\varepsilon})\le\varepsilon.
\]
Moreover, 
for every $t\ge m$,
\[
\Pr\left(
N_{f,\varepsilon}^{\mathrm{RM}}>t\mid\mathcal G_t
\right)
\leq
1\{\alpha_t<\theta_{f,\varepsilon}\}(1-\alpha_t)
=
\Pr\left(
N_{f,\varepsilon}^{\mathrm{TA\text{-}RDC}}>t
\mid\mathcal G_t
\right).
\]
Thus, RM is a universally $(f,\varepsilon)$-approximate algorithm whose M-cost is no greater than that of TA-RDC. 

The M-cost of RM is strictly less whenever
$0<\theta_{f,\varepsilon}<1$ and
$\Pr_p(0<\alpha_t<\theta_{f,\varepsilon})>0$ for some $t\ge m$.
For example, if $\nu=q^{\otimes m}$ with $q$ uniform on $[k]$,
$k\ge2$, then $\alpha_m=k^{-m}$ for the source count process with 
$m$ identical attributes.
Thus, the improvement is strict whenever
$k^{-m}<\theta_{f,\varepsilon}<1$; for KL, 
$\theta_{f,\varepsilon}=e^{-\varepsilon}$,
so this holds for
$0<\varepsilon<m\log k$.

{\bf First-order optimal frontiers and algorithms under other divergences.}
First-order optimal frontiers and procedures under a target joint law $\nu = q^{\otimes m}$
can be characterized for divergences other than KL given in \cref{thm:first-order-uni-opt-ta-rdc}.
For the $\chi^2$ divergence,
\(\chi^2(P\|Q)\coloneqq\sum_x(P(x)-Q(x))^2/Q(x)\), 
every $(\chi^2,\varepsilon)$-approximate algorithm is $(\textrm{KL},\log(1+\varepsilon))$-approximate since
\(D_{\mathrm{KL}}(P\|Q)\le\log\{1+\chi^2(P\|Q)\}\).
Thus, for a sequence of tolerances \(\log(1+\varepsilon_m)/m\to d\), \[
    \frac{\mathcal C_{m,\varepsilon_m}^{\mathrm{TA\text{-}RDC},\chi^2}
    (p;q^{\otimes m})}m,
    \quad
    \frac{\mathcal C_{m,\varepsilon_m}^{\mathrm{univ},\chi^2}
    (p;q^{\otimes m})}m
    \longrightarrow\rho_{p,q}^\star(d),
\]
and, hence, 
TA-RDC is first-order optimal under the $\chi^2$ divergence as well.
Furthermore, 
\[ 
\mathcal C_{m,\varepsilon}^{\mathrm{TA\text{-}RDC}, \chi^2}(p,q)
    \le
    \mathcal C_{m,\log(1+\varepsilon)}^{\mathrm{univ},\textrm{KL}}(p,q)
    +\frac{3\sqrt m}{\min_{i\in[k]}p_i} \leq 
    \mathcal C_{m,\varepsilon}^{\mathrm{univ},\chi^2}(p,q)
    +\frac{3\sqrt m}{\min_{i\in[k]}p_i}.
\]

For fixed \(\varepsilon\in[0,1]\), write
\(\kappa(p,q)\coloneqq\max_{i\in[k]}q_i/p_i\).
The TV and normalized squared Hellinger frontiers are instead
\[
\frac{\mathcal C_{m,\varepsilon}^{\mathrm{univ},\mathrm{TV}}
(p;q^{\otimes m})}m
\longrightarrow1+(1-\varepsilon)\{\kappa(p,q)-1\},
\qquad
\frac{\mathcal C_{m,\varepsilon}^{\mathrm{univ},H^2}
(p;q^{\otimes m})}m
\longrightarrow1+(1-\varepsilon)^2\{\kappa(p,q)-1\},
\]
where \(H^2(P,Q)\coloneqq1-\sum_x\sqrt{P(x)Q(x)}\).  
The optimal rule
E-samples a coin before querying, runs RDC with probability
\(1-\varepsilon\) for TV or \((1-\varepsilon)^2\) for Hellinger, and
otherwise returns the first \(m\) M-samples.  

Thus, TA-RDC is the
first-order optimal procedure for KL and Pearson \(\chi^2\), while the
randomized RDC/direct-sampling rule is the choice for TV and
Hellinger.
However, this randomized rule does not give control over individual runs: it
controls neither the pathwise attained discrepancy as in TA-RDC, nor the number of queries as in CA-RDC.
We provide experimental results with TA-RDC, the first-order optimal frontier and its optimal procedure under TV distance in \cref{app:exp:exp-details}.

\subsection{Evaluation of M-costs and statistical alignment discrepancies}\label{app:metric-evaluation}

In this section, we give details 
for computing M-costs and 
statistical alignment discrepancies (cf.~\cref{def:uncond-sampling-fairness-discrepancy,def:count-cond-sampling-fairness-discrepancy}) 
of the various post-processing methods considered in \cref{sec:exp,app:alg-approx-sampling-uncond-discrep}.
We consider a product law for the target joint law $\nu = q^{\otimes m}$, where $q \in \Delta^{k-1}$ are the user-specified target rates.

Statistical alignment discrepancies with closed form expressions can be directly evaluated
either via a theorem (e.g., the unconditional discrepancy is zero for RDC by \cref{thm:rdc-exact-main})
or complete enumeration of state spaces of appropriate multinomial random variables.

However, when the number of attributes $k$ and/or output samples $m$ is large,
exact computation may be infeasible.
In this case, 
we approximate the statistical alignment discrepancy via Monte Carlo estimation. 

\emph{Setting and notation.}
Throughout this section, we
assume the source and target rates $p$ and $q$ 
can lie in $\Delta^{k-1}$ (in particular, they need not lie in the interior of the simplex).
We
denote their supports as \(S_{\src}=\operatorname{supp}(p)\)
and
\(\textrm{supp}(q)=\operatorname{supp}(q)\), respectively.

Furthermore, let \(\mathcal T_t
=\{\vec\ell\in\mathbb N_{\geq0}^k:
\textstyle\sum_i\ell_i=t\}\) 
be the space of possible count vectors returned by $t\in \N$ samples and
\( \rho_q(\vec\ell)
=\Pr\left(\Mult(m,q)=\vec\ell\right) \) 
index values $\vec \ell \in \mathcal T_m$
of a multinomial p.d.f. with $m$ trials and success rates equal to the target rates $q$, with the convention $0^0=1$.
For a source state \(\vec c\), define its feasible target mass by
\begin{align*}
    \alpha(\vec c)
    =\Pr_{\vec L\sim\Mult(m,q)}(\vec L\leq\vec c)
    =\sum_{\substack{\vec\ell\in\mathcal T_m: ~
                     \vec\ell\leq\vec c}}
       \rho_q(\vec\ell)
\end{align*}
which are multinomial c.d.f.s that can be efficiently computed via recursive identities;
see, e.g., \citet{hayter2014recursive} and \cref{sec:algos-anytime-rdc} for more details.

\emph{Output count law.}
Sampling an attribute sequence from the target joint law $\nu = q^{\otimes m}$ can be equivalently thought of as
drawing a count vector $\vec \ell \in \mathcal{L}_m$
according to the multinomial law $\vec L \sim \Mult(m,q)$ and then
selecting an attribute sequence uniformly over the chosen type class
\begin{equation}\label{eq:type-attribute}
    T(\vec \ell)
    :=
    \left\{
        \vec a=(a_1,\ldots,a_m)\in[k]^m:
        \sum_{t=1}^m 
        1\{a_t=i\}=\ell_i,
        ~\forall i\in[k];
    \right\}
\end{equation}
\citep[see, e.g.,][]{cover1991elements, csiszar1998method}.
Under the target joint law $\nu = q^{\otimes m}$,
the black-box post-processing methods we consider in this paper also select an attribute sequence uniformly over a chosen type class, for a (possibly random) count vector $\vec \ell$ (cf.~\cref{sec:approx-sampling-fairness}).

Consequently, obtaining closed-form expressions for the statistical alignment discrepancies often requires explicitly characterizing the output law of the count vectors returned by a post-processing algorithm; 
see, also, the discussion about minimal sufficiency of the count vectors 
in \cref{sec:count-conditional-target-sampling}.
Thus, let
$\widetilde{K} \in \mathcal{T}_m$ be the attribute counts of the $m$ returned samples from a post-processing algorithm $\textrm{Alg}_m \in \mathfrak{A}_m^o$ (cf.~\cref{def:black-box-post-processing-algs}).
Then,
for every count vector $\vec c \in \N_{\ge 0}^k$ and $\vec \ell \in \mathcal{T}_m$, 
denote
the output count law after conditioning on the realized attribute counts as
\begin{equation}\label{eq:count-cond-output-count-law}
    K^{\vec C}_{\textrm{Alg}_m}(\vec \ell; \vec c)
    \coloneqq
    \Pr (\widetilde{K} =\vec \ell \mid x,\vec C=\vec c)
\end{equation}
(see, also, \cref{eq:count-cond-output-label-law}, which defines the output \textit{attribute} law as the law of the attribute sequence returned by a post-processing algorithm).
Finally, denote the unconditional output count law of a post-processing algorithm as \begin{equation}\label{eq:count-cond-output-label-law-uncond}
    K_{\textrm{Alg}_m}(\vec \ell)
    \coloneqq
    \Pr(\widetilde{K}=\vec \ell \mid x)
\end{equation}

\subsubsection{Metrics for RDC}
\label{app:metric-evaluation-rdc-stopping-time-dist-m-cost}

Assume the support condition
\(\textrm{supp}(q)\subseteq \textrm{supp}(p)\), so that the RDC stopping time satisfies
\(T_{\nu}^\star<\infty\) almost surely (cf.~\cref{thm:rdc-exact-main}).

{\bf M-cost and stopping-time distribution.}
Note that \[
    T_\nu^\star
    =\inf\{t\ge m: \vec A^\star \in \Omega_t\}
    =\inf\{t\ge m:C_i(t)\ge L_i,\ \forall i\in[k]\}
    \eqqcolon T_{\vec L},
\]
where  $L_i=\sum_{j=1}^m1\{A_j^\star=i\}$ are the attribute counts of the target attribute sequence $\vec A^\star \sim \nu$.
In other words,
an equivalent representation of the RDC algorithm involving count vectors draws a target attribute sequence $\vec A^\star\sim\nu$, records its attribute counts $L_i$, 
and then
M-samples until it attains at least $L_i$ outputs with annotated attribute $i$.

Then, independence of the target demand and
the source stream gives the c.d.f. of the stopping time and M-cost of the RDC algorithm as
\begin{align}
    \Pr(T_{\nu}^\star\leq t)
       &=\mathbb E_{\vec C(t)\sim\Mult(t,p)}\left[\alpha(\vec C(t))\right],
       \label{eq:app-synth-rdc-cdf}\\
    \mathcal C_m^{\RDC}(p,q)
       &=\mathbb E[T_{\nu}^\star]
        =\sum_{t=0}^{\infty}
          \left\{1-\mathbb E_{\vec C(t)\sim\Mult(t,p)}
          \left[\alpha(\vec C(t))\right]\right\},
       \notag %
\end{align}
respectively. 
The M-cost of the RDC algorithm also admits an alternative expression involving an expectation over target count vectors.
For a count vector $\ell=(\ell_1,\ldots,\ell_k)\in\mathbb N_{\ge0}^k$, define the hitting time \(T_{\ell}
    =
    \inf\left\{
        t\ge0:C_i(t)\ge\ell_i,\ \forall i
    \right\} \)
    and let 
    \( S_r(u)
    =
    \sum_{j=0}^{r-1}\frac{u^j}{j!}\) with \(S_0(u)=0\).
    
\begin{proposition}[Tail-sum and Poissonized integral formulas]\label{prop:rdc-exact-formulas}
For every $\vec\ell\in\mathbb N_{\geq0}^k$,
\begin{align}
    \mathbb E[T_{\vec\ell}]
    & =
    \sum_{t=0}^{\infty}
    \Pr\left(
        \exists i:C_i(t)<\ell_i
    \right),
    \notag %
    \\
    & =
    \int_0^{\infty}
    \left\{
        1-
        \prod_{i=1}^k
        \left[
            1-e^{-p_it}S_{\ell_i}(p_it)
        \right]
    \right\}dt.
    \label{eq:quota-poisson-integral}
\end{align}
Consequently,
\begin{equation}\label{eq:rdc-integral}
    \mathcal C_m^{\RDC}(p,q)
    =
    \mathbb E_{\vec L\sim\Mult(m,q)}
    \left[
        \int_0^{\infty}
        \left\{
            1-
            \prod_{i=1}^k
            \left[
                1-e^{-p_it}S_{L_i}(p_it)
            \right]
        \right\}dt
    \right].
\end{equation}
\end{proposition}
The proof is analogous to the one given in \citet{doumas2016coupon} which gives the expected number of coupons needed to collect $m$ sets of $k$ coupon types.

For exact evaluations,
the expectation in \cref{eq:app-synth-rdc-cdf} 
can be computed by enumerating count vectors in $\vec c \in \mathcal T_n$ and then calculating the multinomial c.d.f. $\alpha(\vec c)$
in the argument (cf.~\cref{sec:algos-anytime-rdc}).
Similarly, the expectation in 
\cref{eq:rdc-integral} can be computed by enumerating count vectors in $\mathcal T_m$ and then calculating the one-dimensional integral in the argument.

For approximate evaluations, we use Monte Carlo by drawing $r\in[R]$ replicates
\(\vec L^{(r)}\overset{\mathrm{i.i.d.}}{\sim}\Mult(m,q)\) and simulating a source-count process until
\(\vec C^{(r)}(t)\geq\vec L^{(r)}\).
We then either average the resulting stopping times to estimate the M-cost or obtain empirical survival probabilities and quantiles to estimate \cref{eq:app-synth-rdc-cdf}.

\subsubsection{Metrics for thresholded anytime RDC}
Let
\(
\mathcal G_t\coloneqq\sigma\{\vec C(s):0\le s\le t\}
\)
be the source-count filtration, and let
\(\sigma\geq m\) be either a deterministic time or an almost surely finite
\(\{\mathcal G_t\}_{t\geq 0}\)-stopping time. 
Write \(\vec C_\sigma \equiv \vec C (\sigma)\)
as the accumulated source-count vector (cf.~\cref{eq:feasible-support})
of the anytime RDC procedure
(A-RDC; cf.~\cref{def:anytime-rdc}) at time \(\sigma\). By
\cref{thm:opt-count-cond-f-curve}, conditional on
\(\vec C_\sigma=\vec c\), the output count law of A-RDC is
\begin{align*}
K_{\mathrm{A\text{-}RDC},\sigma}^{\vec C}
(\vec\ell;\vec c)
&\coloneqq
\Pr\left(
    \widetilde K=\vec\ell
    \,\middle|\,
    x,\vec C_\sigma=\vec c
\right)
=
\frac{
    \rho_q(\vec\ell)
    1\{\vec\ell\leq\vec c\}
}{
    \alpha(\vec c)
}
\end{align*}
whenever \(\alpha(\vec c)>0\). 
On the event \(\{\alpha(\vec C_\sigma)=0\}\), every feasible output has
zero target probability. 
Therefore, for every
\(\vec\ell\in\mathcal T_m\) satisfying \(\rho_q(\vec\ell)>0\), the
unconditional output count law of A-RDC at time \(\sigma\) satisfies
\begin{align}
K_{\mathrm{A\text{-}RDC},\sigma}(\vec\ell)
&\coloneqq
\Pr\left(
    \widetilde K=\vec\ell
    \,\middle|\,
    x
\right)=
\rho_q(\vec\ell)r_\sigma(\vec\ell),
\label{eq:a-rdc-uncond-output-count-law}
\end{align}
where
\begin{equation}
r_\sigma(\vec\ell)
\coloneqq
\mathbb E\left[
    \frac{
        1
        \left\{
            \alpha(\vec C_\sigma)>0,\,
            \vec\ell\leq\vec C_\sigma
        \right\}
    }{
        \alpha(\vec C_\sigma)
    }
\right],
\label{eq:a-rdc-uncond-density-ratio}
\end{equation}
with the convention that the integrand is zero when
\(\alpha(\vec C_\sigma)=0\). 

Recall from \cref{def:thresholded-anytime-rdc} that TA-RDC
(cf.~\cref{alg:anytime-rdc}) with nominal tolerance
\(\varepsilon\geq 0\) uses the source-count measurable stopping time \(\tau_{f,\varepsilon}
\coloneqq
\inf\left\{
    t\geq m:
    \Psi_f(\alpha_t)\leq\varepsilon
\right\}\).
Since TA-RDC and A-RDC agree pathwise, (i.e., \( \widehat{A}_{T_{\nu}^\star\wedge\tau_{f,\varepsilon}}
=
\widehat{ A}_{\tau_{f,\varepsilon}}
\) a.s.; see, also, \cref{eq:anytime-rdc-snapshot}),
the unconditional output count law of TA-RDC with nominal tolerance
\(\varepsilon\) is
\(K_{\mathrm{A\text{-}RDC},\tau_{f,\varepsilon}}\).
Thus, TA-RDC is obtained from the preceding general expressions by setting
\(\sigma=\tau_{f,\varepsilon}\).

{\bf Unconditional discrepancies.}
For a source-count measurable time \(\sigma\geq m\), define
\begin{equation}\label{eq:infeasibility-prob-beta-n}
    \beta_\sigma
\coloneqq
\Pr\left(
    \alpha(\vec C_\sigma)=0
\right)
\end{equation}
as the probability that the realized source state admits no feasible attribute
sequence having positive probability under the target joint law
\(\nu = q^{\otimes m}\). Equivalently, on this event, every feasible sequence of
output attributes has zero target probability.

The unconditional TV and KL discrepancies of A-RDC at the
source-count measurable time \(\sigma\) are, respectively,
\begin{align}
\Delta_{\mathrm{TV}}^{\mathrm{un}}
\left(
    \mathrm{A\text{-}RDC}_\sigma
\right)
&=
\frac{1}{2}
\mathbb E_{\vec L\sim\Mult(m,q)}
\left[
    \left|
        r_\sigma(\vec L)-1
    \right|
\right]
+
\frac{\beta_\sigma}{2},
\notag\\
\Delta_{\mathrm{KL}}^{\mathrm{un}}
\left(
    \mathrm{A\text{-}RDC}_\sigma
\right)
&=
\begin{cases}
\displaystyle
\mathbb E_{\vec L\sim\Mult(m,q)}
\left[
    r_\sigma(\vec L)\log r_\sigma(\vec L)
\right],
& \beta_\sigma=0,\\[0.8em]
+\infty,
& \beta_\sigma>0.
\end{cases}
\label{eq:app-synth-a-rdc-unconditional-kl-tv}
\end{align}
where we use the relation 
\(\sum_{\vec\ell:\,\rho_q(\vec\ell)>0}
\rho_q(\vec\ell)r_\sigma(\vec\ell)
=
1-\beta_\sigma\).

The discrepancies in
\cref{eq:app-synth-a-rdc-unconditional-kl-tv} also admit equivalent
representations as expectations under the unconditional output count law
of A-RDC at time \(\sigma\):
\begin{align}
\Delta_{\mathrm{TV}}^{\mathrm{un}}
\left(
    \mathrm{A\text{-}RDC}_\sigma
\right)
&=
\mathbb E_{\vec L\sim
K_{\mathrm{A\text{-}RDC},\sigma}}
\left[
    \left(
        1-\frac{1}{r_\sigma(\vec L)}
    \right)_{+}
    1\{\rho_q(\vec L)>0\}
\right]
+
\beta_\sigma,
\notag\\
\Delta_{\mathrm{KL}}^{\mathrm{un}}
\left(
    \mathrm{A\text{-}RDC}_\sigma
\right)
&=
\begin{cases}
\displaystyle
\mathbb E_{\vec L\sim
K_{\mathrm{A\text{-}RDC},\sigma}}
\left[
    \log r_\sigma(\vec L)
\right],
& \beta_\sigma=0,\\[0.8em]
+\infty,
& \beta_\sigma>0.
\end{cases}
\label{eq:app-synth-a-rdc-unconditional-kl-tv-returned-law}
\end{align}

To approximate these discrepancies, 
we simulate \(r \in [R]\) independent source paths
and record \( \vec C_\sigma^{(r)}
\coloneqq
\vec C^{(r)}\bigl(\sigma^{(r)}\bigr) \).
For TA-RDC, each source path is simulated until
\(\sigma^{(r)}=\tau_{f,\varepsilon}^{(r)}\). 
For any fixed \(\vec\ell\), the expectation in
\cref{eq:a-rdc-uncond-density-ratio} is then estimated by
\begin{equation}
\widehat r_{\sigma,R}(\vec\ell)
=
\frac{1}{R}
\sum_{r=1}^R
\frac{
     1
    \left\{
        \alpha(\vec C_\sigma^{(r)})>0,\,
        \vec\ell\leq\vec C_\sigma^{(r)}
    \right\}
}{
    \alpha(\vec C_\sigma^{(r)})
}.
\label{eq:a-rdc-uncond-ratio-mc-estimate}
\end{equation}
where the feasible mass \(\alpha(\vec C_\sigma^{(r)})\) is computed explicitly
using the multinomial recursion identities described in
\cref{sec:algos-anytime-rdc}.

When the count space of \(\mathcal T_m\) is sufficiently small, we evaluate
\(\widehat r_{\sigma,R}(\vec\ell)\) for every
\(\vec\ell\in\mathcal T_m\) and substitute these values into
\cref{eq:app-synth-a-rdc-unconditional-kl-tv}. For larger count spaces, we
instead (E-)sample $O \ge 1$
returned count vectors conditional on each
simulated source state according to
\begin{equation*}
\Pr\left(
    \vec L^{(r,o)}=\vec\ell
    \,\middle|\,
    \vec C_\sigma^{(r)}
\right)
=
\frac{
    \rho_q(\vec\ell)
    1\{\vec\ell\leq\vec C_\sigma^{(r)}\}
}{
    \alpha(\vec C_\sigma^{(r)})
},
\qquad
o\in[O]
\end{equation*}
so that, marginally,
each \(\vec L^{(r,o)}\) has law
\(K_{\mathrm{A\text{-}RDC},\sigma}\).
These conditional draws are sampled using the same multinomial recursion
identities used to evaluate the feasible masses
\(\alpha(\vec C_\sigma^{(r)})\).
These draws are then used 
to
approximate the expectations in
\cref{eq:app-synth-a-rdc-unconditional-kl-tv-returned-law}.

In particular, 
for KL divergence, 
we estimate \(\log r_\sigma\) using a cross-fitted classifier-based density-ratio estimator \citep[see, e.g.,][for details]{nguyen2010estimating,menon2016linking}; 
for TV distance,
we estimate \(r_\sigma(\vec L^{(r,o)})\) for each returned-count \(\vec L^{(r,o)}\)
by averaging
\(1\{\vec L^{(r,o)}\leq\vec C_\sigma^{(s)}\}/
\alpha(\vec C_\sigma^{(s)})\) over the stopped source states 
$s \in S^{(f)} \subseteq [R] \backslash r$ for disjoint hold-out folds $\bigsqcup_{f}S^{(f)} = [R]$.
To obtain the final estimate, we then average \(
\left(
    1-
    {1}/{
        \widehat r_\sigma
        (\vec L^{(r,o)})
    }
\right)_+
\)
over the returned-count samples.

When \(\beta_\sigma>0\), the unconditional KL discrepancy is
\(+\infty\). For TV, the mass \(\beta_\sigma\) can be estimated from the
same source paths by
\[
\widehat\beta_{\sigma,R}
=
\frac{1}{R}
\sum_{r=1}^R
1\left\{
    \alpha(\vec C_\sigma^{(r)})=0
\right\}.
\]

{\bf Attained statewise certificate.}
Because the source-count process is discrete, the statewise discrepancy
at \(\tau_{f,\varepsilon}\) may be strictly smaller than the nominal
tolerance \(\varepsilon\).
In particular, for the TV and KL certification
times, respectively,
\begin{align}
\Delta_{\mathrm{TV}}^{\mathrm{un}}
\left(
    \mathrm{TA\text{-}RDC}_{\mathrm{TV},\varepsilon}
\right)
&\leq
\mathbb E\left[
    1-\alpha\!\left(
        \vec C(\tau_{\mathrm{TV},\varepsilon})
    \right)
\right]
\leq
\varepsilon,
\notag\\
\Delta_{\mathrm{KL}}^{\mathrm{un}}
\left(
    \mathrm{TA\text{-}RDC}_{\mathrm{KL},\varepsilon}
\right)
&\leq
\mathbb E\left[
    -\log\alpha\!\left(
        \vec C(\tau_{\mathrm{KL},\varepsilon})
    \right)
\right]
\leq
\varepsilon.
\label{eq:app-synth-a-rdc-state-kl-tv}
\end{align}

Using the same simulated stopped states $\{\vec C_\sigma^{(r)}\} _{r\in [R]}$, 
we estimate the attained
statewise TV and KL certificates by
\begin{align}
\widehat\Delta_{\mathrm{TV}}^{\vec C}
\left(
    \mathrm{TA\text{-}RDC}
\right)
&=
\frac{1}{R}
\sum_{r=1}^R
\left[
    1-\alpha(\vec C_\sigma^{(r)})
\right],
\notag\\
\widehat\Delta_{\mathrm{KL}}^{\vec C}
\left(
    \mathrm{TA\text{-}RDC}
\right)
&=
\frac{1}{R}
\sum_{r=1}^R
\left[
    -\log\alpha(\vec C_\sigma^{(r)})
\right],
\label{eq:app-synth-a-rdc-state-kl-tv-mc-estimate}
\end{align}
where \(\sigma=\tau_{\mathrm{TV},\varepsilon}\) or
\(\sigma=\tau_{\mathrm{KL},\varepsilon}\), depending on the $f$-divergence used for the statewise certificate. 

{\bf M-cost.} 
The M-cost of TA-RDC with source-count measurable 
stopping time $\tau_{f,\varepsilon}$ (cf.~\cref{def:thresholded-anytime-rdc})
is
\[
\mathcal C_m^{\mathrm{TA\text{-}RDC}_{f,\varepsilon}}(p,q)
=
\mathbb E\left[
    T_{\nu}^\star\wedge\tau_{f,\varepsilon}
\right].
\]
We approximate this expectation using \(r \in R_{\mathrm{MCost}}\) independent
replications of the terminal target demand and the source stream
\begin{equation}
\widehat{\mathcal C}_m^{
    \mathrm{TA\text{-}RDC}_{f,\varepsilon}
}(p,q)
=
\frac{1}{R_{\mathrm{MCost}}}
\sum_{r=1}^{R_{\mathrm{MCost}}}
\left(
    T_{\nu}^{(r)}
    \wedge
    \tau_{f,\varepsilon}^{(r)}
\right).
\label{eq:app-synth-ta-rdc-m-cost-mc}
\end{equation}

\subsubsection{Metrics for capped anytime RDC}
\label{app:metric-evaluation-ca-rdc}

Since capped anytime RDC
(CA-RDC; cf.~\cref{alg:anytime-rdc}) and A-RDC agree pathwise
(i.e., \( \widehat{ A}_{T_{\nu}^\star\wedge n}
=
\widehat{ A}_n\) a.s.; see, also, \cref{eq:anytime-rdc-snapshot}),
the unconditional output count law of CA-RDC with cap \(n\geq m\)
coincides with that of A-RDC at the deterministic source-count measurable time
\(\sigma\equiv n\). 
Accordingly, CA-RDC's unconditional discrepancies and
statewise certificates are computed using exactly the preceding
expressions and estimators with
\begin{equation}\label{eq:metric-eval-ca-rdc}
    \sigma=n,
\qquad
\vec C_\sigma^{(r)}
\overset{\mathrm{i.i.d.}}{\sim}
\Mult(n,p).
\end{equation}

More explicitly, CA-RDC's population-level unconditional discrepancies
are obtained from
\cref{eq:app-synth-a-rdc-unconditional-kl-tv,eq:app-synth-a-rdc-unconditional-kl-tv-returned-law}
by setting \(\sigma=n\),
while their Monte Carlo approximations use
\cref{eq:a-rdc-uncond-ratio-mc-estimate} and the conditional returned-count
draws described above.
When exact
evaluation is computationally feasible, the expectations in
\cref{eq:app-synth-a-rdc-unconditional-kl-tv} are computed by enumerating
the count vectors in \(\mathcal T_m\) and \(\mathcal T_n\).

The infeasibility probability $\beta_n$ also admits a convenient form when \(\sigma=n\) is deterministic: note that 
\( \alpha(\vec C(n))=0
\Longleftrightarrow
\sum_{i\in\operatorname{supp}(q)}C_{n,i}<m\) where $\vec C(n)$ is the source-count process of \cref{eq:feasible-support}.
Thus,
\begin{equation}
\beta_n
=
\Pr\left\{
    \operatorname{Binomial}\left(
        n,
        \sum_{i\in\operatorname{supp}(q)}p_i
    \right)
    <m
\right\}.
\label{eq:beta-ca-rdc-binom-tail-prob}
\end{equation}
For \(n\geq m\), this probability is zero if and only if
\(\operatorname{supp}(p)\subseteq\operatorname{supp}(q)\). 

{\bf M-cost.}
The M-cost of CA-RDC with query cap \(n\geq m\) is
\[
    \mathcal C_m^{\mathrm{CA\text{-}RDC}_n}(p,q)
    = \mathbb E\left[
    T_{\nu}^\star\wedge n
    \right].
\]
We approximate this expectation using
\(r\in[R_{\mathrm{MCost}}]\) independent replications of the terminal
target demand and the source stream:
\begin{equation}
\widehat{\mathcal C}_m^{
\mathrm{CA\text{-}RDC}_n
}(p,q) =
\frac{1}{R_{\mathrm{MCost}}}
\sum_{r=1}^{R_{\mathrm{MCost}}}
\left(
T_{\nu}^{(r)}
\wedge n
\right).
\label{eq:app-synth-ca-rdc-m-cost-mc}
\end{equation}

{\bf Feasible-run discrepancies.}
In the setting where $\exists i \in[k]$ such that $p_i > 0$ but $q_i = 0$,
there is a strictly positive probability that an M-sampled batch of $n$ outputs $\mathcal{B}_n = (Y_1,\ldots, Y_n)$ 
admits no feasible attribute sequence under the target joint law.
In this case, the KL discrepancy for CA-RDC with any capped-budget $n\ge m$ is $+\infty$, so we 
report the feasible-run discrepancy \(D_{\mathrm{KL}}(P_{\textrm{CA-RDC}}(\widehat A_n\mid\alpha_n>0)\|q^{\otimes m})\) and the feasibility probability 
$\Pr(\alpha_n > 0) = 1 - \beta_n$ (cf.~\cref{sec:emp-exp,eq:infeasibility-prob-beta-n}).

The former is computed by filtering the source streams in the empirical expectations of 
\cref{eq:a-rdc-uncond-ratio-mc-estimate} to those that contain at least $m$ outputs whose attribute lies in the support of the target rates (since \( \alpha(\vec C(n))=0
\Longleftrightarrow
\sum_{i\in\operatorname{supp}(q)}C_{n,i}<m\)).
For the latter, we report \cref{eq:beta-ca-rdc-binom-tail-prob} using empirical plug-in estimates for the source rates. 
These estimates may be obtained, for example, through a frozen pool of $S$ outputs $(Y_1,\ldots,Y_S)$, 
so that $\widehat p_i = 1/S \cdot \sum_{s=1}^{S} 1 \{ \phi(Y_s) = i\}, \forall i\in[k]$ (cf.~\cref{sec:emp-exp}).

\subsubsection{Metrics for rejection sampling algorithms}\label{app:metric-evaluation-rs}

{\bf Oracle rejection sampling.} 
The oracle rejection sampling procedure of \cref{app:ors} has M-cost $\mathcal{C}_m^{\mathrm{OracleRS}}(p,q)
    = m\cdot \max_{i\in \textrm{supp}(q)}{q_i}/{p_i}$ 
    and is universally exact, so that its unconditional discrepancy under forward $f$-divergences is zero; see, \cref{prop:opt-perfect-sampling-fairness-oracle}.

{\bf Approximate rejection sampling.}
The output count law of the approximate rejection sampling procedure in \cref{app:alg-approx-sampling-uncond-discrep} does not, in general, have a closed-form expression. 
We thus approximate its M-cost and unconditional discrepancy via Monte Carlo estimate as follows. 

Conditional on a
stopped pilot path \(F\), let \(\tau(F)\), \(\lambda(F)\), and \(s(F)\)
denote the pilot size, the 
proposal acceptance probability, and the
output attribute law of accepted samples 
in \cref{eq:ars-frozen-acceptance,eq:ars-conditional-law}.
The number of
proposals required for \(m\) acceptances is negative binomial,
so the (pilot conditional) M-cost and its estimate using independent pilot paths \(F_1,\ldots,F_R\) are
\begin{equation*}
    \mathcal C_m^{\mathrm{ARS}}(p,q)
    =\mathbb E_p\left[\tau(F)+\frac{m}{\lambda(F)}\right]; \qquad \widehat{\mathcal C}_m^{\mathrm{ARS}}(p,q)
    =\frac1R\sum_{r=1}^R
    \left\{\tau(F_r)+\frac{m}{\lambda(F_r)}\right\},
\end{equation*}
respectively.
For the estimation of the sampling discrepancies, the unconditional output count law is
\begin{equation}
    K_m(\ell)
    =\mathbb E_F\!\left[\rho_{s(F)}(\ell)\right],
    \qquad
    \rho_v(\ell)
    =\frac{m!}{\prod_i\ell_i!}\prod_{i=1}^k v_i^{\ell_i},
    \qquad \ell\in\mathcal L_m.
    \label{eq:ars-unconditional-count-mixture}
\end{equation}
Following the 
standard cross-fitting construction of 
\citep{chernozhukov2018double},
we split the pilot paths into folds 
so that $\bigsqcup_{f} I_{f} = [R]$. 
If pilot path
\(F_r\) is in fold \(f\), draw \(B\) count vectors
\(L_{rb}\sim\Mult(m,s(F_r))\), and use the other folds $I_{-f} \coloneq [R]\backslash I_f$
to form
the density ratio
\begin{equation*}
    \widehat w_{-f}(\ell)
    =\frac{1}{|I_{-f}|}
    \sum_{u\in I_{-f}}
    \frac{\rho_{s(F_u)}(\ell)}{\rho_q(\ell)}
    =\frac{1}{|I_{-f}|}
    \sum_{u\in I_{-f}}
    \prod_{i=1}^k
    \left\{\frac{s_i(F_u)}{q_i}\right\}^{\ell_i}.
\end{equation*}
We then construct estimates for the unconditional statistical alignment discrepancies as
\begin{align*}
    \widehat\Delta_{\mathrm{KL}}^{\mathrm{un}}
    &=\frac{1}{R\cdot B}\sum_{r=1}^R\sum_{b=1}^B
      \log \widehat w_{-f(r)}(L_{rb}),
      \\
    \widehat\Delta_{\mathrm{TV}}^{\mathrm{un}}
    &=\frac{1}{R\cdot B}\sum_{r=1}^R\sum_{b=1}^B
      \left(1-\widehat w_{-f(r)}(L_{rb})^{-1}\right)_+.
\end{align*}

\subsection{Additional details and results for experiments}\label{app:exp:exp-details}

\subsubsection{Additional results for simulation experiments}\label{app:exp:results-synth-exp}

{\bf Anytime RDC methods.}
We record the unconditional statistical alignment discrepancies of the post-processing algorithms in \cref{sec:synth-exp}
under both the KL divergence and TV distance across various settings outside of the main text.
We further record the 
first-order optimal frontier under the 
TV distance \[
    1+(1-\varepsilon)\{\kappa(p,q)-1\}, \quad \kappa(p,q) \coloneq \max_{i\in[k]} q_i / p_i,
\] 
as well as a randomized RDC/direct-sampling rule 
(\emph{RDC/DS mixture}) 
that runs RDC with probability $1-\varepsilon$ and otherwise returns the first $m$ samples with probability $\varepsilon$ (cf.~\cref{sec:algos-anytime-rdc}). 
Its M-cost is $\E[N_{\varepsilon}^{\textrm{RDC/DS}}] =
\E[T_{\nu}^\star] \cdot (1-\varepsilon) + m \cdot \epsilon$, and its 
unconditional discrepancies are 
\begin{align*}
\Delta_{\mathrm{KL}}^{\mathrm{un}}
\left(
    \mathrm{RDC/DS}
\right)
&=
D_{\mathrm{KL}}
\left(
    (1-\varepsilon)\cdot q^{\otimes m}
    +
    \varepsilon \cdot p^{\otimes m}
    \middle\|
    q^{\otimes m}
\right),\\
\Delta_{\mathrm{TV}}^{\mathrm{un}}
\left(
    \mathrm{RDC/DS}
\right)
&=
\varepsilon\cdot 
\mathrm{TV}
\left(
    p^{\otimes m},
    q^{\otimes m}
\right).
\end{align*}
We vary the normalized tolerance level $d = \varepsilon /m$ either in the interval $[0,D_{\mathrm{KL}}(p\|q)]$ for KL or the unnormalized tolerance \(\varepsilon\in[0,1]\)
for TV to trace out a cost-discrepancy profile.

\emph{Additional source settings.}
We fix the target joint law to $\nu = q^{\otimes m}$ with $q$ uniform and $m=300$, while
varying the number of attributes $k$ and the source rates; see \cref{tab:additional-source-settings}. 
The ratio $\kappa = \max_{i\in[k]} q_i / p_i$ captures the inherent difficulty of the setting,
with the oracle rejection sampler achieving the minimum possible M-cost among universally exact 
black-box post-processing algorithms, $m\cdot \kappa$ (cf.~\cref{app:ors}).
For CA-RDC, the cap $n$
is varied over a grid from $m$ to $n_{\max} = \operatorname{round}(m\cdot \max \{1.25 \kappa, \kappa + 0.5\})$.

\begin{table*}[t]
\centering
\caption{Additional synthetic source settings. In every setting, the target joint law is $\nu = q^{\otimes m}$ with the target rates uniform, $q_i=1/k$, and the output size is $m=300$.}
\label{tab:additional-source-settings}
\small
\setlength{\tabcolsep}{4pt}
\begin{tabular}{lclc}
\toprule
Setting  & $k$ & Source distribution $p$ & $\kappa=\max_i q_i/p_i$ \\
\midrule
mild skew & 3 & $(0.20,0.30,0.50)$ & 1.67 \\
one rare attribute & 4 & $(0.08,0.22,0.30,0.40)$ & 3.12 \\
graded source & 8 & $(0.04,0.06,0.08,0.10,0.12,0.16,0.20,0.24)$ & 3.12 \\
graded source & 12 & $p_i=(i+1)/90,\ i=1,\ldots,12$ & 3.75 \\
two mild groups & 16 & $((5/64)^{\times 8},(3/64)^{\times 8})$ & 1.33 \\
four rare attributes & 16 & $((1/64)^{\times 4},(5/64)^{\times 12})$ & 4.00 \\
\bottomrule
\end{tabular}
\end{table*}

As shown in \cref{fig:add-source-settings-kl-profiles,fig:add-source-settings-tv-profiles}, 
all methods exhibit a tradeoff between normalized M-cost and unconditional discrepancy across the six source settings.
TA-RDC and CA-RDC produce similar cost-discrepancy profiles under both KL and TV, while the RDC/DS mixture traces a linear tradeoff between its RDC and direct-sampling endpoints.
The profiles generally track their first-order optimal frontiers, 
although the size of the finite-sample deviations varies across the number of attributes and source rates.
Overall, the qualitative comparisons suggest that our anytime RDC methods remain stable across various source settings.

\begin{figure}[tbp]
    \centering

    \begin{subfigure}[t]{0.88\linewidth}
        \centering
        \includegraphics[width=\linewidth]{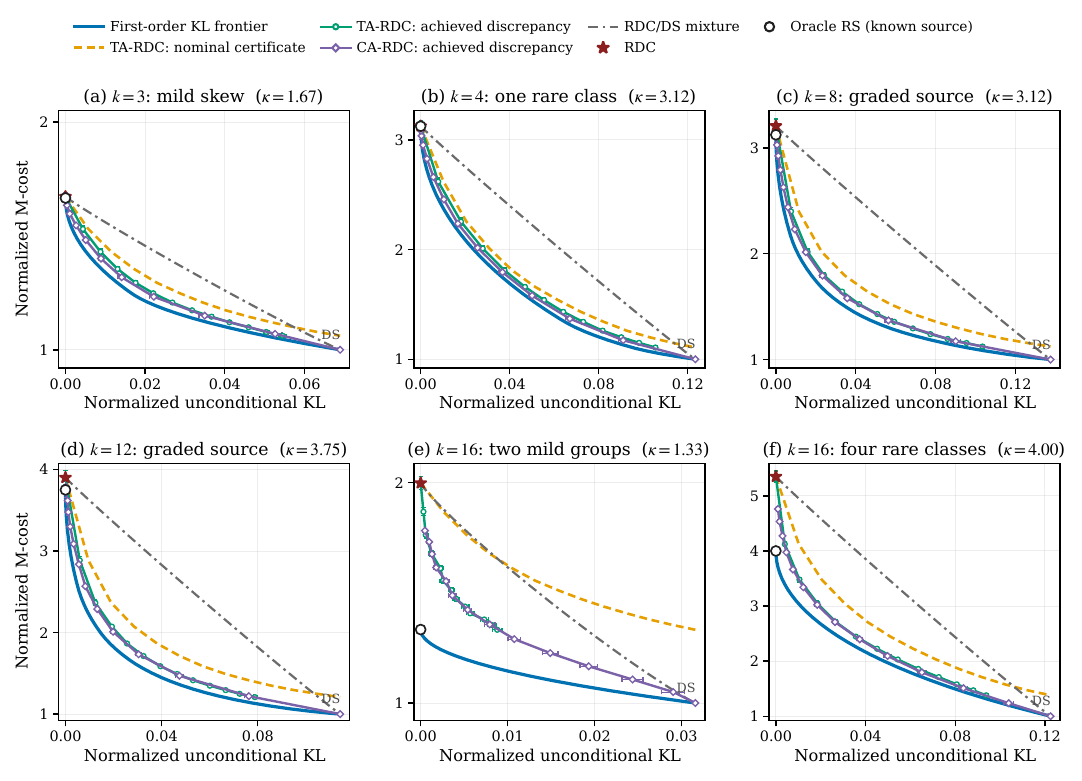}
        \caption{KL discrepancy profiles.}
        \label{fig:add-source-settings-kl-profiles}
    \end{subfigure}

    \vspace{0.25em}

    \begin{subfigure}[t]{0.88\linewidth}
        \centering
        \includegraphics[width=\linewidth]{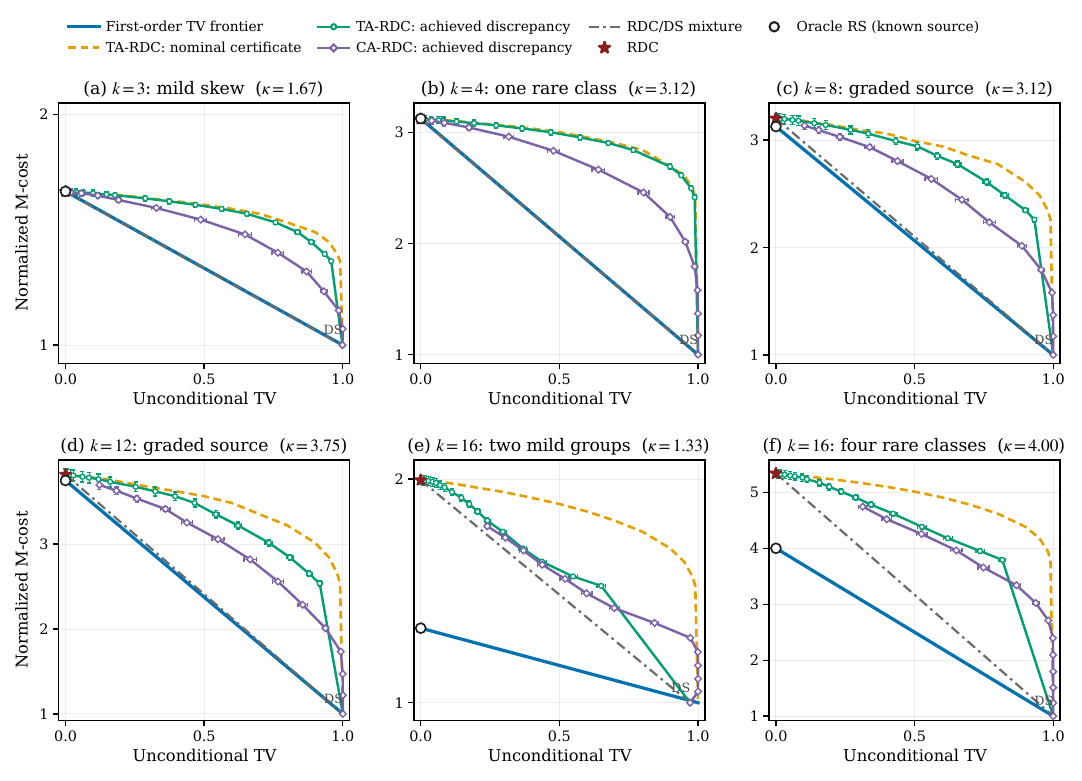}
        \caption{TV discrepancy profiles.}
        \label{fig:add-source-settings-tv-profiles}
    \end{subfigure}

    \caption{Cost--discrepancy profiles for the additional synthetic source 
    settings with a uniform target distribution; see \cref{tab:additional-source-settings} for exact source rate specifications.}
    \label{fig:add-source-settings-profiles}
\end{figure}

\emph{Finite sample performance of $m$.}
We vary the output size $m$ 
in the main text setting with $k=8$ attributes to examine the finite sample effects of the first-order optimality guaranteed by \cref{thm:first-order-uni-opt-ta-rdc,sec:algos-anytime-rdc}.
For CA-RDC, the cap $n$
is varied over a grid from $m$ to $n_{\max}=6m$, as in the main text.

As can be seen in \cref{fig:finite-sample-perf-m-k8}, the finite sample effects are pronounced at $m=5$ for both the KL divergence and TV distance, with the profiles of TA-RDC and CA-RDC achieving a lower M-cost than the first-order optimal RDC/DS mixture procedure for larger values of the unconditional TV (panel (f)).
However, the appropriate 
first-order optimal procedures converge to their respective frontiers even for moderate output sizes $m=100,200$.

\begin{figure}
    \centering
    \includegraphics[]{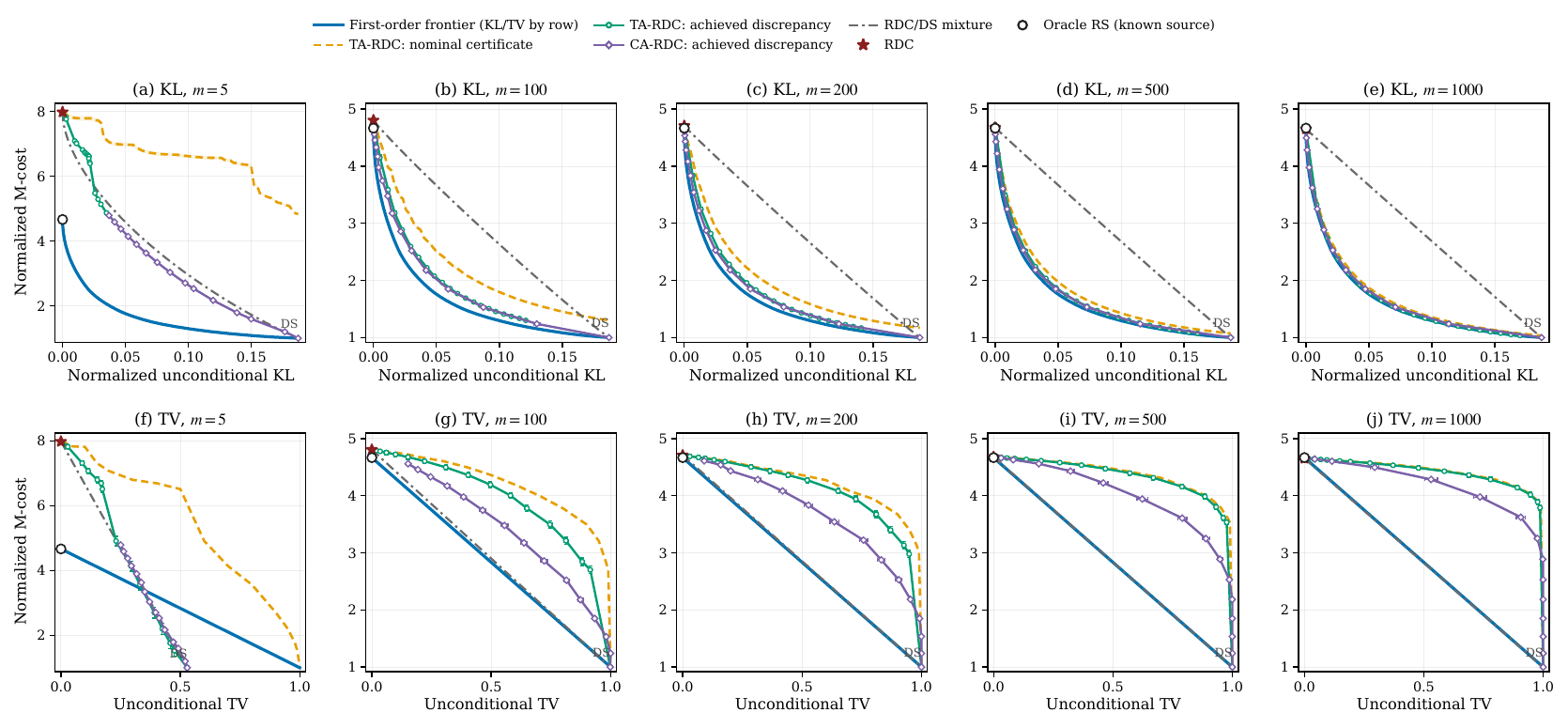}
    \caption{Finite output sample size performance of various post-processing methods in the main text setting with $k=8$ attributes and non-uniform source and target rates \( p=(0.03,0.05,0.08,0.12,0.16,0.18,0.18,0.20)\) and  \(q=(0.14,0.12,0.13,0.15,0.14,0.12,0.11,0.09)\).
    First-order optimal procedures -- TA-RDC and CA-RDC for KL divergence, and RDC/DS mixture for TV distance -- closely track 
    their respective frontiers even for moderate output sizes $m=100$ and $200$.}
    \label{fig:finite-sample-perf-m-k8}
\end{figure}

{\bf Approximate rejection sampling.}
We evaluate the ARS procedure in
the synthetic setting of the main text with $k=8$ attributes and non-uniform source and target rates.
We vary the output size \(m\in\{100,500,1000\}\).
Computational details for unconditional discrepancies and M-costs are given in \cref{app:metric-evaluation}.

\emph{Implementation details.}
We use the confidence sequence in
\cref{eq:ars-lindon-malek-cs} with \(\delta_{\rm CS}=0.05\).  The first look
is at time \(t_0=16\), with
subsequent looks following the dense geometric grid
\(t_{j+1}=\max\{\lceil1.04t_j\rceil,t_j+1\}\).
We vary the threshold for the 
relative width stopping time 
over
\(\gamma\in\{0.30,0.42,0.60,0.85,1.25,1.75,2.50,3.50,5.00\}\), with smaller values corresponding to more precise source-rate
estimates and pilot samples. 
All thresholds are evaluated on the same \(1{,}500\) independent pilot paths.
To estimate the sampling discrepancies, 
we generate $B=20$ count vectors per
pilot path with
three folds for the cross-fitting (cf.~\cref{app:metric-evaluation}). 

\emph{Results.} \cref{fig:ars-relative-width-fresh} plots the cost-discrepancy profile of ARS as we vary the threshold $\gamma$.
The M-cost and KL discrepancy are normalized by the output size $m$. 
At the loosest threshold \(\gamma=5\), its normalized M-costs are \(4.88\),
\(4.32\), and \(4.24\) for \(m=100,500,1000\), respectively, with normalized KL discrepancies of \(0.057\), \(0.077\), and
\(0.081\).
Tightening the threshold to \(\gamma=0.30\) reduces the normalized KL to
\(1.51\times10^{-6}\), \(6.38\times10^{-6}\), and
\(1.15\times10^{-5}\), and TV to \(0.0066\), \(0.0293\), and \(0.0560\),
but substantially raises its normalized M-cost to \(241.94\), \(52.12\), and \(28.39\).
Across all output sizes,
ARS is unable to track the first-order optimal frontier (cf.~\cref{thm:first-order-uni-opt-ta-rdc}).
In contrast,
TA-RDC (cf.~\cref{alg:anytime-rdc}) attains lower M-costs at comparable discrepancies and is first-order optimal under the KL divergence. 

\begin{figure}[tbp]
    \centering
    \includegraphics[width=\linewidth]{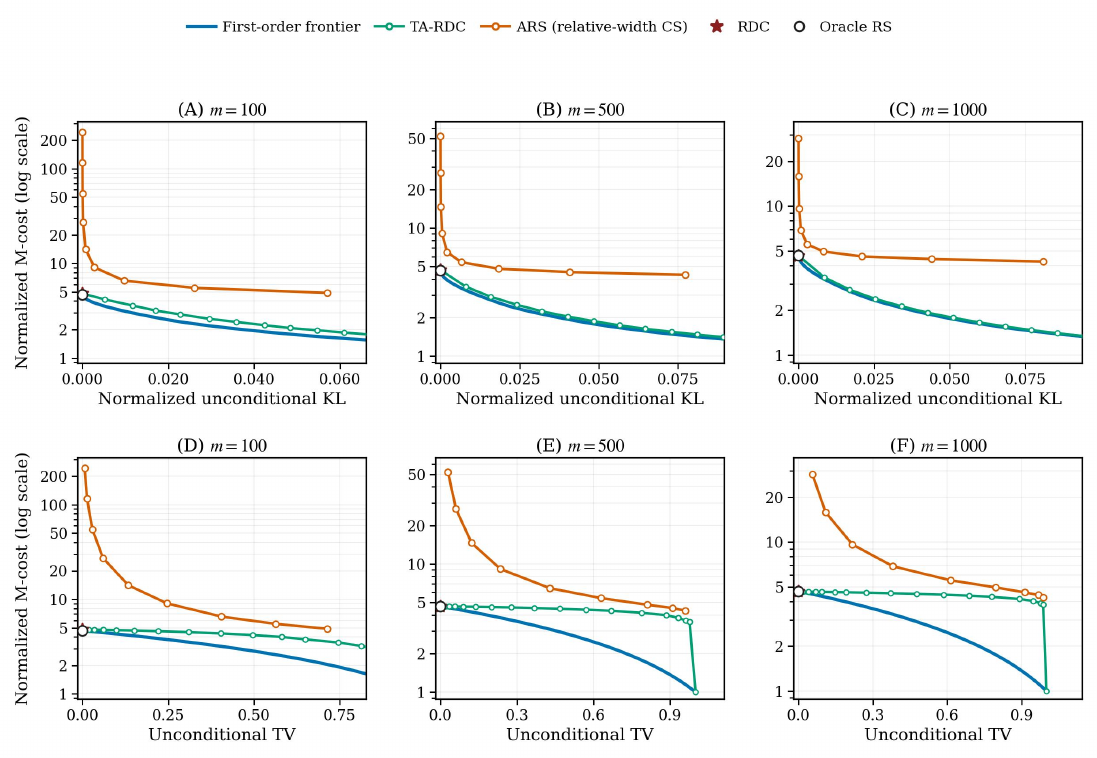}
    \caption{Cost--discrepancy profiles for approximate rejection sampling
    (ARS) as described in \cref{app:alg-approx-sampling-uncond-discrep}.
    Here, the setting is exactly that of the synthetic experiments in the main text, with $k=8$ attributes, non-uniform source and target rates, and output sizes $m \in \{100,500,1000\}$ (cf.~\cref{sec:synth-exp}).
    The top row reports unconditional KL discrepancy normalized by the output size $m$
    and the
    bottom row reports unconditional TV distance.
    The orange points of ARS vary the threshold \(\gamma\) of the relative width stopping time rule, with smaller thresholds corresponding to more  pilot samples and thus a larger M-cost.
    }
    \label{fig:ars-relative-width-fresh}
\end{figure}

\subsubsection{General implementation details for empirical experiments}\label{app:exp:details-emp-exp-general}

{\bf T2I models.} %
We use open-weight text-to-image models because publicly available and
open-weight diffusion models
have found recent, widespread 
use for synthetic data
generation and augmentation applications 
across various domains, 
including representation learning,
medical imaging, 
earth observation, and industrial inspection
\citep{tian2023stablerep,bluethgen2025vision,
sagers2023augmenting,sousa2025data,deng2025weed,valvano2024controllable}.

All image-generation calls are made through \texttt{fal}-hosted endpoints, with one PNG
image requested per prompt.
\begin{enumerate}
    \item[(i)] {\texttt{Flux-2-dev}} uses \texttt{fal} endpoint \texttt{fal-ai/flux-2}.
    We use portrait $4{:}3$ output, 28 inference steps, guidance scale 2.5,
    regular acceleration, and the safety checker enabled
    \citep{blackforestlabs2026flux2dev}.

    \item[(ii)] {\texttt{Qwen-Image-2512}} uses \texttt{fal} endpoint \texttt{fal-ai/qwen-image-2512}.
    We use portrait $4{:}3$ output, 28 inference steps, guidance scale 4.0,
    regular acceleration, and the safety checker enabled
    \citep{wu2025qwen, qwen2025qwenimage2512}.

    \item[(iii)] {\texttt{HiDream-O1-Image}} uses \texttt{fal} endpoint \texttt{fal-ai/hidream-o1-image}.
    We use portrait $4{:}3$ output, 50 inference steps, guidance scale 5.0,
    and the safety checker enabled
    \citep{hidream2026o1image}.
\end{enumerate}

{\bf Annotation and label construction for T2I experiments.} %
We annotate generated images using OpenAI \texttt{GPT-5.5} as a vision-language annotator \citep{openai2026gpt55docs}. 
For all T2I annotation runs, we authenticate using an API key supplied at run time and use the OpenAI Batch API to upload a JSONL file containing one \texttt{POST} request to \texttt{/v1/responses} per image, with an HTTPS image URL, a fixed annotation prompt, a strict JSON Schema, and a 24-hour completion window. 
For all image requests, 
we explicitly set the reasoning effort to \texttt{medium} while not explicitly setting the temperature, \texttt{top\_p}, output-token limit, or random seed.

\emph{Justification for vision-language annotation.}
We use \texttt{GPT-5.5} only as an operational annotator $\phi:\mathcal Y\to\mathcal A$, not as an alignment intervention and not as a trained replacement for a task-specific computer-vision classifier. This distinction is important: in our experiments $\phi$ maps each generated image to a small, prespecified set of coarse perceived-attribute bins under a fixed prompt, fixed allowable categories, and a strict JSON response schema. Modern vision-language models are explicitly designed to analyze image inputs \citep{openai2026gpt55docs,openai2026visiondocs}, and recent empirical studies suggest that they can provide accurate and cost-effective labels for large scale visual annotation when the label space is fixed and the task is specified carefully. 

For example, \citet{lu2024vlmannotatorsceleba} evaluate VLM-based annotation on the CelebA face-attribute dataset and report $79.5\%$ agreement with the original human labels, increasing to $89.1\%$ after re-annotating disagreement cases by majority vote, while costing less than $1\%$ of manual annotation. 
In a satellite-image classification study, \citet{beck2025chatgptsatelliteannotation} compare ChatGPT-4V with expert and non-expert human annotators and find that a hybrid strategy which routes only low-certainty cases to humans can reduce total annotation costs by more than $50\%$ with minimal loss in label quality. 
More broadly, \citet{xie2026mllmannotators} report that MLLM-generated image-tagging labels can reduce annotation cost to about $10^{-3}$ of human annotation cost and can achieve over $90\%$ of the downstream training performance obtained from human labels; they also find that structured prompting and label disambiguation substantially improve annotation quality. 

These findings support our use of a frontier VLM as a scalable first-pass annotator for coarse, prespecified visual labels. To mitigate annotation noise, we use the same fixed annotation protocol for all methods, require structured outputs, locally aggregate and validate returned fields, and interpret the demographic variables as perceived visual annotations rather than ground-truth identity.

\subsubsection{Additional implementation
details for individual-headshot experiments}
\label{app:exp:details-t2i-individual-headshot}

{\bf T2I models and individual-headshot prompts.} %
We use the three T2I models described in \cref{app:exp:details-emp-exp-general}.
The
prompt families for the individual headshot experiment 
are:
\begin{enumerate}
    \item[(i)] {HPS}
    \begin{quote}
    \small\itshape
    Show a photorealistic professional headshot of a person working in
    medicine or healthcare,
    in professional attire.
    \end{quote}

    \item[(ii)] {TSI}:
    The TSI prompt uses the same base prompt as HPS and appends
    \begin{quote}
    \small\itshape
    The image should be a single centered head-and-shoulders portrait of one
    individual. For this single image, internally choose exactly one demographic
    category uniformly at random from the 16 categories formed by gender
    \textnormal{\texttt{\{man, woman\}}},
    age
    \textnormal{\texttt{\{young adult, older adult\}}},
    and race/ethnicity
    \textnormal{\texttt{\{White/Caucasian,}} 
    \textnormal{\texttt{Black/African descent, Asian,
    Hispanic/Latino\}}}.
    Depict this person as belonging to the chosen category. Across separate
    generations, each of the 16 gender-age-race categories should be equally likely.
    \end{quote}

    \item[(iii)] {Demographic HPS}
    \begin{quote}
    \small\itshape
    Show a photorealistic professional headshot of
    \textnormal{\texttt{\{age article\}}}
    \textnormal{\texttt{\{race\}}}
    \textnormal{\texttt{\{gender\}}}
    working in medicine or healthcare,
    in professional attire.
    \end{quote}
    Gender ranges over \{\texttt{man}, \texttt{woman}\};
    age over \{\texttt{young}, \texttt{old}\} with articles
    \{\texttt{a young}, \texttt{an old}\}; and race over
    \{\texttt{White/Caucasian}, \texttt{Black/African-descent},
    \texttt{Asian}, \texttt{Hispanic/Latino}\}.
\end{enumerate}

{\bf Annotation and label construction.} %
Each image is annotated
with one structured request using a fixed prompt, fixed label set, and strict
JSON response schema.

For individual headshots, the prompt asks for perceived
gender, perceived race/ethnicity presentation, and age bin:
\begin{quote}
\small\itshape\raggedright
Annotate the visible person in this professional headshot. Return only JSON. Use exactly one bin for each field. gender must be one of: man, woman. race should mean perceived race/ethnicity presentation, not true identity. race must be one of: White/Caucasian, Black/African descent, Asian, Hispanic/Latino. age must be one of: young, old. Use the closest bin when a demographic attribute is visually ambiguous.
\par
\end{quote}

\subsubsection{Additional results for individual-headshot experiments}
\label{app:exp:results-individual-headshot}

{\bf Hard prompt failures.} %
\Cref{fig:individual-headshot-hard-prompt-failures} 
illustrates demographic HPS failures for the individual headshot T2I experiment.
Hard-prompt failures are driven almost entirely by race/ethnicity annotation mismatches, especially when the prompt requests a Hispanic/Latino person. 
For \texttt{Flux}, \(46.9\%\) of prompted Hispanic/Latino images are annotated as Asian, with smaller off-diagonal mass to Black/African descent and White/Caucasian. 
\texttt{Qwen} and \texttt{HiDream} show the same dominant failure direction, with \(22.4\%\) and \(12.7\%\) 
of prompted Hispanic/Latino images annotated as Asian, respectively.
Thus, 
demographic HPS is highly effective at controlling gender and age in this experiment, and it greatly improves support coverage, but it still leaves ambiguity in race/ethnicity presentation.

\begin{figure}
    \centering
    \includegraphics[width=0.99\linewidth]{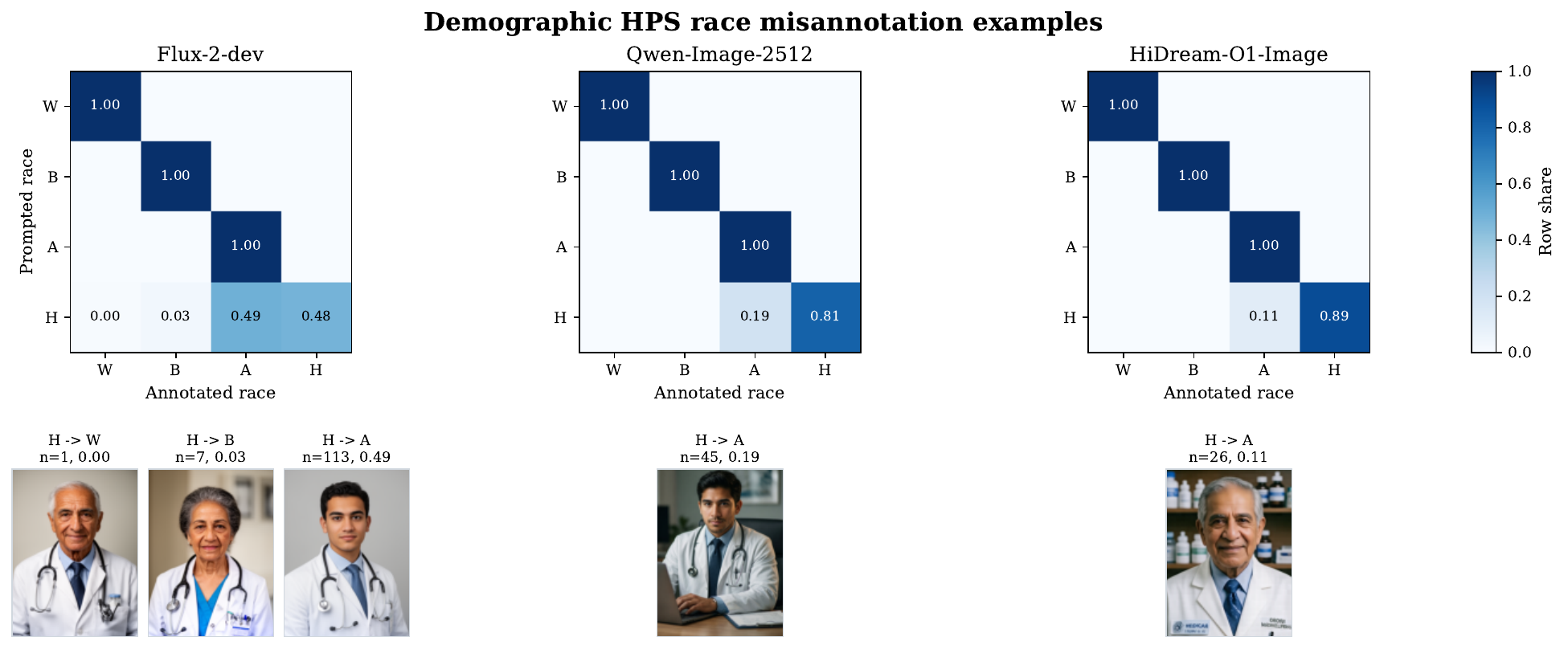}
    \caption{\footnotesize Race/ethnicity hard-prompt failure diagnostics for individual-headshot demographic HPS. Heatmaps show annotated race conditional on prompted race for each T2I model, and image panels show representative off-diagonal race/ethnicity annotation failures. Race abbreviations are W/B/A/H for White, Black, Asian, and Hispanic.}
    \label{fig:individual-headshot-hard-prompt-failures}
\end{figure}

{\bf Sensitivity analysis.} %
We vary the output size \(m\in\{25,50,100\}\) and report results for demographic HPS across all three T2I models. For each output size, we vary the normalized KL tolerance \(d\in[0,3]\), the TV tolerance \(\varepsilon\in[0,0.999]\), and the CA-RDC cap \(n/m\in[1,20]\). 

As shown in \cref{fig:individual-headshot-sensitivity-analysis}, both anytime methods consistently reduce the KL and TV discrepancies as the normalized M-cost increases. 
The cost-discrepancy 
profiles are stable across output sizes and models, although \texttt{Flux} generally requires a larger normalized M-cost than \texttt{Qwen} and \texttt{HiDream} to attain comparable discrepancies.

\begin{figure}
\centering
\includegraphics[width=\linewidth]{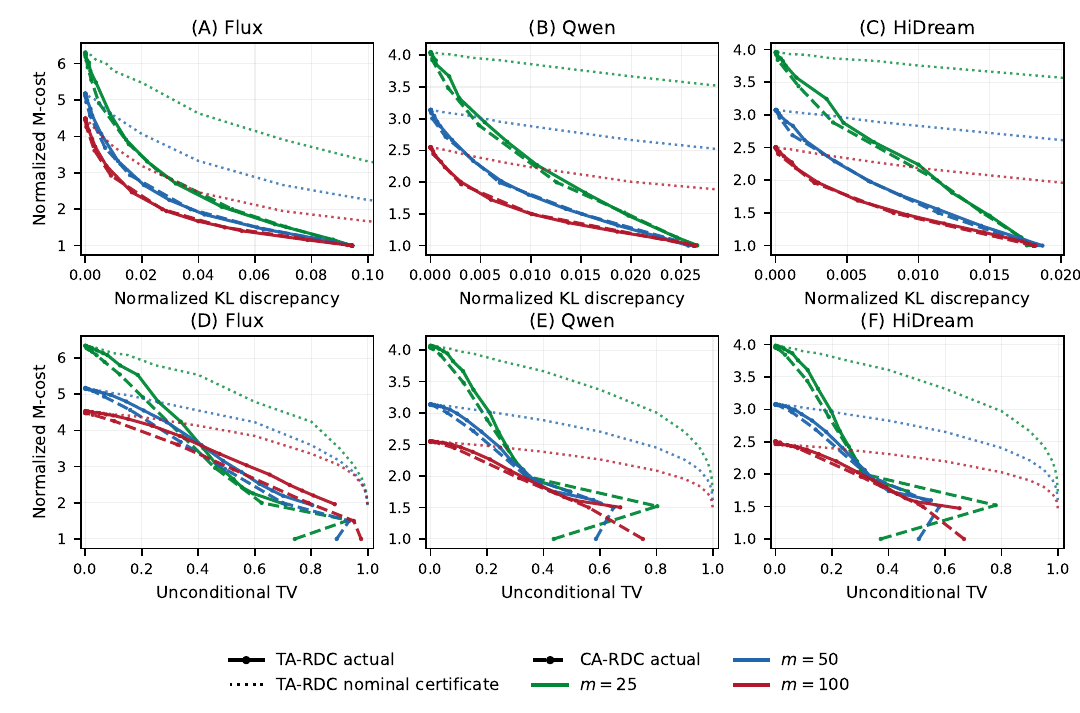}
\caption{\footnotesize
Output-size sensitivity analysis for individual-headshot demographic HPS.
Columns correspond to three T2I models, while the top and bottom rows report normalized
KL discrepancy and unconditional TV discrepancy, respectively.
Colors distinguish \(m\in\{25,50,100\}\); solid and dashed curves report
the actual discrepancies of TA-RDC and CA-RDC, respectively, and dotted
curves report the nominal TA-RDC certificates.
}
\label{fig:individual-headshot-sensitivity-analysis}
\end{figure}

\subsubsection{Additional implementation details for group-photo experiments}\label{app:exp:details-t2i-group-photo}
{\bf T2I models and group-photo prompts.} %
The same T2I models are used as in the individual-headshot photo experiments
(cf.~\cref{app:exp:details-emp-exp-general}).
In the main experiments, we consider HPS, TSI, and 
demographic HPS, all of 
which modify the input prompt to the T2I models.
In detail, their complete prompt
families are:
\begin{enumerate}
    \item[(i)] {HPS}
    \begin{quote}
    \small\itshape
    Generate a single-frame photorealistic candid scene centered on a
    foreground trio of exactly three adults in a
    medicine or 
    healthcare 
    setting. 
    The foreground trio should be
    the main subject, appear together in one camera shot, and be clearly
    countable as three distinct people. The scene should feel natural and in
    the moment, with realistic activity, posture, camera viewpoint, and
    interaction.
    \end{quote}

    \item[(ii)] {TSI}:
    The TSI prompt uses the HPS prompt and appends
    \begin{quote}
    \small\itshape
    For this image, first choose exactly one group composition uniformly at random
    from the target set of foreground-trio demographic count vectors: gender count
    is either 2 women and 1 man or 1 woman and 2 men; race/ethnicity count is
    either 1 White/Caucasian adult, 1 Black/African descent adult, and 1
    Hispanic/Latino adult, or 1 White/Caucasian adult, 1 Black/African descent
    adult, and 1 Asian adult; age count is either 2 young adults and 1 older adult,
    or 3 young adults. Depict the foreground trio so that its demographic counts
    match the chosen composition. These counts refer only to the three adults in
    the foreground trio. Across repeated independent generations, each target
    group-composition count vector should be equally likely.
    \end{quote}

    \item[(iii)] {Demographic HPS}:
    The demographic HPS 
    prompt uses the HPS prompt and appends
    \begin{quote}
    \small\itshape
    The foreground trio consists of:
    \textnormal{\texttt{\{gender count phrase\}}};
    \textnormal{\texttt{\{race count phrase\}}};
    \textnormal{\texttt{\{age count phrase\}}}. These counts refer only to
    the three adults in the foreground trio.
    \end{quote}
\end{enumerate}

For the three-person demographic HPS group-photo condition,
\[
    \begin{aligned}
    \{\text{gender count phrase}\} &\in
    \{\text{2 women and 1 man},\ \text{1 woman and 2 men}\},\\
    \{\text{race count phrase}\} &\in
    \left\{
    \begin{array}{l}
    \text{1 White/Caucasian adult, 1 Black/African descent adult,}\\
    \text{and 1 Hispanic/Latino adult},\\
    \text{1 White/Caucasian adult, 1 Black/African descent adult,}\\
    \text{and 1 Asian adult}
    \end{array}
    \right\},\\
    \{\text{age count phrase}\} &\in
    \{\text{2 young adults and 1 older adult},\ \text{3 young adults}\}.
    \end{aligned}
\]

{\bf Annotation and label construction.} %
We annotate generated images using OpenAI \texttt{GPT-5.5} as a
vision-language annotator \citep{openai2026gpt55docs} using the same API access settings as \cref{app:exp:details-emp-exp-general}.

For group photos, the annotator first counts the number of
people in the foreground. It then returns the gender, race/ethnicity
presentation, and age 
of each identified individual:
\begin{quote}
\small\itshape\raggedright
Annotate this group photo. Return only JSON. First count only the distinct visible people in the foreground/main group. Do not count background bystanders, distant incidental people, reflections, posters, or partial figures outside the main group. Then list each counted foreground person and use exactly one bin for each person's fields. gender must be one of: male, female. race should mean perceived race/ethnicity presentation, not true identity. race must be one of: White/Caucasian, Black/African descent, Asian, Hispanic/Latino. age must be one of: young, old. Base all demographic fields and group-level answers only on the counted foreground people.
\par
\end{quote}

After retrieving the annotations, we then locally aggregate these per-person
annotations into gender, race, and age count vectors. 
A group image is treated
as on-support only when the annotated number of foreground people equals the
requested group size and the resulting count vectors map to one of the
configured group-composition attributes (cf.~\cref{sec:exp-group-photos}).
Images outside this support are retained in the candidate pool but assigned
target mass zero for the group-composition target distribution.

\subsubsection{Additional results for group-photo experiments}\label{app:exp:results-group-photo}
{\bf Qualitative results.} %
\Cref{fig:group-photo-status-examples-tsi} and \cref{fig:group-photo-status-examples} 
show sampled images for various annotation categories.
On- and off-support images refer to those whose observed demographic count vector lies within or outside the support of the target rates, $\mathrm{supp}(q)$, respectively.
TSI can produce
valid group-composition attributes but often produces off-support images
by
realizing an invalid foreground count vector or adding/removing demographic
mass relative to the target support (cf.~\cref{fig:group-photo-status-examples-tsi}).

Images produced by demographic HPS directly specify the target demographic count vectors, so we further differentiate between hard-compliant and non-compliant images:
a hard-compliant image exactly matches the prompted count vector, while 
an on-support non-compliant image lies in the eight-attribute target support but matches a different count vector than requested (cf.~\cref{fig:group-photo-status-examples}). 
Post-processing methods such as A-RDC can use these non-compliant but on-support images to reduce the count-conditional statistical alignment discrepancy. 
Finally, in a similar fashion to TSI, 
off-support examples of HPS
often exhibit an incorrect 
foreground group size or demographic count vector.

\begin{figure}
    \centering
    \includegraphics[width=0.99\linewidth]{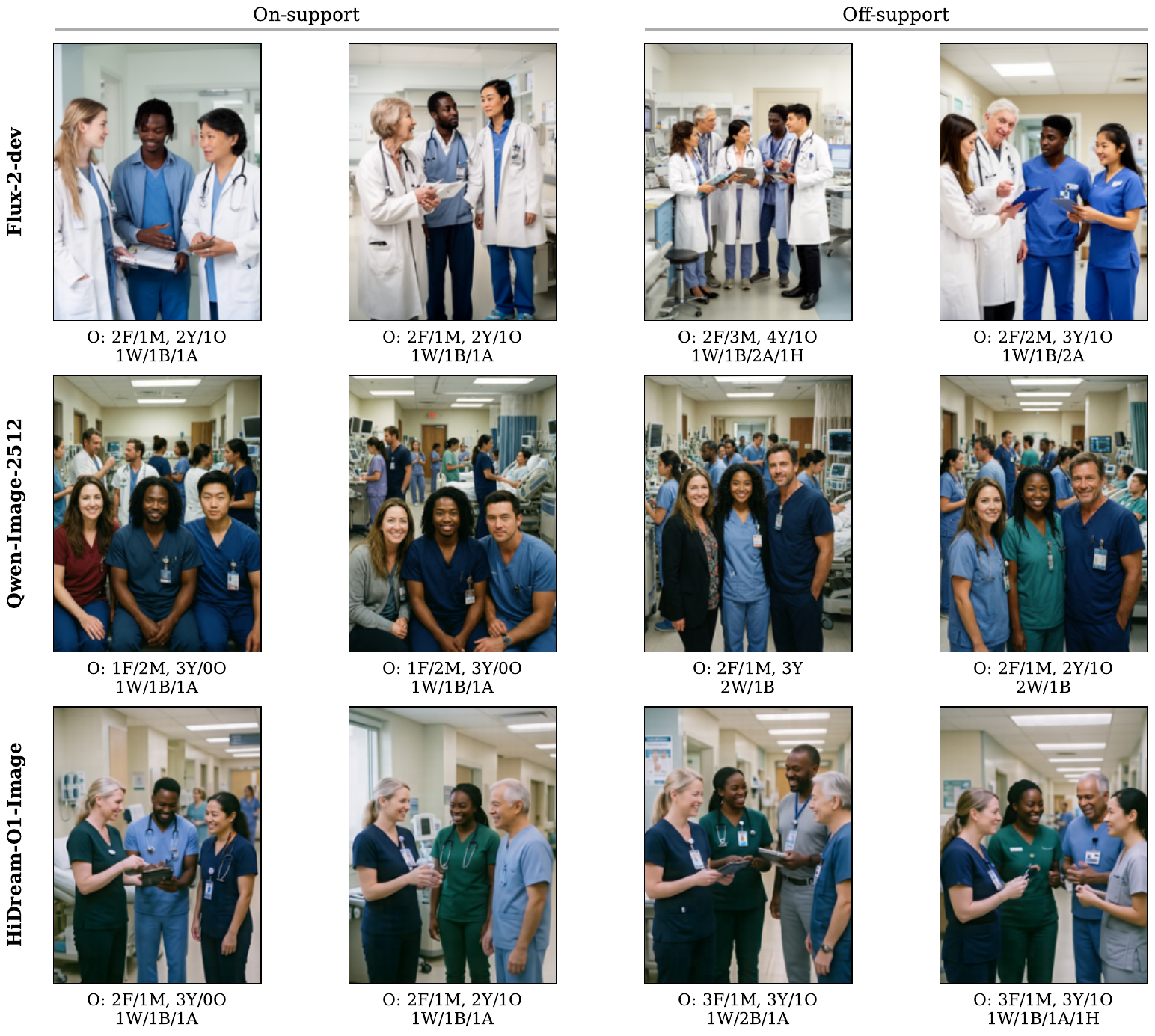}
    \caption{\footnotesize Examples of on- and off-support group-photo generations for each image model for the TSI method.
    Captions report the annotated observed count vector (O).
    The demographics are abbreviated as follows:
    F/M stands for female/male; 
    W/B/A/H stands for White, Black, Asian, and Hispanic;
    and Y/O stands for young/old.
    The number preceding each letter represents the number of foreground figures
    identified 
    as the associated demographic.}
    \label{fig:group-photo-status-examples-tsi}
\end{figure}

\begin{figure}
    \centering
    \includegraphics[width=0.99\linewidth]{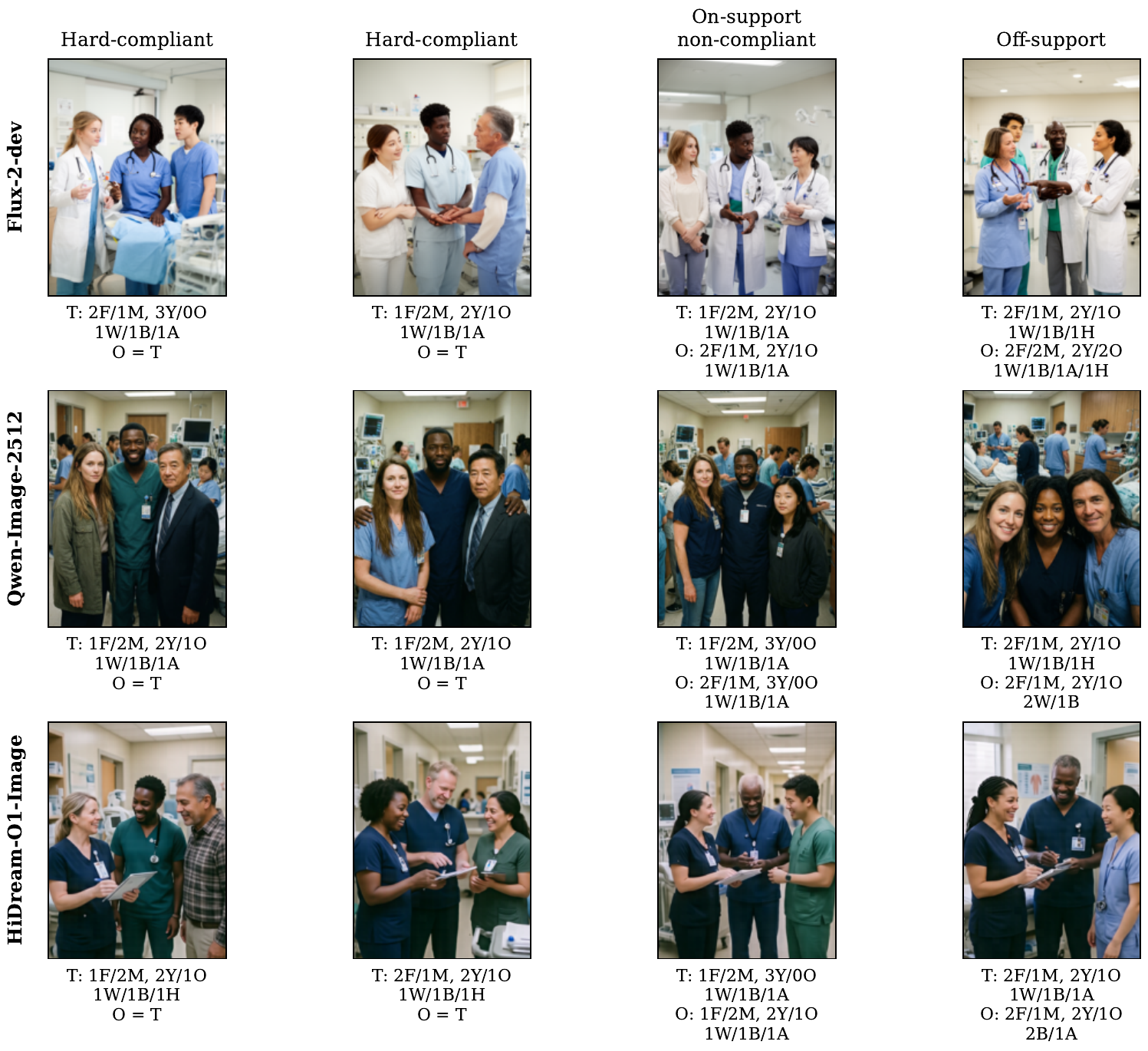}
    \caption{\footnotesize Examples of hard-compliant, on-support non-compliant, and off-support group-photo generations for each image model under demographic HPS. 
    Captions report the prompted target count vector (T) and annotated observed count vector (O).
    The demographic count vectors 
    are abbreviated using the same protocol 
    as \cref{fig:group-photo-status-examples-tsi}.
    }
    \label{fig:group-photo-status-examples}
\end{figure}

{\bf Sensitivity analysis.} %
We vary the output size \(m\in\{10,25,50\}\) and report results for demographic HPS across all three T2I models. 
For each output size, we vary the normalized KL tolerance \(d\in[0,3]\) and the TV tolerance \(\varepsilon\in[0,0.999]\).
For CA-RDC, we vary the capped-budget from \(n/m=1\) to \(22.32\), \(156.25\), and \(20\) for \texttt{Flux}, \texttt{Qwen}, and \texttt{HiDream}, respectively.
Feasible-run KL discrepancy is reported whenever there is a nonzero probability of a queried sample lying outside of the support of the target rates. 

As shown in \cref{fig:group-photo-sensitivity-analysis}, both anytime methods consistently reduce the KL and TV discrepancies as the normalized M-cost increases. 
The cost-discrepancy profiles are stable across output sizes, although \texttt{Qwen} requires substantially larger normalized M-costs because demographic HPS produces one target attribute at a particularly low rate of $\widehat p_{i} = 1/1000$.

\begin{figure}
\centering
\includegraphics[width=\linewidth]{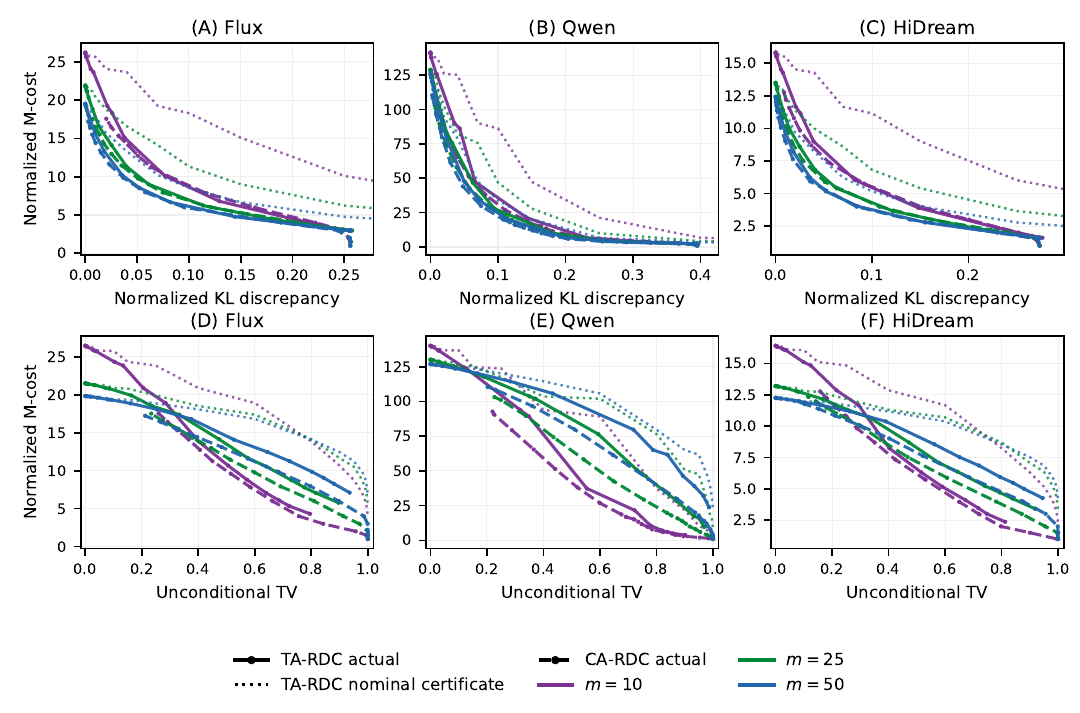}
\caption{\footnotesize
Output-size sensitivity analysis for group-photo demographic HPS.
Columns correspond to three T2I models, while the top row reports
unconditional KL discrepancy for TA-RDC and feasible-run KL discrepancy
for CA-RDC.
The bottom row reports unconditional TV discrepancy for both methods.
Colors distinguish \(m\in\{10,25,50\}\); solid and dashed curves report
the actual discrepancies of TA-RDC and CA-RDC, respectively, and dotted
curves report the nominal TA-RDC certificates.
}
\label{fig:group-photo-sensitivity-analysis}
\end{figure}

\subsubsection{Additional implementation details
for geocoded persona generation}\label{app:exp:details-svi}

{\bf SVI persona-generation model and prompts.} %
For the geocoded persona experiment, we call \texttt{GPT-5.5} and input a
text prompt asking the model to return one synthetic JSON patient persona.
The shared base prompt specifies the JSON schema, fixes the non-geographic
case-seed fields, and asks the model to generate the remaining fields:
\begin{quote}
\small\itshape
Generate one realistic synthetic adult patient persona for testing a
patient-facing digital health assistant.

Return strict JSON only with exactly these fields:
\textnormal{\texttt{sample\_id}},
\textnormal{\texttt{prompt\_variant}},
\textnormal{\texttt{requested\_label}},
\textnormal{\texttt{case\_seed\_id}},
\textnormal{\texttt{age\_band}},
\textnormal{\texttt{care\_setting}},
\textnormal{\texttt{condition\_category}},
\textnormal{\texttt{request\_type}},
\textnormal{\texttt{message\_tone}},
\textnormal{\texttt{city}},
\textnormal{\texttt{state}},
\textnormal{\texttt{zip\_code}},
\textnormal{\texttt{patient\_portal\_message}},
\textnormal{\texttt{barriers\_or\_preferences}}.

Set \textnormal{\texttt{sample\_id}} to
\textnormal{\texttt{\{sample id\}}}.
Set \textnormal{\texttt{prompt\_variant}} to
\textnormal{\texttt{\{prompt variant\}}}.
Set \textnormal{\texttt{requested\_label}} to
\textnormal{\texttt{\{requested label or null\}}}.

Use these controlled case-seed fields exactly:
\textnormal{\texttt{\{case\_seed\_id, age\_band, care\_setting,}}
\textnormal{\texttt{condition\_category, request\_type, message\_tone\}}}.
\end{quote}

The two prompt families used in the main comparison are:
\begin{enumerate}
    \item[(i)] {Target-label-blind county-seeded source.}
    The \texttt{TLB3\_county\_seeded} prompt uses the shared base prompt, sets
    \textnormal{\texttt{requested\_label}} to \textnormal{\texttt{null}}, and
    appends:
    \begin{quote}
    \small\itshape
    Location instruction:
    The patient is located in or near
    \textnormal{\texttt{\{county/state seed\}}}.
    Choose a plausible city and valid ZIP code in or near that county.
    Do not name hidden evaluation categories in the message body.
    \end{quote}
    This prompt never mentions SVI, social vulnerability, or the target
    quartile. The SVI quartile used in evaluation is assigned after generation
    by matching the generated ZIP to the CDC/ATSDR SVI file.

    \item[(ii)] {Hard SVI prompt baseline.}
    The \texttt{HARD\_SVI4\_DEFINITION} prompt uses the shared base prompt, sets
    \textnormal{\texttt{requested\_label}} to the requested SVI quartile, and
    appends:
    \begin{quote}
    \small\itshape
    Target-location instruction:
    The patient's ZIP should correspond to a place whose CDC/ATSDR
    Social Vulnerability Index percentile is in
    \textnormal{\texttt{\{quartile range\}}}.
    Do not mention SVI, vulnerability, deprivation, poverty, or the quartile
    in the message body unless it naturally follows from the case seed.
    \end{quote}
    The quartile ranges are
    \([0,.25)\) for \texttt{SVI\_Q1\_lowest},
    \([.25,.50)\) for \texttt{SVI\_Q2},
    \([.50,.75)\) for \texttt{SVI\_Q3}, and
    \([.75,1]\) for \texttt{SVI\_Q4\_highest}.
\end{enumerate}

\begin{table}[t]
\centering
\small
\caption{\footnotesize Representative generated persona extracts from the county-seeded source, one from each observed SVI quartile. The observed label is assigned after generation by matching the generated ZIP to the CDC/ATSDR SVI file.}
\label{tab:app-svi-persona-examples}
\begin{tabular}{@{}lllll@{}}
\toprule
Observed label & Age band & City, state, ZIP code & Care setting & Request type \\
\midrule
SVI Q1 & 18--29 & Utica, KS 67584 & Cardiology & Medication question \\
SVI Q2 & 45--64 & Lidgerwood, ND 58053 & Cardiology & Symptom question \\
SVI Q3 & 45--64 & Alpharetta, GA 30009 & Cardiology & Symptom question \\
SVI Q4 & 65+ & Algodones, NM 87001 & Cardiology & Symptom question \\
\bottomrule
\end{tabular}
\end{table}

{\bf Annotation and label construction.} %
The SVI experiment uses the CDC/ATSDR 2022 ZCTA-level Social Vulnerability Index file.
SVI percentiles are relative area-level rankings: an overall percentile of $0.80$ means that the ZCTA ranks above roughly $80\%$ of ZCTAs on that index, not that an individual patient has an $80\%$ vulnerability level.
The quartiles move from lower to higher area-level vulnerability, but the index is multidimensional; a ZIP code can have a high overall SVI because of different combinations of socioeconomic, household, racial/ethnic-minority, and housing/transportation conditions.
\Cref{fig:app-svi-quartile-context} summarizes the same national ZCTA file used for the audit.

\begin{figure}[tbp]
\centering
\includegraphics[width=0.66\linewidth]{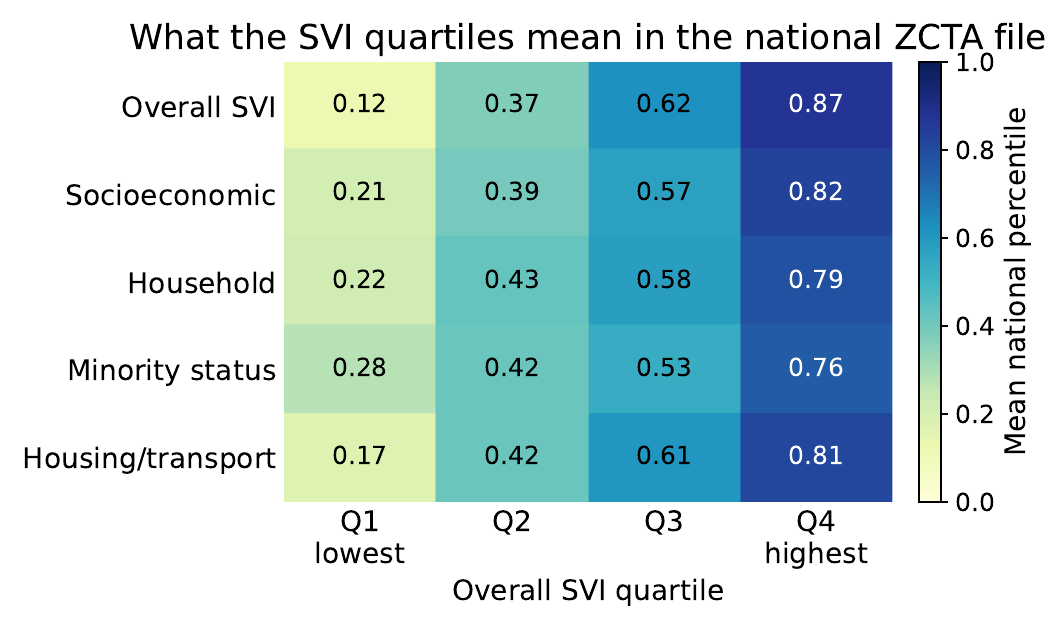}
\caption{\footnotesize Mean overall and theme-specific SVI percentiles for all matched ZCTAs in the 2022 CDC/ATSDR file, grouped by overall SVI quartile.}
\label{fig:app-svi-quartile-context}
\end{figure}

\subsubsection{Additional results
for geocoded persona generation}\label{app:exp:results-svi}

{\bf Representative generated locations.} %
The examples in \cref{tab:app-svi-persona-examples} illustrate how the deterministic ZIP-level audit can convert ordinary-looking location fields into SVI quartile labels.
For example, Utica, KS 67584 is classified as Q1 because its overall SVI percentile is $0.112$; its theme percentiles are also low on the national scale, including $0.145$ for socioeconomic status, $0.141$ for racial/ethnic-minority status, and $0.184$ for housing/transportation.
In the  SVI indicators, this ZCTA has low unemployment in the SVI file, a small racial/ethnic-minority share ($2.9\%$), and few no-vehicle households ($3.8\%$), so even though it has a small population and some older-adult and disability prevalence, the combined national ranking remains low.
Algodones, NM 87001 is classified as Q4 for the opposite reason: its overall SVI percentile is $0.890$, with very high theme percentiles for socioeconomic status ($0.906$), household characteristics ($0.955$), and racial/ethnic-minority status ($0.977$).
The  indicators report $37.4\%$ of residents below 150\% of poverty, $15.3\%$ unemployment, $24.0\%$ without a high-school diploma, and a $94.2\%$ racial/ethnic-minority share.

{\bf Hard SVI prompting diagnostics.} %
We diagnose the hard SVI prompt by comparing the requested quartile to the audited quartile of the generated ZIP.
Across the 1000 hard prompted candidates in the frozen candidate pool, all generated ZIP codes are matched to the SVI file so that all generated personas are on-support.
However, the requested-to-observed compliance is only $0.356$ and,
among the 644 failed hard prompts, $84.0\%$ are off by one quartile and $16.0\%$ are off by two quartiles.
Moreover, failures are almost all upward misses: $99.5\%$ of failed generations have an audited SVI quartile above the requested quartile.

\begin{figure}[tbp]
\centering
\includegraphics[width=0.66\linewidth]{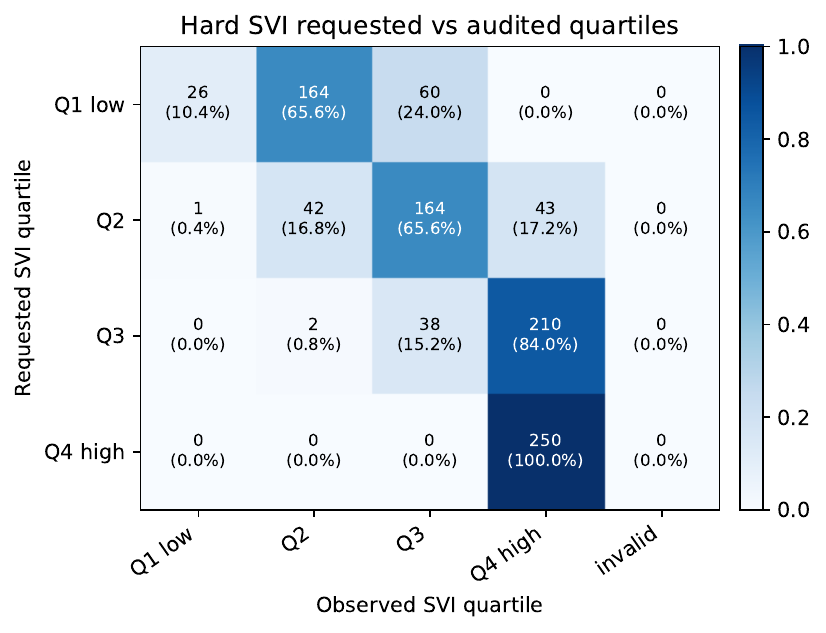}
\caption{\footnotesize 
Hard SVI prompting confusion matrix over the full frozen hard-prompt candidate pool of 1000 personas.
Each cell reports the integer count,
with the row percentage in parentheses. 
Cell shading encodes the row percentage.}
\label{fig:app-hard-svi-confusion}
\end{figure}

In more detail, \cref{fig:app-hard-svi-confusion} shows the asymmetry when hard prompting fails. 
Requests for the highest-vulnerability quartile are much more successful than requests for lower quartiles, whereas low- and middle-SVI requests are systematically shifted upward. Altogether, hard SVI prompting responds directionally to the requested quartile but exhibits a strong upward bias.

{\bf Sensitivity analysis.} %
We vary the output size \(m\in\{25,50,100\}\) and report results for county-seeded prompting.
For each output size, we vary the normalized KL tolerance \(d\in[0,3]\), the TV tolerance \(\varepsilon\in[0,0.999]\), and the CA-RDC capped-budget from \(n/m=1\) to \(20\).
Because county-seeded prompting has a nonzero probability of producing a persona outside the support of the target rates, we report feasible-run KL discrepancy for CA-RDC.

As shown in \cref{fig:svi-persona-sensitivity-analysis}, both anytime methods consistently reduce the KL and TV discrepancies as the normalized M-cost increases.
The cost-discrepancy profiles are stable across output sizes, with normalized KL discrepancies generally decreasing as \(m\) increases, while TV discrepancies tend to be larger for larger returned batches.

\begin{figure}
\centering
\includegraphics[width=\linewidth]{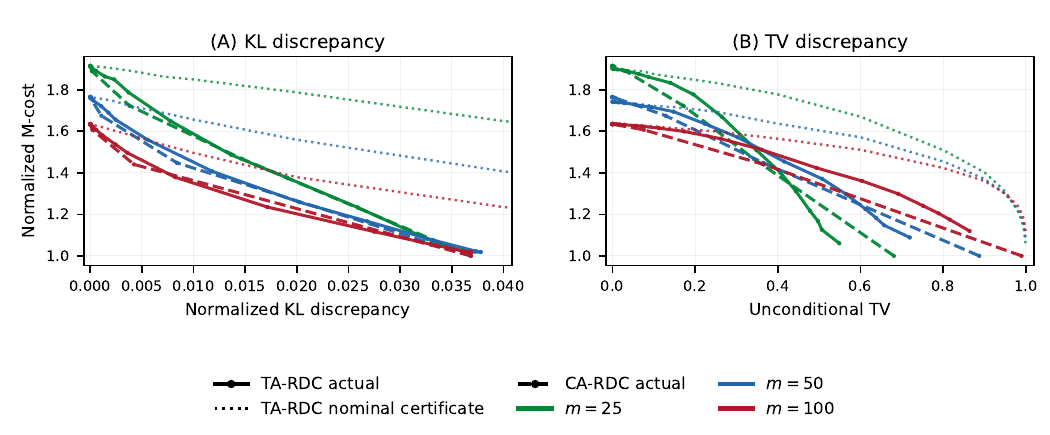}
\caption{\footnotesize
Output-size sensitivity analysis for county-seeded SVI persona generation.
The left panel reports unconditional KL discrepancy for TA-RDC and
feasible-run KL discrepancy for CA-RDC, while the right panel reports
unconditional TV discrepancy for both methods.
Colors distinguish \(m\in\{25,50,100\}\); solid and dashed curves report
the actual discrepancies of TA-RDC and CA-RDC, respectively, and dotted
curves report the nominal TA-RDC certificates.
}
\label{fig:svi-persona-sensitivity-analysis}
\end{figure}

\subsubsection{Implementation details and results for the synthetic pairwise audit}
\label{app:exp:details-results-pairwise-audit-ca-rdc-vs-pm}
\textbf{Methodology for post-processing.}
We use the demographic HPS individual-headshot pools for \texttt{Flux}, \texttt{Qwen}, and \texttt{HiDream} in the synthetic data task of
\cref{sec:exp-pairwise-audit}. 
For each model, we sample ten candidate batches of $n=200$ images without replacement from its frozen pool of 1000 images, and return $m=50$ images for each post-processing algorithm with target joint law $\nu=q^{\otimes m}$, where $q$ is uniform over the $k=16$ demographic attributes. 

Conditional on the source count process $\vec C(n)$, proportion matching (PM) returns an attribute sequence 
that minimizes the TV distance between the empirical distribution of its attributes and target marginal proportions; in our setting, for a target joint law, these target proportions 
are the target rates $q \in \Delta^{k-1}$.
Like RDC and its anytime variants, it then uniformly 
selects returned outputs selected within each attribute (cf.~\cref{alg:rdc,alg:anytime-rdc}).

{\bf False-match rate audit statistic.}
For $\mathrm{Alg}\in\{\textrm{CA-RDC}, \textrm{PM}\}$, write $\widehat L_{i,\mathrm{Alg}} = \sum_{s=1}^n 1\{\phi(\widetilde Y_s)=i\}$ for the number of outputs with annotated attribute $i$ selected by the post-processing algorithm. 
Let 
\[
    \widehat S_{2,\mathrm{Alg}}
    =\sum_{i=1}^k\binom{\widehat L_{i,\mathrm{Alg}}}{2},
    \qquad T_2=\binom{m}{2}=1225,
\]
for the numbers of same-attribute and total unordered pairs.

Within each model's frozen pool, we estimate the same- and different-demographic-attribute match rates, $\widehat\alpha_{\mathrm{same}}^{(2)}(t)$ and $\widehat\alpha_{\mathrm{diff}}^{(2)}(t)$, from all pairs with valid embeddings. Holding these plug-in rates fixed, we calculate
\begin{equation}
\label{eq:pairwise-synthetic-audit-rate}
    \widehat r_{2,\mathrm{Alg}}(t)
    =\frac{\widehat S_{2,\mathrm{Alg}}}{T_2}
        \widehat\alpha_{\mathrm{same}}^{(2)}(t)
    +\left(1-\frac{\widehat S_{2,\mathrm{Alg}}}{T_2}\right)
        \widehat\alpha_{\mathrm{diff}}^{(2)}(t).
\end{equation}

The oracle has output attribute law exactly equal to the target joint law $\nu = q^{\otimes m}$, and thus has
\(\E[\widehat S_{2,\mathrm{oracle}}]
    =
    \binom{m}{2}\sum_{i=1}^k q_i^2
    =
    T_2 \cdot 1/16\) same-attribute pairs.
We report its expected false-match rate using 
$\widehat S_{2,\mathrm{Alg}}/T_2 = 1/16$ in \cref{eq:pairwise-synthetic-audit-rate}.

\textbf{Results.}
\Cref{tab:pairwise-identity-audit-cross-model} in the main text 
reports the results at $t=0.45$. Same-attribute match rates are substantially larger than different-attribute rates for all three image models. PM suppresses same-attribute pairs relative to independent target sampling, thereby underestimating the oracle audit rate. 
In contrast, CA-RDC approximates the target joint law directly and better preserves pairwise statistics. 
For \texttt{HiDream}, CA-RDC selects an average of $72.6\;(2.03)$ same-attribute pairs per returned batch, compared with $54.0\;(0.00)$ for PM.
Its audit rate is $0.305$ percentage points away from that of the oracle, compared with $1.734$ percentage points for PM.

\Cref{fig:pairwise-identity-examples-all-models} provides examples of positive and negative pairwise match decisions, and
\Cref{fig:pairwise-identity-threshold-gap} shows results for a sensitivity analysis on the cosine threshold. 
The gap between the downstream false-match rate is largest when same-attribute pairs frequently exceed the threshold while different-attribute pairs rarely do. 
As the threshold becomes sufficiently low or high, nearly every pair is classified as a (non) false-match and thus the audited rates approach one another.

\begin{figure}[tbp]
\centering
\includegraphics[width=0.99\linewidth]{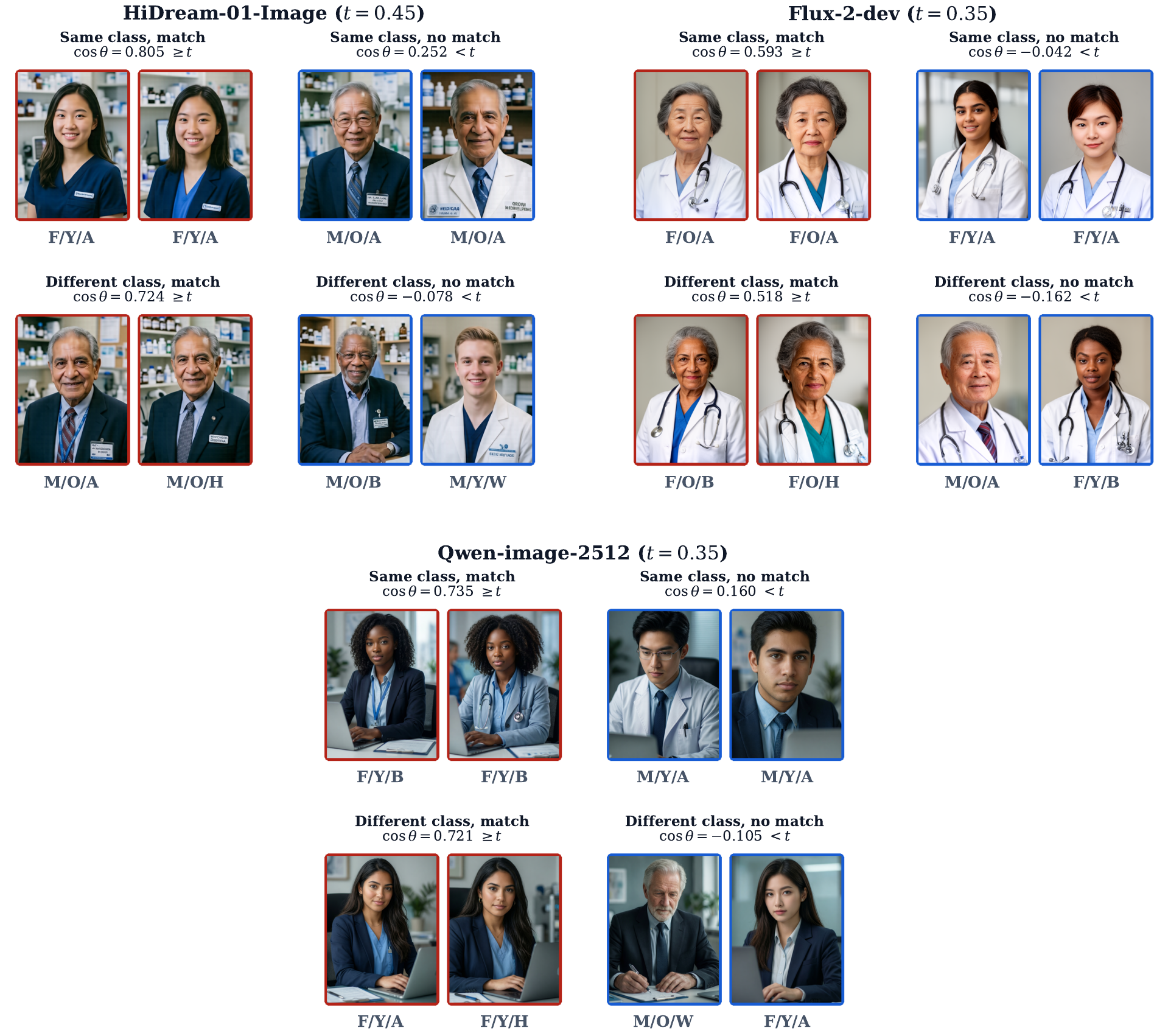}
\caption{\footnotesize
Example pairwise decisions on generated headshots, grouped by whether the annotated demographic attributes agree. The thresholds are $t=0.45$ for \texttt{HiDream} and $t=0.35$ for \texttt{Flux} and \texttt{Qwen}, as shown in the panels.}
\label{fig:pairwise-identity-examples-all-models}
\end{figure}

\begin{figure}[tbp]
\centering
\includegraphics[width=0.66\linewidth]{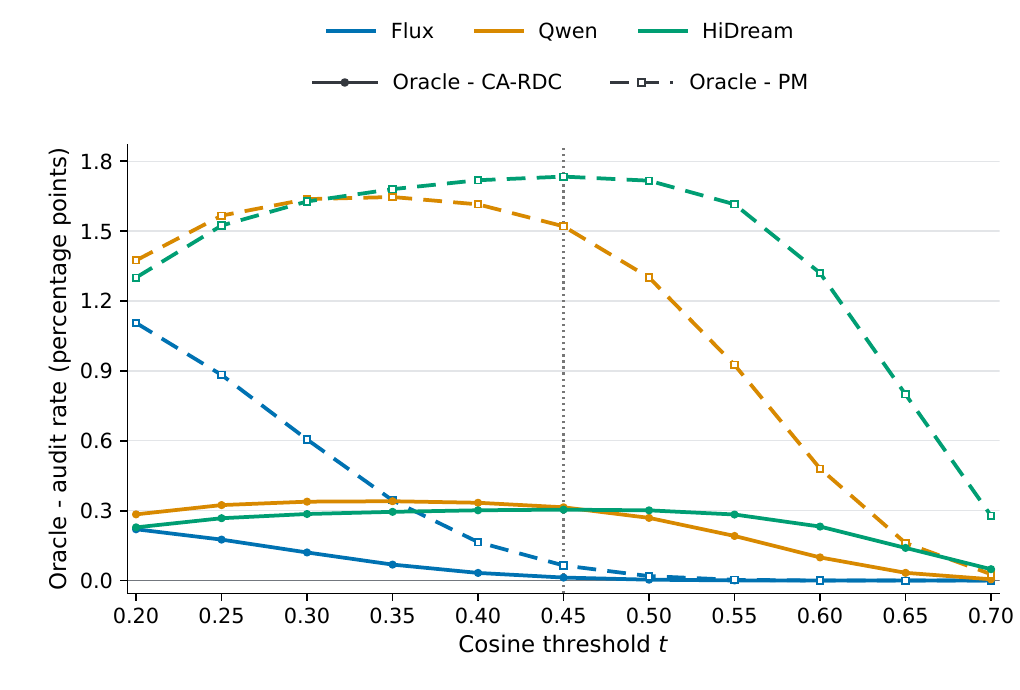}
\caption{\footnotesize
Threshold sensitivity on the cosine threshold for the pairwise audit task of \cref{sec:exp-pairwise-audit}.
The gap in the false match rate between the oracle and CA-RDC (solid) and proportion matching (dashed) are plotted across all three T2I models.
The vertical line denotes the cosine threshold $t=0.45$ reported in the main text.
The gap between the Oracle and PM procedure remains large through non-trivial intermediate ranges of the threshold. 
In contrast, CA-RDC better captures pairwise interactions in the same- and different-attribute class counts by targeting a joint law. This results in better estimates of the false-match rate. 
}
\label{fig:pairwise-identity-threshold-gap}
\end{figure}

\FloatBarrier
\subsection{Proofs}

\subsubsection{Proof of \texorpdfstring{\Cref{thm:rdc-exact-main}}{Theorem \getrefnumber{thm:rdc-exact-main}}}
\label{prf:rdc-exact-main}
\begin{proof}
Let $L_i=\sum_{j=1}^m 1\{A_j^\star=i\}$ be the attribute counts of the target attribute sequence $\vec A^\star \sim \nu$.
By definition of the feasible support (cf.~\cref{eq:feasible-support}),
\begin{equation}\label{eq:rdc-stopping-time-equivalence}
    T_\nu^\star
    =\inf\{t\ge m: \vec A^\star \in \Omega_t\}
    =\inf\{t\ge m:C_i(t)\ge L_i,\ \forall i\in[k]\}
    \eqqcolon T_{\vec L}.
\end{equation}
For each $i$ with $p_i = 0$,
the support condition $\operatorname{supp}(\nu)\subseteq(\operatorname{supp}(p))^m$
implies $L_i=0$ almost surely.
For each $i$ with $p_i>0$, we have $C_i(t)\to\infty$ almost surely.
Since there are finitely many attributes, all attribute demands are met in
finite time, so \(T_{\nu}^\star<\infty\) almost surely.
Moreover, 
at completion time \(T_\nu^\star\), RDC selects outputs whose attribute sequence is
\(\vec A^\star\) by construction. 
Therefore, its output attribute law is
\(\vec A^\star \sim \nu\), and RDC is universally exact.

\end{proof}

\subsubsection{Proof of \texorpdfstring{\Cref{thm:rdc-m1-optimal}}{Theorem \getrefnumber{thm:rdc-m1-optimal}}}
\label{prf:rdc-m1-optimal}
We first prove a lemma that generalizes the Wald identity to a conditional statement on the returned attribute sequence for universally exact algorithms.

\begin{lemma}[Conditional Wald identity]
\label{lem:conditional-wald-main}
Let a source-rate-free algorithm be universally exact for a target joint
law $\nu$ on $[k]^m$,
and let
\(\widetilde A\in[k]^m\) denote its returned attribute sequence.
If $\nu((\operatorname{supp}(p))^m)=1$ and
$\mathbb E_p[N]<\infty$, then, for every $i\in[k]$,
\begin{equation}\label{eq:conditional-wald-main}
    \mathbb E_p[C_i(N)\mid\widetilde A]
    =p_i\mathbb E_p[N\mid\widetilde A]
    \qquad\text{almost surely}.
\end{equation}
\end{lemma}

\begin{proof}
For any $p'$ with
$\operatorname{supp}(p')=\operatorname{supp}(p)$, 
the likelihood ratio of observing a 
sequence of 
$N$ M-samples is
$\prod_{i\in \operatorname{supp}(p)}(p_i'/p_i)^{C_i(N)}$.
For a fixed $\vec a$ with
$\nu(\vec a)>0$, 
applying a change of measure on each event
$\{N=n,\widetilde A=\vec a\}$ and summing over $n$
gives
\[
    \Pr_{p'}(\widetilde A=\vec a)
    =\mathbb E_p\left[
        1\{\widetilde A=\vec a\}
        \prod_{i\in \operatorname{supp}(p)}\left(\frac{p_i'}{p_i}\right)^{C_i(N)}
      \right].
\]
Universal exactness
implies $\Pr_{p'}(\widetilde A=\vec a)
=\Pr_p(\widetilde A=\vec a)=\nu(\vec a)>0$ and therefore
\begin{equation}\label{eq:conditional-stopped-lr}
    \mathbb E_p\left[
        \prod_{i\in \operatorname{supp}(p)}\left(\frac{p_i'}{p_i}\right)^{C_i(N)}
        \middle|\widetilde A=\vec a
      \right]=1.
\end{equation}
Jensen's inequality then yields
\[
    g_{\vec a}(p')
    \coloneq \sum_{i\in \operatorname{supp}(p)}
    \mathbb E_p[C_i(N)\mid\widetilde A=\vec a]
    \cdot \log\frac{p_i'}{p_i}\le0 = g_{\vec a}(p).
\]
Thus, \(p\) maximizes the differentiable function \(g_{\vec a}(p')\) 
subject
to \(\sum_{i\in \operatorname{supp}(p)} p_i'=1\). 
Stationarity at the optimizer \(p\) then
implies that there is a constant \(\lambda_{\vec a}\) such that
\[
    \frac{\mathbb E_p[C_i(N)\mid\widetilde A=\vec a]}{p_i}
    =\lambda_{\vec a},
    \qquad i\in \operatorname{supp}(p).
\]
Since \(\sum_{i\in \operatorname{supp}(p)} C_i(N)=N\) almost surely and \(\sum_{i \in \operatorname{supp}(p)} p_i = 1\), multiplying by $p_i$ and summing over
$i\in\operatorname{supp}(p)$ identifies this constant as
\(\lambda_{\vec a}=\mathbb E_p[N\mid\widetilde A=\vec a]\).
For $i\notin \operatorname{supp}(p)$, both sides of \eqref{eq:conditional-wald-main} are exactly zero.
\end{proof}

\begin{proof}[Proof of \Cref{thm:rdc-m1-optimal}]
For any $\mathrm{Alg}_1\in\mathfrak A_1^o(\nu)$ and any $i$ with
$\nu_i>0$, \cref{lem:conditional-wald-main} gives
\[
    p_i\mathbb E_p[N\mid\widetilde A_1=i]
    =\mathbb E_p[C_i(N)\mid\widetilde A_1=i]\ge1.
\]
Averaging over $\widetilde A_1\sim\nu$ shows that its M-cost is at least
$\sum_{i:\nu_i>0}\nu_i/p_i$. 
For RDC, conditional on the target draw $A_1^\star=i$,
the stopping time is geometric with success probability $p_i$ and mean
$1/p_i$. Again, averaging over $ A^\star_1\sim\nu$ then shows
its M-cost equals the lower bound, proving the result.
\end{proof}

\subsubsection{Proof of \texorpdfstring{\Cref{thm:rdc-witness-and-upper}}{Theorem \getrefnumber{thm:rdc-witness-and-upper}}}
\label{prf:rdc-witness-and-upper}
\begin{proof}
Let $\mathrm{Alg}_m\in\mathfrak A_m^o(\nu)$.
On the event the returned 
attribute sequence of $\mathrm{Alg}_m$
is
$\widetilde A=\vec a$, 
feasibility (cf.~\cref{eq:feasible-support}) implies
$C_i(N)\ge \sum_{j=1}^m1\{a_j=i\}$. 
Hence,
\Cref{lem:conditional-wald-main} and averaging over $\widetilde A\sim\nu$
gives
\begin{equation}\label{eq:rdc-demand-lower-bound}
    \mathbb E_p[N]
    \ge
    \mathbb E_\nu
    \left[\max_{i \in \operatorname{supp}(p)}\frac{L_i}{p_i}\right],
\end{equation}
where 
$L_i=\sum_{j=1}^m1\{\widetilde A_j=i\}$ are the attribute counts associated with an output attribute sequence $\widetilde A\sim\nu$.
Taking the
infimum gives a lower bound for
$\mathcal C_m^{\mathrm{univ}}(p,\nu)$.

For RDC, let $L_i=\sum_{j=1}^m 1\{A_j^\star=i\}$ be the attribute counts of the target attribute sequence $\vec A^\star \sim \nu$.
Define $T_i=\inf\{t\ge0:C_i(t)\ge L_i\}$ for $i$ with $p_i>0$,
so that $T_i=0$ when $L_i=0$ and
$T_\nu^\star=
\max_{i \in \operatorname{supp}(p)}T_i$ almost surely.
Conditional
on $\vec L$, $T_i$ is thus a negative binomial random variable with success rates $p_i$ and number of successes $L_i$.
Consequently,
\begin{equation}\label{eq:rdc-attribute-arrival-moments}
    \mathbb E_p[T_i\mid\vec L]=\frac{L_i}{p_i},
    \qquad
    \operatorname{Var}_p(T_i\mid\vec L)
    =\frac{L_i(1-p_i)}{p_i^2}.
\end{equation}
Furthermore,
\[
    T_{\nu}^\star
    \le\max_{i \in \operatorname{supp}(p)}\frac{L_i}{p_i}
       +\left(\sum_{i \in \operatorname{supp}(p)}
          \left(T_i-\frac{L_i}{p_i}\right)^2\right)^{1/2}
\]
by the deterministic inequality \(\max_i x_i
    \le \max_i y_i+\Bigl\{\sum_i(x_i-y_i)^2\Bigr\}^{1/2}\)
for real vectors with coordinates $x_i=T_i$ and $y_i=L_i/p_i$.
Jensen's inequality and \cref{eq:rdc-attribute-arrival-moments}
gives
\[
    \mathbb E_p[T_{\nu}^\star\mid\vec L]
    \le\max_{i \in \operatorname{supp}(p)}\frac{L_i}{p_i}
       +\left(\sum_{i \in \operatorname{supp}(p)}
          \frac{L_i(1-p_i)}{p_i^2}\right)^{1/2}.
\]
Averaging over $\vec L$ and applying Jensen's inequality again yields
\[
    \mathcal C_m^{\mathrm{RDC}}(p,\nu)
    \le
    \mathbb E_\nu
       \left[\max_{i \in \operatorname{supp}(p)}\frac{L_i}{p_i}\right]
    +\left(\sum_{i \in \operatorname{supp}(p)}
       \frac{\mathbb E_\nu[L_i](1-p_i)}{p_i^2}\right)^{1/2}.
\]
Combining this with \eqref{eq:rdc-demand-lower-bound} 
and noting $\E_\nu[L_i] = \sum_{j=1}^m \nu_j(i)$
gives the final result
\begin{equation}\label{eq:rdc-general-gap}
    \mathcal C_m^{\mathrm{RDC}}(p,\nu)
       -\mathcal C_m^{\mathrm{univ}}(p,\nu)
    \le\left(
           \sum_{i\in \operatorname{supp}(p)} \left( \sum_{j=1}^m \nu_j(i) \right) \cdot  \frac{1-p_i}{p_i^2}
       \right)^{1/2}.
\end{equation}
\end{proof}

\subsubsection{Proof of \texorpdfstring{\Cref{thm:universal-exact-m-cost-lower}}{Theorem \getrefnumber{thm:universal-exact-m-cost-lower}}}
\label{prf:universal-exact-m-cost-lower}
\begin{proof}
Let $\textrm{Alg}_m\in\mathfrak A_m(\nu)$. Denote its stopping time as $N$ and its returned attribute sequence \(\widetilde A \sim \nu\).
The process $C_i(t)-p_i t$ is a martingale
with bounded increments
for
the source count filtration \(
\mathcal G_t\coloneqq\sigma\{\vec C(s):0\le s\le t\}.
\)
Doob's optional stopping theorem applied at the bounded stopping
time $N\wedge n$ gives
$\mathbb E_p[C_i(N\wedge n)]=p_i\mathbb E_p[N\wedge n]$.
Letting $n\to\infty$ and applying the monotone convergence theorem
then yields
$\mathbb E_p[C_i(N)]=p_i\mathbb E_p[N]$.

Universal exactness 
requires
$\Pr_p(\widetilde A_j=i)=\nu_j(i)$ for each $j\in[m]$, so that 
feasibility (cf.~\cref{eq:feasible-support}) implies
\[
    \sum_{j=1}^m\nu_j(i)
    =\mathbb E_p\left[\sum_{j=1}^m1\{\widetilde A_j=i\}\right]
    \le\mathbb E_p[C_i(N)]
    =p_i\mathbb E_p[N].
\]
Dividing by $p_i$ and maximizing over $i \in \operatorname{supp}(p)$ 
gives the final result
\begin{equation}\label{eq:general-nu-marginal-cost-lower}
    \mathbb E_p[N]
    \ge\max_{i \in \operatorname{supp}(p)}
       \frac{\sum_{j=1}^m\nu_j(i)}{p_i}.
\end{equation}
\end{proof}

\subsubsection{Minimal sufficiency of source counts}
\label{prf:count-minimal-sufficient}

\begin{proposition}[Minimal sufficiency of source counts]
\label{prop:count-minimal-sufficient}
Fix a deterministic sample size $n\ge m$. Under
$A_1,\ldots,A_n\overset{\mathrm{i.i.d.}}{\sim}\mathrm{Categorical}(p)$,
the count vector $\vec C(n)$ is minimal sufficient for
$p\in\operatorname{int}(\Delta^{k-1})$.
\end{proposition}

\begin{proof}
An ordered attribute sample with counts $\vec c$ has likelihood
$\prod_i p_i^{c_i}$, so the factorization theorem immediately 
gives sufficiency.
For two samples with counts $\vec c$ and $\vec c'$, their likelihood
ratio is $\prod_i p_i^{c_i-c_i'}$. This ratio is constant in $p$
if and only if $\vec c=\vec c'$, since we may otherwise let a coordinate
with nonzero exponent tend to zero while fixing the remaining coordinates
to be positive.
This shows that the ratio cannot be constant.
The likelihood-ratio characterization of minimal sufficiency
therefore yields the claim.
\end{proof}

\subsubsection{Proof of \texorpdfstring{\Cref{thm:opt-count-cond-f-curve}}{Theorem \getrefnumber{thm:opt-count-cond-f-curve}}}
\label{prf:optimal-f-curve}
\begin{proof}
Fix $t\ge m$ and condition on the feasible support $\Omega_t \equiv \Omega_m(\vec c(t))$.
On the event $\{T_\nu^\star>t\}$,
A-RDC
selects the newly feasible support
$\mathcal D_s \in \{\mathcal{D}_m,\ldots, \mathcal{D}_t\}$ with probability
$\nu(\mathcal D_s)/\alpha_t$ by the memoryless property of exponential random variables, whenever $\alpha_t>0$.
Since
$\{T_\nu^\star\le t\}=\{\vec A^\star\in\Omega_t\}$ and $\vec A^\star \sim \nu$, we have
$\Pr(T_\nu^\star \le t \mid \Omega_t) = \nu(\Omega_t)$.
Then, 
for $0<\alpha_t\le 1$
and $\vec a\in\Omega_t$, the output attribute law of A-RDC conditional on the feasible support is
\begin{equation}\label{eq:ardc-history-conditional-law}
    \Pr_p(\widehat A_t=\vec a\mid\Omega_t)
    =\nu(\vec a)+(1-\alpha_t)\frac{\nu(\vec a)}{\alpha_t}
    =\frac{\nu(\vec a)}{\alpha_t}=\nu(\cdot\mid\Omega_t).
\end{equation}
Since the feasible support depends on the source-count process only through the count vector $\vec C$, this shows \cref{eq:a-rdc-count-cond-output-label-law}.

For a realized count vector $\vec C(t)=\vec c$, feasibility implies the count-conditional output attribute law $\mu$ of any algorithm $\mathrm{Alg}_m\in\mathfrak{A}_m$ is supported
on $\Omega_m(\vec c)$.
Applying the data processing inequality 
to the indicator function $\vec a \mapsto 1\{\vec a \in \Omega_m(\vec c)\}$ then yields
\begin{equation}\label{eq:statewise-support-projection-bound}
    D_f(\mu\|\nu)
    \ge
    D_f\bigl((1,0)\|(\alpha_t,1-\alpha_t)\bigr)
    =\Psi_f(\alpha_t).
\end{equation}
When $\alpha_t>0$, the law $\mu = \nu(\cdot\mid\Omega_t)$ attains the lower
bound.
When $\alpha_t=0$, both sides of the inequality are
$f_\infty+f(0)=\Psi_f(0)$, which may be infinite.
Averaging \eqref{eq:statewise-support-projection-bound} over the randomness of the source-count process $\vec C(t)$ and taking the infimum over $\mu$
proves the lower bound of \cref{eq:count-conditional-optimal-lb}.

Finally, \eqref{eq:ardc-history-conditional-law}
shows that the count-conditional output attribute law of A-RDC is $\nu(\cdot \mid \Omega_t)$ and thus
attains this lower bound
simultaneously for each $t \ge m$.
\end{proof}

\subsubsection{Proof of \texorpdfstring{\Cref{thm:thresholded-anytime-rdc}}{Theorem \getrefnumber{thm:thresholded-anytime-rdc}}}
\label{prf:thresholded-anytime-rdc}
\begin{proof}
Since $p\in\operatorname{int}(\Delta^{k-1})$, there is almost surely
a finite time $t_0$ such that $C_i(t_0)\ge m$ for every $i\in[k]$.
At this time, $\Omega_{t_0}=[k]^m$ and $\alpha_{t_0}=1$,
so both $T_\nu^\star$ and $\tau_{f,\varepsilon}$ are bounded
above by $t_0$ and hence almost surely finite, i.e. terminates almost surely.

Observe that TA-RDC and A-RDC agree pathwise:
$\widehat A_{T_\nu^\star\wedge\tau_{f,\varepsilon}}
=\widehat A_{\tau_{f,\varepsilon}}$,
since both equal $\vec A^\star$ when
$T_\nu^\star\le\tau_{f,\varepsilon}$.
Then,
by \cref{eq:a-rdc-count-cond-output-label-law} in \cref{thm:opt-count-cond-f-curve},
TA-RDC's
output attribute law conditional on the feasible support $\Omega_{\tau_{f,\varepsilon}}$
is 
$\nu(\cdot\mid\Omega_{\tau_{f,\varepsilon}})$, and thus has divergence $\Psi_f(\alpha_{\tau_{f,\varepsilon}})$
from $\nu$. 
By convexity of $f$-divergences and construction of the certification stopping time, we then have
\[
    \Delta_f^{\mathrm{un}}(\mathrm{TA\text{-}RDC};\nu;p)
    \le \mathbb E_p[\Psi_f(\alpha_{\tau_{f,\varepsilon}})]
    \le \varepsilon.
\]
\end{proof}

\subsubsection{Proof of \texorpdfstring{\Cref{thm:first-order-uni-opt-ta-rdc}}{Theorem \getrefnumber{thm:first-order-uni-opt-ta-rdc}}}
\label{prf:first-order-uni-opt-ta-rdc}

\begin{lemma}[Bounds for a capped KL projection]
\label{lem:capped-kl-projection-bounds}
Fix $p,q\in\operatorname{int}(\Delta^{k-1})$, and write
$p_{\min}=\min_i p_i$ and $\kappa=\max_i q_i/p_i$.
For $1<R<\kappa$, let $s^\star$ be the unique minimizer of
the capped KL projection problem
\[
    \min_{\substack{s\in\Delta^{k-1}\\
                    s_i\le Rp_i,\ i\in[k]}}
    D_{\mathrm{KL}}(s\|q).
\]
Its coordinates have the form
\begin{equation}\label{eq:ta-rdc-kl-projection-form}
    s_i^\star=\min\{\lambda q_i,Rp_i\},
    \qquad
    \sum_i s_i^\star=1,
    \qquad
    \lambda>1.
\end{equation}
Define the weighted sum of the KKT multipliers
\[
    g\coloneqq
    \sum_i p_i\log\frac{\lambda q_i}{s_i^\star}.
\]
Then, for any $1\le\rho<R$ and $r\in\Delta^{k-1}$
satisfying $r_i\le\rho p_i$ for every $i$,
\begin{equation}\label{eq:ta-rdc-projection-convexity}
    D_{\mathrm{KL}}(s^\star\|q)+(R-\rho)g
    \le D_{\mathrm{KL}}(r\|q).
\end{equation}
Moreover,
\begin{equation}\label{eq:ta-rdc-projection-varentropy}
    g\ge4p_{\min}^2(\kappa-R),
    \qquad\text{and}\qquad
    \operatorname{Var}_{A\sim s^\star}
    \left(\log\frac{s_A^\star}{q_A}\right)
    \le\frac{g^2}{4p_{\min}^2}.
\end{equation}
\end{lemma}

\begin{proof}
A minimizer exists by compactness and continuity, and is unique
by strict convexity of the KL objective.
The KKT conditions \citep[see, e.g.,][]{boyd2004convex}
give the form in \eqref{eq:ta-rdc-kl-projection-form},
with $\lambda>1$ because $q$ is infeasible for $1< R <\kappa$.

Write
$\eta_i=\log(\lambda q_i/s_i^\star)\ge0$,
so that $g=\sum_i p_i\eta_i$.
For the comparison bound, convexity gives
\begin{align*}
    D_{\mathrm{KL}}(r\|q)
    &\ge D_{\mathrm{KL}}(s^\star\|q)
       +\sum_i\log\frac{s_i^\star}{q_i}(r_i-s_i^\star)\\
    &=D_{\mathrm{KL}}(s^\star\|q)
       +\sum_i\eta_i(s_i^\star-r_i)\\
    &=D_{\mathrm{KL}}(s^\star\|q)
       +\sum_i\eta_i(Rp_i-r_i)\\
    &\ge D_{\mathrm{KL}}(s^\star\|q)+(R-\rho)g,
\end{align*}
where the first equality uses $\sum_i(r_i-s_i^\star)=0$;
the second equality uses
$\eta_i(s_i^\star-Rp_i)=0$ for every $i$;
and the final inequality uses
$r_i\le\rho p_i$ and $\eta_i\ge0$.
This proves \eqref{eq:ta-rdc-projection-convexity}.

Choose $j$ with $q_j/p_j=\kappa$.
Since $Rp_j<q_j<\lambda q_j$, we have $s_j^\star=Rp_j$.
Moreover,
\[
    1-Rp_j=\sum_{i\ne j}s_i^\star
    \le\lambda(1-q_j).
\]
Thus, using $0<Rp_j<q_j<1$,
\begin{align*}
    \eta_j
    &\ge
    \log\frac{q_j(1-Rp_j)}{Rp_j(1-q_j)}\\
    &=\int_{Rp_j}^{q_j}\frac{du}{u(1-u)}\\
    &\ge4(q_j-Rp_j)
     =4p_j(\kappa-R),
\end{align*}
where the last inequality follows from $u(1-u)\le1/4$.
Consequently,
\[
    g\ge p_j\eta_j
    \ge4p_j^2(\kappa-R)
    \ge4p_{\min}^2(\kappa-R),
\]
proving the first assertion of
\eqref{eq:ta-rdc-projection-varentropy}.

For the variance bound, nonnegativity of the $\eta_i$ gives
\[
    0\le\eta_i\le\frac{g}{p_i}\le\frac{g}{p_{\min}}
    \qquad\text{for every }i.
\]
Since
\[
    \log\frac{s_A^\star}{q_A}=\log\lambda-\eta_A,
\]
this random variable lies in an interval of length at most
$g/p_{\min}$. Its variance is therefore at most
$g^2/(4p_{\min}^2)$, proving the second assertion of
\eqref{eq:ta-rdc-projection-varentropy}.
\end{proof}

\begin{proof}[Proof of \texorpdfstring{\Cref{thm:first-order-uni-opt-ta-rdc}}]
\emph{A lower bound on the universal source-rate-free M-cost.}
Fix $m\ge1$ and $\varepsilon\ge0$.
Let $\textrm{Alg}_m\in\mathfrak{A}_m^o$ be universally
$(\mathrm{KL},\varepsilon)$-approximate, and let
$P_{\textrm{Alg}_m}(\widetilde A)$ be its output attribute law
(cf.~\cref{eq:output-label-law}).
Write $N$ for its stopping time.
We may assume $\mathbb E_p[N]<\infty$, since otherwise
the lower bound is immediate.
Define its average marginal attribute probabilities as
\[
    c_i=\frac1m\sum_{j=1}^m
    P_{\textrm{Alg}_m}(\widetilde A_j=i),
    \qquad i\in[k].
\]
Further denote
$p_{\min}=\min_i p_i$,
$\rho=\rho_{p,q}^\star(\varepsilon/m)$, and
$\kappa=\max_i q_i/p_i$.
By definition, $1\le\rho\le\kappa$.
The KL chain rule and convexity of KL divergence imply
\begin{align}\label{eq:marginal-attribute-countkl-divergence}
    \varepsilon
    &\ge D_{\mathrm{KL}}\left(
       P_{\textrm{Alg}_m}(\widetilde A)
       \middle\|q^{\otimes m}\right) \notag\\
    &\ge\sum_{j=1}^m
       D_{\mathrm{KL}}\left(
         P_{\textrm{Alg}_m}(\widetilde A_j)
         \middle\|q\right)
    \ge mD_{\mathrm{KL}}(c\|q).
\end{align}
Feasibility and Wald's identity then give
\[
    mc_i
    =\mathbb E_p\left[
       \sum_{j=1}^m1\{\widetilde A_j=i\}\right]
    \le p_i\mathbb E_p[N],
    \qquad i\in[k];
\]
see, also, the proof of
\Cref{thm:universal-exact-m-cost-lower}.
Hence, by \cref{eq:marginal-attribute-countkl-divergence},
$\mathbb E_p[N]\ge m\max_i(c_i/p_i)
\ge m\rho_{p,q}^\star(\varepsilon/m)$,
and taking the infimum gives
\begin{equation}\label{eq:ta-rdc-universal-marginal-kl-converse}
    \mathcal C_{m,\varepsilon}^{\mathrm{univ}}(p,q)
    \ge m\rho_{p,q}^\star(\varepsilon/m).
\end{equation}

\emph{Completion time bounds for multinomial demands.}
For an arbitrary $s\in\Delta^{k-1}$, let
$\vec L\sim\mathrm{Mult}(m,s)$ be an auxiliary demand
sampled independently of the source-count process $\vec C(t)$
(cf.~\cref{eq:feasible-support}).
Define the attribute completion times
$T_i=\inf\{t\ge0:C_i(t)\ge L_i\}$
(as in the proof of \cref{thm:rdc-witness-and-upper})
and the completion time $T_s\coloneqq\max_i T_i$.
Then, from \eqref{eq:rdc-attribute-arrival-moments}
and the law of total variance,
\[
    \mathbb E_{p,s}[T_i]=\frac{ms_i}{p_i},
    \qquad
    \operatorname{Var}_{p,s}(T_i)
    =\frac{ms_i(2-p_i-s_i)}{p_i^2}.
\]
Furthermore, since $T_s=\max_i T_i$,
\[
    \left(T_s-m\max_i\frac{s_i}{p_i}\right)_+^2
    \le\sum_i\left(T_i-\frac{ms_i}{p_i}\right)^2,
\]
and taking expectations gives
\begin{equation}\label{eq:ta-rdc-auxiliary-completion-second-moment}
    \mathbb E_{p,s}\left[
       \left(T_s-m\max_i\frac{s_i}{p_i}\right)_+^2
    \right]
    \le m\sum_i\frac{s_i(2-p_i-s_i)}{p_i^2}
    \le\frac{2m}{p_{\min}^2}.
\end{equation}
In particular, the Cauchy--Schwarz inequality gives
\begin{align}\label{eq:ta-rdc-auxiliary-completion-mean}
    \mathbb E_{p,s}[T_s]
    &\le m\max_i\frac{s_i}{p_i}
       +\mathbb E_{p,s}\left[
          \left(T_s-m\max_i\frac{s_i}{p_i}\right)_+
        \right] \notag\\
    &\le m\max_i\frac{s_i}{p_i}
       +\sqrt{\mathbb E_{p,s}\left[
          \left(T_s-m\max_i\frac{s_i}{p_i}\right)_+^2
        \right]} \notag\\
    &\le m\max_i\frac{s_i}{p_i}
       +\frac{\sqrt{2m}}{p_{\min}}.
\end{align}

Define the source-count stopping time
$\sigma_s=\inf\{t\ge m:s^{\otimes m}(\Omega_t)\ge1/2\}$.
Conditional on the entire source-count process
$\mathcal G_\infty\coloneqq\sigma\{\vec C(u):u\ge0\}$,
the completion time has distribution function
\[
    \Pr_{p,s}(T_s\le t\mid\mathcal G_\infty)
    =s^{\otimes m}(\Omega_t),
    \qquad t\ge m.
\]
Then, by definition of $\sigma_s$ and since $T_s\ge m$,
\[
    \Pr_{p,s}(T_s\ge\sigma_s\mid\mathcal G_\infty)
    \ge\frac12.
\]
It follows that
\[
    \frac12
    \left(\sigma_s-m\max_i\frac{s_i}{p_i}\right)_+^2
    \le
    \mathbb E_{p,s}\left[
       \left(T_s-m\max_i\frac{s_i}{p_i}\right)_+^2
       \middle|\mathcal G_\infty
    \right].
\]
Taking expectations, using Cauchy--Schwarz, and using
\eqref{eq:ta-rdc-auxiliary-completion-second-moment} yields
\begin{equation}\label{eq:ta-rdc-auxiliary-median-time}
    \mathbb E_p[\sigma_s]
    \le m\max_i\frac{s_i}{p_i}
       +\frac{2\sqrt m}{p_{\min}}.
\end{equation}

\emph{The upper bound for TA-RDC.}
First suppose
$\kappa-\rho\le3/(2p_{\min}\sqrt m)$.
TA-RDC stops no later than RDC, so applying
\eqref{eq:ta-rdc-auxiliary-completion-mean} with $s=q$ gives
\begin{equation}\label{eq:ta-rdc-near-exact-branch}
\begin{split}
    \mathcal C_{m,\varepsilon}^{\mathrm{TA\text{-}RDC}}(p,q)
    &\le m\kappa+\frac{\sqrt{2m}}{p_{\min}}\\
    &\le m\rho+
       \left(\frac32+\sqrt2\right)\frac{\sqrt m}{p_{\min}}
     <m\rho+\frac{3\sqrt m}{p_{\min}}.
\end{split}
\end{equation}

Suppose now that
$\kappa-\rho>3/(2p_{\min}\sqrt m)$, and set
\[
    R=\rho+\frac1{p_{\min}\sqrt m},
    \qquad 1<R<\kappa.
\]
Let $s=s^\star$ and $g$ be as in
\cref{lem:capped-kl-projection-bounds}.
An optimizer $r^\star$ in the definition of
$\rho=\rho_{p,q}^\star(\varepsilon/m)$ satisfies
$r_i^\star\le\rho p_i$ and
$D_{\mathrm{KL}}(r^\star\|q)\le\varepsilon/m$.
Thus, \eqref{eq:ta-rdc-projection-convexity} gives
\begin{equation}\label{eq:ta-rdc-projection-kl-slack}
    mD_{\mathrm{KL}}(s\|q)
    \le\varepsilon-m(R-\rho)g
    =\varepsilon-\frac{\sqrt m}{p_{\min}}g.
\end{equation}

Fix a source state at time $t$ such that
$\beta=s^{\otimes m}(\Omega_t)\ge1/2$.
For a sequence drawn from $s^{\otimes m}$, let
$Z=\log(ds^{\otimes m}/d\nu)$, where $\nu=q^{\otimes m}$.
Both product laws are positive on $[k]^m$.
Conditional Jensen gives
\[
    \alpha_t
    =\beta\,\mathbb E_{s^{\otimes m}}[e^{-Z}\mid\Omega_t]
    \ge\beta\exp\{
       -\mathbb E_{s^{\otimes m}}[Z\mid\Omega_t]\}.
\]
By Cauchy--Schwarz,
\begin{align*}
    \mathbb E_{s^{\otimes m}}[Z\mid\Omega_t]
       -\mathbb E_{s^{\otimes m}}[Z]
    &=\frac{
       \operatorname{Cov}_{s^{\otimes m}}(Z,1_{\Omega_t})
       }{\beta}\\
    &\le\sqrt{\operatorname{Var}_{s^{\otimes m}}(Z)}
          \sqrt{\frac{1-\beta}{\beta}}
     \le\sqrt{\operatorname{Var}_{s^{\otimes m}}(Z)}.
\end{align*}
Since
$\mathbb E_{s^{\otimes m}}[Z]=mD_{\mathrm{KL}}(s\|q)$
and
$\operatorname{Var}_{s^{\otimes m}}(Z)
=m\operatorname{Var}_{A\sim s}(\log(s_A/q_A))$,
we conclude that
\begin{align*}
    -\log\alpha_t
    &\le mD_{\mathrm{KL}}(s\|q)
       +\sqrt{m\operatorname{Var}_{A\sim s}
          \left(\log\frac{s_A}{q_A}\right)}
       +\log2\\
    &\le\varepsilon-\frac{\sqrt m}{2p_{\min}}g+\log2,
\end{align*}
where the last line uses
\eqref{eq:ta-rdc-projection-varentropy}
in \cref{lem:capped-kl-projection-bounds}
and \eqref{eq:ta-rdc-projection-kl-slack}.
Our case assumption gives
$\kappa-R>1/(2p_{\min}\sqrt m)$, so
\eqref{eq:ta-rdc-projection-varentropy} implies
\[
    \frac{\sqrt m}{2p_{\min}}g>1>\log2.
\]
Thus every state with $s^{\otimes m}(\Omega_t)\ge1/2$
satisfies the TA-RDC certificate. In particular,
\[
    N_{\mathrm{KL},\varepsilon}^{\mathrm{TA\text{-}RDC}}
    \le\tau_{\mathrm{KL},\varepsilon}\le\sigma_s.
\]
Because $R<\kappa$, at least one coordinate of $s$ attains
its cap, so $\max_i(s_i/p_i)=R$.
Applying \eqref{eq:ta-rdc-auxiliary-median-time} therefore gives
\begin{equation}\label{eq:ta-rdc-projection-branch}
    \mathcal C_{m,\varepsilon}^{\mathrm{TA\text{-}RDC}}(p,q)
    \le mR+\frac{2\sqrt m}{p_{\min}}
    =m\rho+\frac{3\sqrt m}{p_{\min}}.
\end{equation}
Together, \eqref{eq:ta-rdc-universal-marginal-kl-converse}, \eqref{eq:ta-rdc-near-exact-branch}, and
\eqref{eq:ta-rdc-projection-branch} yield,
for $\varepsilon_m/m\to\delta$,
\[
\begin{split}
    \rho_{p,q}^\star(\varepsilon_m/m)
    &\le
    \frac{\mathcal C_{m,\varepsilon_m}^{\mathrm{univ}}(p,q)}{m}
    \le
    \frac{
       \mathcal C_{m,\varepsilon_m}^{\mathrm{TA\text{-}RDC}}(p,q)
    }{m}\\
    &\le\rho_{p,q}^\star(\varepsilon_m/m)
       +\frac{3}{p_{\min}\sqrt m} \\
    & \le \frac{
       \mathcal C_{m,\varepsilon_m}^{\mathrm{univ}}(p,q)
    }{m}+\frac{3}{p_{\min}\sqrt m}.
\end{split}
\]
The last inequality
is the finite-sample bound on the M-cost of TA-RDC in \cref{eq:ta-rdc-finite-m-additive-bound}.
The sandwich bound of the first and second inequality show the first-order asymptotic limits in \cref{eq:first-order-opt-ta-rdc}.
Finally, 
since $ \mathcal C_{m,\varepsilon}^{\mathrm{univ}}(p,q) \ge m$ for any post-processing algorithm $\textrm{Alg}_m\in \mathfrak{A}_{m}$,
\[
    1\le
    \frac{
       \mathcal C_{m,\varepsilon}^{\mathrm{TA\text{-}RDC}}(p,q)
    }{
       \mathcal C_{m,\varepsilon}^{\mathrm{univ}}(p,q)
    }
    \le1+\frac{3}{p_{\min}\sqrt m}
\]
holds uniformly in $\varepsilon\ge0$.
Taking the supremum over $\varepsilon$ proves the uniform relative
optimality statement of \cref{eq:ta-rdc-uniform-first-order-optimality}.
\end{proof}

\subsubsection{Proof of \texorpdfstring{\cref{thm:no-capped-budget-uni-exactness}}{Theorem \getrefnumber{thm:no-capped-budget-uni-exactness}}}\label{prf:no-capped-budget-uni-exactness}
\begin{proof}
    Denote the set in the statement by $\mathcal E_n$. Fix an interior source rate $p_0 \in \operatorname{int}(\Delta^{k-1})$. The assumed bound gives $\Pr_{p_0}(N\leq n)=1$. 
    Collect all E randomization in a random seed $U$ whose law does not depend on $p$. In the fixed attribute conditional mixture family, the likelihood ratio of the first $n$ M-samples and $U$ under an interior $p$ relative to $p_0$ is
\[
\frac{d\Pr_p}{d\Pr_{p_0}}(Y_{1:n},U)
=
\prod_{t=1}^n\frac{p_{A_t}}{p_{0,A_t}}
=
\prod_{i=1}^k\left(\frac{p_i}{p_{0,i}}\right)^{C_i(n)}.
\]
This likelihood ratio is strictly positive and finite. Since $N$ is a stopping time, $\{N>n\}$ is determined by the first $n$ M-samples and the E randomization used by then. 
It follows that $\Pr_p(N\leq n)=1$ for every interior $p$. Since $N\geq m$, necessarily $n\geq m$. 
Continue drawing unused M-samples after the algorithm stops until time $n$, without changing its output, and let $\vec C \equiv \vec C(n)$ be the count vector of these $n$ samples. 
For $\vec w\in[k]^n$ and $\vec a\in[k]^m$, put
\[
R_{\vec w}(\vec a)
=
\Pr(\widetilde{A}=\vec a\mid A_{1:n}=\vec w).
\]
This probability integrates over the fixed attribute conditional output laws and the algorithm's E randomization, so it does not depend on $p$. Conditional on $\vec C=\vec c$, all attribute sequences with count vector $\vec c$ are equally likely. Therefore
\[
Q_{\vec c}(\vec a)
=
\Pr_p(\widetilde{A}=\vec a\mid\vec C=\vec c)
=
\frac{1}{M_{\vec c}}
\sum_{\substack{\vec w\in[k]^n\\ C(\vec w)=\vec c}}
R_{\vec w}(\vec a),
\qquad
M_{\vec c}=\frac{n!}{\prod_{i=1}^k c_i!},
\]
and $Q_{\vec c}$ does not depend on $p$. 

Write $p^{\vec c}=\prod_i p_i^{c_i}$.
For every $\vec a\in[k]^m$,
\[
F_{\vec a}(p)
=
\Pr_p(\widetilde{A}=\vec a)-\nu(\vec a)
=
\sum_{|\vec c|=n}
M_{\vec c}p^{\vec c}
\left(Q_{\vec c}(\vec a)-\nu(\vec a)\right),
\]
where the last equality uses $\sum_{|\vec c|=n}M_{\vec c}p^{\vec c}=(\sum_i p_i)^n=1$. 
After substituting $p_k=1-\sum_{i=1}^{k-1}p_i$, each $F_{\vec a}$ is a polynomial in $p_1,\ldots,p_{k-1}$. Suppose $\mathcal E_n$ had positive $(k-1)$ dimensional Lebesgue measure. Each $F_{\vec a}$ would then vanish on a set of positive measure. Since the zero set of a nonzero real polynomial has Lebesgue measure zero, every $F_{\vec a}$ must vanish identically on the open simplex. Hence
\[
\mathbb E_p\left[
Q_{\vec C}(\vec a)-\nu(\vec a)
\right]
=0
\]
for every interior $p$ and every $\vec a\in[k]^m$. The family $\vec C\sim\Mult(n,p)$ is complete, so
\[
Q_{\vec c}(\vec a)=\nu(\vec a)
\]
for every $\vec c$ with $|\vec c|=n$ and every $\vec a\in[k]^m$. 
Here, every such $\vec c$ has positive probability under an interior $p$. But when $\vec c=n\vec e_i$, every available attribute is $i$, and feasibility forces $Q_{n\vec e_i}=\delta_{(i,\ldots,i)}$. 
Taking two different attributes would make the fixed law $\nu$ equal to two different point masses. This is impossible because $k\geq2$. 
Therefore, $\mathcal E_n$ has measure zero. 
\end{proof}

\end{document}